%% file: main.tex
\pdfoutput=1

\documentclass[11pt]{article}

\usepackage[final]{acl}

\usepackage{amsmath}
\usepackage{times}
\usepackage{amssymb}
\usepackage{latexsym}
\usepackage{multirow}
\usepackage{enumitem}
\usepackage{makecell}
\usepackage{mdframed} 
\usepackage{ragged2e} 
\usepackage{algorithm}
\usepackage{algpseudocode}

\usepackage[T1]{fontenc}

\usepackage[utf8]{inputenc}

\usepackage{microtype}

\usepackage{inconsolata}

\usepackage{graphicx}
\usepackage{booktabs} 

\usepackage{amsthm}
\usepackage[english]{babel}
\newtheorem{theorem}{Theorem}
\newtheorem{assumption}{Assumption}
\newtheorem{lemma}{Lemma}

\newcommand{\shrink}{\vspace*{-.9\baselineskip}}

\newcommand{\axiomOne}{Positively Consistent}
\newcommand{\axiomTwo}{Negatively Consistent}
\newcommand{\axiomThree}{Positively Changed}
\newcommand{\axiomFour}{Negatively Changed}
\newcommand{\axiomFive}{Neutrally Consistent}

\newcommand{\llminput}{x} 
\newcommand{\doc}{c}
\newcommand{\query}{q}
\newcommand{\response}{r}

\title{Why Uncertainty Estimation Methods Fall Short in RAG:\\ An Axiomatic Analysis}

\author{Heydar Soudani \\
  Radboud University \\
  The Netherlands \\
  \texttt{heydar.soudani@ru.nl} \\\And
  Evangelos Kanoulas \\
  University of Amsterdam \\
  The Netherlands \\
  \texttt{e.kanoulas@uva.nl} \\\And
  Faegheh Hasibi \\
  Radboud University \\
  The Netherlands \\
  \texttt{faegheh.hasibi@ru.nl} \\}

\begin{document}
\maketitle
\begin{abstract}
Large Language Models (LLMs) are valued for their strong performance across various tasks, but they also produce inaccurate or misleading outputs. Uncertainty Estimation (UE) quantifies the model's confidence and helps users assess response reliability. However, existing UE methods have not been thoroughly examined in scenarios like Retrieval-Augmented Generation (RAG), where the input prompt includes non-parametric knowledge. This paper shows that current UE methods cannot reliably estimate the correctness of LLM responses in the RAG setting. We propose an axiomatic framework to identify deficiencies in existing UE methods. Our framework introduces five constraints that an effective UE method should meet after incorporating retrieved documents into the LLM's prompt. Experimental results reveal that no existing UE method fully satisfies all the axioms, explaining their suboptimal performance in RAG. We further introduce a simple yet effective calibration function based on our framework, which not only satisfies more axioms than baseline methods but also improves the correlation between uncertainty estimates and correctness.
\end{abstract}

\input{sections/1_introduction}

\input{sections/3_background}

\input{sections/4_methodology}

\input{sections/5_experimental_setup}

\input{sections/6_results}
\input{sections/7_conclusions}

\input{sections/8_limitations}
\section*{Acknowledgments}
This publication is part of the project LESSEN with project number NWA.1389.20.183 of the research program NWA ORC 2020/21 which is (partly) financed by the Dutch Research Council (NWO).

\bibliography{main}

\newpage
\appendix
\input{sections/appendix}

\end{document}

%% file: sections/1_introduction.tex
\section{Introduction}~\label{sec:intro}
\shrink

Large Language Models (LLMs) have recently demonstrated promising capabilities in various tasks, including question-answering, and various classification and clustering tasks ~\cite{jin:2025:searchr1, lin-2024-generate, Survey24Soudani, trivedi-2023-ircot}.
However, LLMs are prone to generating incorrect information for multiple reasons, such as lack of parametric knowledge~\cite{Mallen23popqa}, temporal knowledge shifts~\cite{Zhao24Set, kordjamshidi24spatial}, or noisy information introduced through retrieved documents in Retrieval-Augmented Generation (RAG)~\cite{Soudani24FTvsRAG, Min23FActScore}. 
As a result, the trustworthiness of LLM-generated responses has become a critical concern, directly impacting user satisfaction~\cite{Hou024Decomposing, Mahaut24Factual}.

\begin{figure}[t]
  \centering 
  \includegraphics[width=0.46\textwidth]{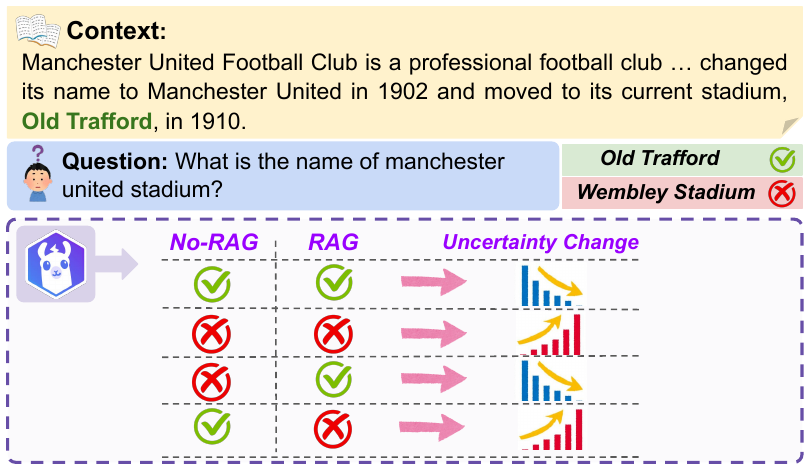}
  \caption{
 Desired behavior of uncertainty estimation methods with and without RAG.
For instance, the first row indicates that when an LLM generates a correct response both without and with RAG (i.e., the retrieved document supports the internal belief of the LLM), the uncertainty in the RAG setup should decrease compared to the no-RAG setup.
These principles form an axiomatic framework for evaluating and understanding uncertainty behavior in RAG.}
  
  \label{fig:main_fig}
  \shrink
\end{figure}

Uncertainty Estimation (UE) is a widely studied approach for assessing the reliability of LLM outputs. A UE method assigns an uncertainty score to each (input, output) pair, reflecting its truthfulness. Ideally, a perfect UE method would assign lower uncertainty to correct samples and higher uncertainty to incorrect ones~\cite{Duan24SAR}.
%
While existing UE methods mainly focus on scenarios where the input is just a query,
real-world applications like RAG involve non-parametric knowledge in more complex prompts~\cite{Huang24Calibration}. 
Research shows that non-parametric knowledge significantly influences LLM responses, often aligning them with the provided context~\cite{Cuconasu24Power, Mallen23popqa}. Despite this, it is unclear how current UE methods account for non-parametric knowledge.

In this paper, we investigate a critical question:
\textit{(\textbf{RQ1}) How do UE methods perform when the input prompt includes non-parametric knowledge, such as in RAG?} 
We study UE in the context of RAG with retrievers of varying effectiveness: (i) a deliberately weak synthetic retriever that returns irrelevant documents, (ii) an idealized retriever that consistently ranks the gold document at the top, and (iii) several widely used retrievers with varying performance levels. 
Our findings unveil that the performance of existing UE methods is inconsistent and mainly deteriorates when non-parametric knowledge is included in the input prompt.  
Most notably, improvements on the proposed UE methods in the literature do not add up when considering RAG setup.

Against this background, it is clear that UE  requires a methodological departure; existing methods are developed without paying attention to the specific properties that UE methods must satisfy in the RAG setup.
The question that arises here is: \textit{(\textbf{RQ2}) What properties can guarantee optimal performance of UE considering  LLMs' both parametric and non-parametric knowledge?}
We approach this question theoretically using axiomatic thinking, which is proven effective in various fields and tasks, including information retrieval~\cite{exploration05Fang, bondarenko2022axiomatic}, 
interpretability~\cite{Chen24Axiomatic, Parry25MechIR}, and preference modeling~\cite{Rosset23AxiomaticPreference}.
In axiomatic thinking, a set of formal constraints is defined based on desired properties, which are then used as a guide to search for an optimal solution. In this work, we define an axiomatic framework for UE and establish five axioms considering the desired behavior of a UE method with and without external knowledge.
Our axiomatic analysis reveals that current UE methods can satisfy only two axioms, violating the remaining three axioms in the majority of cases.

The axiomatic framework helps explaining deficiencies of existing UE methods for the RAG setup.
The next question is: \textit{\textbf{(RQ3)} Can the axiomatic framework guide us in deriving an optimal UE method?}  
We use the constraints of the axiomatic framework to define a calibration function based on three components. We implement three instantiations of this function and apply it to  different UE methods on a number of representative datasets. 
The results show that the derived functions are not only more stable than the existing UE methods but also improve overall performance with respect to AUROC. 
This highlights two key insights: first, satisfying the axioms leads to performance improvements, and second, existing UE methods can still be used for RAG by incorporating an axiomatically informed coefficient. 

%

\noindent
The main \textbf{contributions} of this paper include: \\
\noindent \textbf{(1)} Analyzing existing UE methods and showing their deficiencies in RAG setup. \\
\noindent \textbf{(2)}  Proposing an axiomatic framework for UE with five formalized constraints and demonstrating deficiencies of existing methods in satisfying them. \\
\noindent \textbf{(3)} Introducing a calibration function guided by axioms and showing consistent improvements of the UE methods as a result of alignment with axioms.

%% file: sections/3_background.tex
\section{Background}~\label{sec:background}
\shrink
 
\noindent
UE methods are typically divided into white-box approaches, which utilize token probabilities and entropy~\cite{Kadavath22PE, Kuhn23SE}, and black-box approaches, which rely solely on final outputs~\cite{Lin24ECC, Band24Linguistic}. This section reviews methods of both categories that are explored in this paper. For further details on related work, see Appendix~\ref{sec:related_work}.

\setlength{\abovedisplayskip}{0pt} 
\setlength{\belowdisplayskip}{0pt} 

\subsection{White-box Methods}

\noindent
\textbf{\textit{Predictive Entropy (PE)}} for generative models
quantifies uncertainty as the entropy of responses for an LLM input. The entropy is maximized when all outcomes are equally likely, indicating low informativeness~\cite{Kadavath22PE, Kuhn23SE}. Given an LLM parametrized by $\theta$ and an input $x$, the LLM uncertainty is estimated by computing entropy using Monte-Carlo approximation:
\begin{equation}
PE(\llminput, \theta) = -\frac{1}{B} \sum_{b=1}^B \ln P\left(\response_b \mid \llminput, \theta\right),
\label{eq:pe}
\end{equation}
where \( \response_b \) is a beam-sampled response and \( B \) is the number of samples.
The probability of generating a response \( \response = \{\response^1, \response^2, ..., \response^N\} \), comprising $N$ tokens, given the input $x$ is computed as the product of the conditional probabilities of each token, given its preceding tokens and the input \( \llminput \). For a model with parameters \( \theta \), the sequence probability is defined as:
\begin{equation}
P(\response \mid \llminput, \theta)=\prod_{n=1}^N P\left(\response^n \mid \response^{<n}, \llminput; \theta\right),
\label{eq:sequence_probability}
\end{equation}
where \( \response^{<n} \) denotes the tokens generated before \( \response^n \).

\medskip
\noindent\textbf{\textit{Semantic Entropy (SE)}}~\cite{Kuhn23SE} extends PE by incorporating the semantic meaning of sampled responses.
In this approach, generated samples are clustered into semantic clusters $c_i\in C$, and SE is defined as:
\begin{equation}
SE(\llminput, \theta)=-\frac{1}{|C|} \sum_{i=1}^{|C|} \log \tilde{P}\left(c_i \mid \llminput, \theta\right),
\label{eq:se}
\end{equation}
where \( c_i \) represents a semantic cluster, containing semantically similar responses. The cluster score $\tilde{P}(c_i|.)$ is computed as:
\[
\tilde{P}(c_i \mid \llminput, \theta)=\sum_{\response \in c_i} P(\response \mid \llminput, \theta).
\]

\noindent
\textbf{\textit{Length Normalization and Semantic Awareness}}
are two important components in UE.
It has been observed that the sequence probability in Equation~\eqref{eq:sequence_probability} is biased against longer generations~\cite{Malinin21LNPE}. To address this, a length-normalized probability is introduced to generate equal weighting of tokens and reduce bias toward shorter sequences:
\begin{equation*}
\label{eq:ln}
P_{\text{ln}}(\response \mid \llminput, \theta) = \prod_{n=1}^N P\left(\response^n \mid \response^{<n}, \llminput ; \theta\right)^{\frac{1}{N}}.
\end{equation*}

MARS~\cite{Bakman24mars} and TokenSAR~\cite{Duan24SAR} further refined this approach by incorporating semantic importance. These approaches assign weights based on each token's contribution, resulting in the meaning-aware probability:
\begin{equation*}
\label{eq:me}
P_{\text{me}}(\response \mid \llminput, \theta) = \prod_{n=1}^N P\left(\response^n \mid \response^{<n}, \llminput ; \theta\right)^{w(\response, \llminput, N, n)}
\end{equation*}
where \( w(\response, \llminput, N, n) \) is the importance weight for the \( n \)-th token. Both the length-normalized and meaning-aware probabilities can be used in the PE~\eqref{eq:pe} and SE~\eqref{eq:se} equations.

\subsection{Black-box Methods}
We examine state-of-the-art semantic similarity-based methods~\cite{Lin24ECC}, following these steps:
(i) generate \( B \) sampled responses \(\{\response_1, \dots, \response_B\}\) for a given input \( \llminput \);
(ii) compute pairwise similarity scores \( a_{i,j} = a(\response_i, \response_j) \) between the responses; and
(iii) derive uncertainty from these scores. 
Three approaches are proposed for computing uncertainty scores, described below.

\medskip
\noindent\textbf{\textit{Sum of Eigenvalues (EigV)}}~\cite{Lin24ECC}. SE groups responses into semantic equivalence subsets and uses their count (\textit{NumSet}) as an uncertainty metric; greater diversity implies higher uncertainty. To compute a more nuanced and continuous value for uncertainty than \textit{NumSet}, \citet{Lin24ECC} define uncertainty as:
\begin{equation}
    U_{\text{EigV}}(\llminput) = \sum_{k=1}^B \max \left(0,1-\lambda_k\right),
\end{equation}
where \( \lambda_1, \dots, \lambda_B \) are the eigenvalues of symmetric normalized
Graph Laplacian~\cite{Luxburg07spectral}, defined as:
\begin{equation*}
    L := I - D^{-\frac{1}{2}} W D^{-\frac{1}{2}}.
\end{equation*}
Here, $W$ represents a symmetric weighted adjacency matrix for a graph, where each node represents a response $r_i$ for input \( \llminput \) and weights are \( w_{i,j} = (a_{i,j} + a_{j,i})/2 \). The degree matrix $D$ is defined as: 
\begin{equation}
\label{eq:deg}
    D_{i,j} = \begin{cases}\sum_{j' \in [B]} w_{i,j'} & \text{if } i = j, \\ 0 & \text{if } i \neq j.\end{cases}
\end{equation}



\medskip
\noindent\textbf{\textit{Degree Matrix (Deg)}} relies on the degree matrix in Eq.~\eqref{eq:deg} to computer uncertainty. Here, the intuition is that \( D \) reflects node connectivity, and nodes with higher degrees indicate confident regions in the LLM~\cite{Lin24ECC}. Building on this, the uncertainty score is computed by:
\begin{equation*}
   U_{\text{Deg}}(\llminput)=\operatorname{trace}(BI-D) / B^2.
\end{equation*}

\medskip
\noindent\textbf{\textit{Eccentricity (ECC)}} is
defined as the average distance of response embeddings from their centroid, which can serve as an uncertainty measure. Since access to the embeddings is not possible in black-box LLMs, the embeddings are driven from graph Laplacian. Let $\mathbf{u}_1, \dots, \mathbf{u}_k \in \mathbb{R}^B$ be the $k$ smallest eigenvectors of $L$. For each response $ \response_j$, define the embedding as
$\mathbf{v}_j = [u_{1,j}, \dots, u_{k,j}]$~\cite{Ng01Advances},
and its centroid as
$\mathbf{v}_j^{\prime} = \mathbf{v}_j - \frac{1}{B} \sum_{j'=1}^B \mathbf{v}_{j'}.$
Uncertainty  is computed as:
\begin{equation*}
U_{\text{ECC}}(x) = \left\|\left[\mathbf{v}_1^{\prime \top}, \ldots, \mathbf{v}_B^{\prime \top}\right]\right\|_2.
\end{equation*}

%% file: sections/4_methodology.tex
\section{Axiomatic Framework}~\label{sec:axiom}
\shrink

\setlength{\abovedisplayskip}{0pt} 
\setlength{\belowdisplayskip}{0pt} 

The assumption of an axiomatic framework for UE is that by satisfying a set of formal constraints, a UE method would likely have an optimal correlation with correctness for both RAG and no-RAG setups. 
To define the framework, we introduce five \emph{axioms} based on a set of \emph{functions} that form our search space for an optimal UE.
These axioms, while necessary, do not represent an exhaustive list, as increasing the number of axioms can, in reality, introduce stringent, contradictory, or biased constraints. In the following, we introduce the functions and constraints of our axiomatic framework.

\subsection{Functions}
\label{sec:axiom:func}
We define UE as the task of learning a function \(\mathcal{U}\) that predicts a score \(s\), quantifying the LLM's uncertainty for its output~\cite{Liu24Uncertainty}. Formally, let $\llminput$ be the input given to a generative LLM $\mathcal{M}_{\theta}$, parameterized by $ \theta $. The uncertainty estimator function is formulated as follows:
\begin{equation*}
    \mathcal{U}: \mathcal{M}_{\theta}(x), \response \mapsto s
\end{equation*}
where the input consists of an LLM with the given input $\llminput$ and a generated response $\response$. 
In a no-RAG setting, the input $\llminput$ is only the query $\query$, while for the RAG setup, the input $\llminput$ consists of a query $\query$ and a context $\doc$, denoted as $ \mathcal{M}_{\theta}(\query, \doc) = \response$. We define context $\doc$ broadly, including an individual document or a set of documents.

Before defining the axioms, we introduce functions that formalize the relation between a context, a query, and an LLM-generated response. These functions, defined based on Natural Language Inference (NLI)~\cite{Pavlick16Most, Williams18mnli}, are as follows: 


\noindent
\textbf{Entailment ($ \doc \vDash (\query, \response) $)}: 
Given the context $\doc$, a human can infer that \(\response\) is the correct response to the query $\query$; i.e., the premise $\doc$ entails the hypothesis $(\query, \response)$ (asymmetric relation).

\noindent
\textbf{Contradiction ($ \doc \bot (\query, \response)$)}: 
Given the context $\doc$, a human can infer that $\response$ is an incorrect response to $\query$; i.e., the premise $\doc$ contradicts the hypothesis $(\query, \response)$ and vice versa (symmetric relation).

\noindent
\textbf{Independence ($ \doc \# (\query, \response)$)}: 
Given the context $\doc$, a human cannot infer any information about the correctness of response $\response$ to query $\query$; i.e., the premise $\doc$ does not guarantee the truth or falsity of hypothesis $(\query, \response)$ and vice versa (symmetric relation).

\noindent
\textbf{Equivalence (\( \response_1 \equiv \response_2 \))}:
Two LLM responses, $\response_1$ and $\response_2$, convey the same meaning; i.e., the premise $\response_1$ entails the hypothesis $\response_2$ and vice versa (symmetric relation).

\setlength{\abovedisplayskip}{2pt} 
\setlength{\belowdisplayskip}{2pt} 

\subsection{Axioms}
The axioms are defined based on two key assumptions to ensure the validity of axioms and the four aforementioned functions:




\begin{assumption}~\label{as:1}
The context $\doc$ is trustworthy and contains factually correct information.
\end{assumption}

\begin{assumption}~\label{as:2}
The context $\doc$, given to the LLM for the query $\query$, does not contain contradictory information about the query $\query$.
\end{assumption}

We now define five constraints that any reasonable UE method should satisfy, considering LLM's both parametric and non-parametric knowledge. Our working hypothesis is that UE is a proxy for the correctness of the model~\cite{Bakman24mars}. Two of these constraints are proven based on this hypothesis, and three of them are intuitively driven.

\begin{theorem}[\axiomOne]
$\forall {\query, \doc}$ if $ \mathcal{M}_{\theta}(\query) =$ $\response_1$, $ \mathcal{M}_{\theta}(\query, \doc) = \response_2$, $\response_1 \equiv \response_2$,  $\doc \vDash (q, \response_2$), then~~
$\mathcal{U}(\mathcal{M}_{\theta}(\query), \response_1) > \mathcal{U}(\mathcal{M}_{\theta}(\query, \doc), \response_2).$
\end{theorem}
%

This constraint states that if applying RAG does not alter the LLM’s response and the RAG context supports LLM's generated response $\response_2$, then LLM’s internal belief aligns with the context. In such a scenario, the uncertainty after applying RAG should be lower than before, as the retrieved context reinforces the LLM’s prior knowledge.  
For instance, consider the example in Figure~\ref{fig:main_fig}. Given the query, "\textit{What is the name of Manchester United’s stadium?}" if the LLM initially generates the correct response, "\textit{Old Trafford}," and the input context mentions "\textit{Old Trafford}" as the name of the stadium, then the uncertainty value after applying RAG should be lower than before.

\begin{theorem}[\axiomTwo]
$\forall {q, c}$ if $ \mathcal{M}_{\theta}(\query) = \response_1$,  $ \mathcal{M}_{\theta}(\query, \doc) = \response_2$,  $ \response_1 \equiv \response_2 $, $ \doc \bot (q, \response_2) $, then~~
$\mathcal{U}(\mathcal{M}_{\theta}(\query), \response_1) < \mathcal{U}(\mathcal{M}_{\theta}(\query, \doc), \response_2).$
\end{theorem}

This constraint states that if the LLM’s response remains unchanged after applying RAG, but the retrieved context $\doc$ contradicts the generated response $\response_2$,  then the LLM's internal belief does not align with the context. In such a case, the uncertainty after applying RAG should be higher than before, as the retrieved information challenges the LLM’s internal belief.  
For example, in Figure~\ref{fig:main_fig}, if LLM's response before and after RAG is "\textit{Wembley Stadium}," and RAG context contradicts the LLM’s response, then the uncertainty of the RAG response should increase. This means that although the LLM persists with its incorrect response, it does so with a lower confidence. 


\begin{theorem}[\axiomThree]
$\forall {q, c}$ if $\mathcal{M}_{\theta}(\query) = \response_1$, $ \mathcal{M}_{\theta}(\query, \doc) = \response_2$,  $ \neg (\response_1 \equiv \response_2) $, $ \doc \bot (q, \response_1) $,  $ \doc \vDash (q, \response_2) $, then 
%
\begin{equation*}
\label{eq:axiom3_2}
\mathcal{U}(\mathcal{M}_{\theta}(\query), \response_1) > \mathcal{U}(\mathcal{M}_{\theta}(\query, \doc), \response_2).
\end{equation*}
\end{theorem}
Theorem 3 directly follows from the statement in the following lemma:
\begin{lemma}
If $ \mathcal{M}_{\theta}(x_1) = \response_1$ , $\mathcal{M}_{\theta}(x_2) = \response_2$, $\response_1$ is $False$, $\response_2$ is $True$, then 
\begin{equation*}
    \mathcal{U}(\mathcal{M}_{\theta}(x_1), \response_1) > \mathcal{U}(\mathcal{M}_{\theta}(x_2), \response_2).
\end{equation*}
\end{lemma}
\begin{proof}
Given Assumptions \ref{as:1} and \ref{as:2} and $ \doc \bot (q, \response_1) $, then response $\response_1$ is $False$.
Similarly, given that $ \doc \vDash (q, \response_2) $,  then response $\response_2$ is $True$.
Given these events and Lemma 1, then  $\mathcal{U}(\mathcal{M}_{\theta}(q), \response_1) > \mathcal{U}(\mathcal{M}_{\theta}(q, c), \response_2)$.
\end{proof}

This constraint states that if the LLM’s response changes from $ \response_1 $ to $ \response_2 $ after applying RAG, and the RAG context $\doc$ supports $ \response_2 $ while contradicting $ \response_1 $, then the estimated uncertainty for $\response_2$ should be lower than one for $\response_1$. 
For example, consider the case illustrated in Figure~\ref{fig:main_fig}. If the LLM initially generates "\textit{Wembley Stadium}" but then, after seeing a context containing the correct response, changes its output to "\textit{Old Trafford}," the uncertainty of "\textit{Old Trafford}" with RAG should be lower than the uncertainty of "\textit{Wembley Stadium}" without RAG.  

\begin{theorem}[\axiomFour]
$\forall {q, c}$ if $ \mathcal{M}_{\theta}(\query) = \response_1$, $ \mathcal{M}_{\theta}(\query, \doc) = \response_2$, $ \neg(\response_1 \equiv \response_2) $, $ \doc \vDash (q, \response_1) $, $ \doc \bot (q, \response_2) $, then 
\begin{equation*}
\label{eq:axiom4_2}
    \mathcal{U}(\mathcal{M}_{\theta}(\query), \response_1) < \mathcal{U}(\mathcal{M}_{\theta}(\query, \doc), \response_2).
\end{equation*}
\end{theorem}
%
This theorem follows from the statement
in the Lemma 1 with the following proof.
\begin{proof}
The proof is similar to that of
Theorem~3. Given Assumptions \ref{as:1} and \ref{as:2} and $ \doc \vDash (q, \response_1) $,  then response $\response_1$ is correct. 
Similarly, response $\response_1$ is incorrect because $ \doc \bot (q, \response_1) $.
Based on Lemma 1 and these events, then  $\mathcal{U}(\mathcal{M}_{\theta}(q), \response_1) < \mathcal{U}(\mathcal{M}_{\theta}(q, c), \response_2)$.
\end{proof}

This constraint states that if the LLM’s response changes from $ \response_1 $ to $ \response_2 $ after applying RAG, where $\response_1$ is correct, and $\response_2$ is incorrect, then the estimated uncertainty of $ \response_2 $ should be higher than the one for $ \response_1$.
In the example of Figure~\ref{fig:main_fig}, the LLM generates the correct response "\textit{Old Trafford}" and changes its response to "\textit{Wembley Stadium}" in the RAG setup, which is incorrect. In this scenario, the uncertainty of the RAG response should be higher than that of the original response without RAG.

\begin{theorem}[\axiomFive]
$\forall \query,\doc$ if $ \mathcal{M}_{\theta}(\query) = \response_1$, $ \mathcal{M}_{\theta}(\query, \doc) = \response_2$, $\response_1 \equiv \response_2$, $\doc \# (\query, \response_1)$, then~~
    $\mathcal{U}(\mathcal{M}_{\theta}(\query), \response_1) \approx \mathcal{U}(\mathcal{M}_{\theta}(\query,\doc), \response_2). $
\end{theorem} \
This constraint states that if the LLM’s response remains unchanged after applying RAG, and the retrieved context $\doc$ is unrelated to the query and responses $ \response_1 $ and $ \response_2$, then the context neither supports nor contradicts the LLM’s belief. In this case, the estimated salary should remain similar.
For example, consider the query "\textit{Who wrote the book The Origin of Species?}". If, in the RAG setup, the LLM is provided with the context shown in Figure~\ref{fig:main_fig}, which is unrelated to the query, then as long as the response remains unchanged, the uncertainty value should remain unaffected.  

\subsection{Instantiation}~\label{sec:axioms:inst}
To empirically examine UE methods against these axioms, we need to define a specific instantiation of functions in our framework (cf. Sec.~\ref{sec:axiom:func}). We introduce two instantiations of these functions: \emph{reference-based} and \emph{reference-free}. 
The reference-based instantiation assumes the existence of a benchmark containing ground truth responses to queries. Such a benchmark is not available for reference-free instantiation.


%

\medskip
\noindent \textbf{Reference-based.}
In this setup, we rely on ground truth labels to check the condition of each axiom. We assume that for every $\query$, the correct response $\hat{\response}$ is available in our ground truth. The implementation of \emph{Entailment} and \emph{Contraction} functions then boils down to comparing the generated response $r$ against the ground truth response $\hat{\response}$. The comparison is performed using a matching function $\mathcal{E}(r_1, r_2)$, which assesses whether the two responses are equivalent. This function is also used to implement the \emph{Equivalence} function (cf. Sec~\ref{sec:axiom:func}). For datasets containing factual queries with short responses, $\mathcal{E}(.)$ is an Exact Match (EM) function, which returns \textit{True} if and only if the two responses are identical on a token-by-token basis~\cite{Mallen23popqa}. 
Using this setup, the following conditions can be inferred for our axioms: 

\noindent
\textit{Axiom 1. } $\mathcal{E}(\response_1, \response_2) = True$,  $\mathcal{E}(\response_2, \hat{\response}) = True$.

\noindent
\textit{Axiom 2.} $ \mathcal{E}(\response_1, \response_2) = True$, $\mathcal{E}(\response_2, \hat{\response}) = False$.

\noindent
\textit{Axiom 3.} $ \mathcal{E}(\response_1, \response_2) = False$, $ \mathcal{E}(\response_1, \hat{\response}) = False $,  $ \mathcal{E}(\response_2, \hat{\response}) = True $.

\noindent
\textit{Axiom 4.} $ \mathcal{E}(\response_1, \response_2) = False$,  $ \mathcal{E}(\response_1, \hat{\response}) = True $, $ \mathcal{E}(\response_2, \hat{\response}) = False $.

\noindent
\textit{Axiom 5.} $\mathcal{E}(\response_1, \response_2) = True, \doc $ is not relevant to $q$.


\medskip
\noindent \textbf{Reference-free.}
Since access to the correctness labels of LLM's responses limits the applicability of axioms to unseen queries, we propose a reference-free implementation of axioms.
Specifically, we leverage an NLI classifier to assess the relationship between the generated response and the context, denoted as $\mathcal{R}(.)$.
Following~\cite{Kuhn23SE, Lin24ECC}, we implement \emph{Entailment} by merging entailment and neutral classes into a single class. The contradiction class of the NLI classifier is considered for the \emph{Contradiction} function. Similar to the reference-based instantiation, function $\mathcal{E}(.)$ is used for \emph{Equivalence}. 
Using these definitions, the axioms are defined as follows:

\noindent
\textit{Axiom 1.} $ \mathcal{E}(\response_1, \response_2) = True$, $ \mathcal{R}(\doc, \query, \response_2) = Entailment $.

\noindent
\textit{Axiom 2.} $ \mathcal{E}(\response_1, \response_2) = True$, $ \mathcal{R}(\doc, \query, \response_2) = Contradiction $.

\noindent
\textit{Axiom 3.} $ \mathcal{E}(\response_1, \response_2) = False$, $ \mathcal{R}(\doc, \query, \response_1) = Contradiction $, $ \mathcal{R}(\doc, \query, \response_2) = Entailment $.

\noindent
\textit{Axiom 4.} $ \mathcal{E}(\response_1, \response_2) = False$, $ \mathcal{R}(\doc, \query, \response_2) = Entailment $, $ \mathcal{R}(\doc, \query, \response_2) = Contradiction $.

\noindent
Axiom 5  mirrors the reference-based setup, due to the limitations of existing NLI methods in predicting the neutral relation.



\section{Derivation of a Calibration Function}
\label{sec:calibration_function}


In this section, we derive a calibration function that improves existing UE methods using our axiomatic framework. 
To recap, our formal constraints are built around four functions that are examined for LLM responses without RAG ($\response_1$) and with RAG ($\response_2$).
In the reference-free instantiation of our framework (cf. Sec.~\ref{sec:axioms:inst}), we showed that these functions are of two types: (i) Equivalence that examines the relation between two LLM-generated responses, represented as  $\mathcal{E}(\response_1, \response_2)$, and (ii) other functions that examine entailment, contradiction, and independence relations between context, query, and an LLM generated response, represented as $\mathcal{R}(c, q, r)$.
We define a calibration coefficient by searching the space of our axiomatic constraints using these two types of functions:
\begin{equation*}
    \begin{aligned}
        \alpha_{\text{ax}} = 
        k_1\cdot\mathcal{E}(\response_1, \response_2) +
        k_2 \cdot \mathcal{R}(\doc, \query, \response_1) + \\
        k_3 \cdot \mathcal{R}(\doc, \query, \response_2), 
    \end{aligned}
\label{eq:axiomatic_coefficient}
\end{equation*}
where \(k_1\), \(k_2\), \(k_3\) are hyper parameters, and $\response_1, \response_2$ represent LLM generated responses without and with RAG, respectively.
The calibrated UE function for RAG is then defined as:
\[
\mathcal{U}(\mathcal{M}_{\theta}(\doc, \query), \response_2)^{\text{cal}} = (k_4 - \alpha_{\text{ax}}) \cdot \mathcal{U}(\mathcal{M}_{\theta}(\doc, \query), \response_2).
\]

The hyper parameters $k_1$--$k_4$ are set to satisfy the axioms using a validation set. 
This calibration enables increasing the uncertainty score of RAG for samples associated with axioms 2 and 4 while decreasing it for samples related to axioms 1 and 3.


\subsection{Instantiation}~\label{sec:calibration_function:inst}
We propose three instantiations of the calibration function, where three different models are used to implement $\mathcal{R}$. 


\noindent\textbf{CTI}. The first model is based on the Context-sensitive Token Identification (CTI) task, which has been applied in self-citation and groundedness evaluation~\cite{Sarti24Quantifying, Qi24Model}.
In this approach, each token in \( \response = \{\response^1, \response^2, \dots, \response^N\} \) is evaluated using a contrastive metric \(m\) (e.g., KL divergence, comparing the LLM's response distributions with and without the context. The resulting scores are $\{m_1, m_2, \dots, m_N\}$, where $m_n = \mathrm{KL}(
    P\left(\response^n \mid \response^{<n}, (\query, \doc); \theta\right)
    \parallel 
    P\left(\response^n \mid \response^{<n}, \query; \theta\right)
    ).$
These scores are converted into binary values via the selector function \(S_{\text{CTI}}\). The overall relation score is then computed as:
\begin{equation*}
\mathcal{R}(\doc, \query, \response) = \frac{1}{N}\sum_{n=1}^{N} S_{\text{CTI}}(m_n).
\end{equation*}

\noindent\textbf{NLI.}
The second model employs an NLI-based approach that quantifies the relationship using entailment probability:
\begin{equation*}
\mathcal{R}(\doc, \query, \response) = \mathcal{N}_{\vDash}(\doc, (\query, \response)).
\end{equation*}

\noindent\textbf{MiniCheck.}
Finally, the third model employs MiniCheck~\cite{Tang24MiniCheck}, which performs sentence-level fact-checking using a fine-tuned model. 
It produces a score between $0$ and $1$ indicating how well the \(\response\) is grounded in the \(\doc\):
\begin{equation*}
\mathcal{R}(\doc, \query, \response) = \text{MiniCheck}(\doc, (\query,\response)).
\end{equation*}


In all three instantiations, the equivalence function \(\mathcal{E}(\response_1, \response_2)\) is an NLI classifier, wherein the entailment probability serves as a continuous measure of similarity between \(\response_1\) and \(\response_2\)~\cite{Kuhn23SE}; formally $\mathcal{E}(\response_1, \response_2) = \mathcal{N}_{\vDash}(\response_1, \response_2)$.

%% file: sections/5_experimental_setup.tex
\section{Experimental Setup}~\label{sec:experimental_setup}
\shrink

\noindent
\textbf{Datasets.}  
We evaluate our approach on three open-book QA datasets, Natural Questions (NQ-open)~\cite{Lee19nqopen}, TriviaQA~\cite{Joshi17TriviaQA}, and \textsc{PopQA}~\cite{Mallen23popqa}. 
For each dataset, we randomly sample 3,000 examples as the test set.
We create a validation set for each dataset, comprising 300 samples, which is used to compute calibration coefficients as described in Section~\ref{sec:calibration_function}. For NQ-open and TriviaQA, the validation set is sampled from the training set, whereas for \textsc{PopQA}, it is derived from the test set.

\noindent\textbf{Methods.}  
We evaluate three white-box UE methods: PE, SE, and MARS applied to PE and SE (denoted as PE+M and SE+M), as well as three black-box methods: Deg, ECC, and EigV (cf. Sec.~\ref{sec:background}).
\noindent\textbf{Experimental setup.}  
Our experiments involve the reproduction of existing UE methods for the RAG setup. To ensure a fair comparison, we employ LLMs that are used in the original papers: Llama2-chat 7B and Mistral-7B. For uncertainty computation, 10 responses per query are generated with a temperature setting of \( T = 1 \); for correctness evaluation, the most likely response is considered.
%
%

Following \citet{Kuhn23SE}, we use Deberta-large model fine-tuned on MNLI as NLI classifier.
%
BM25, Contriever~\citep{Unsupervised22Izacard}, and BM25+Reranker are used as retrievers. Manually chosen relevant and irrelevant documents are denoted with $\text{Doc}^{+}$ and $\text{Doc}^{-}$, respectively.

\noindent\textbf{Metrics.}  
We report the Exact Match for correctness and AUROC~\cite{Bakman24mars}. 
We report on statistical significance using Wilcoxon test with $p$-value < $0.01$; see Appendix~\ref{sec:appendix_es} for further details.

\noindent
\textbf{Calibration Function.}
We perform a grid search on the validation set of each dataset to determine the axiomatic coefficients (\(k_1, k_2, k_3, k_4\)) as described in Section~\ref{sec:calibration_function}. This grid search simultaneously pursues two objectives: satisfying the axioms and maximizing the overall AUROC. For the CTI method, the optimal coefficients are (0.05, 0.20, 0.75, 1.30); for the NLI and MiniCheck methods, the optimal coefficients are (0.05, 0.05, 0.90, 1.20) consistently across all datasets.
We observed that the calibration coefficient values are consistent across different datasets and LLMs, which is expected given that the range of uncertainty scores does not vary significantly across datasets and LLMs. Specifically, $k_3$ consistently takes higher values than other values. The lower value of $k_3$ for CTI compared to other NLI and MiniCheck is due to its higher error rate in capturing the relationship between the retrieved document and the generated output. We found that decreasing $k_3$ (or increasing $k_1$ and $k_2$) consistently leads to lower AUROC scores, while reducing $k_1$ or $k_2$ results in fewer satisfied axioms. $k_4$ remains relatively stable across configurations.

%% file: sections/6_results.tex
\section{Results}~\label{sec:results}
\shrink
\shrink

\begin{table}[t]
\scriptsize
\centering \shrink
\setlength{\tabcolsep}{1.4pt}
\begin{tabular}{@{~}cc|cccccc}
\hline
\textbf{LLM} & \textbf{Unc.} &
\multicolumn{6}{c}{\textbf{PopQA}} \\ \hline

\hline
& &
\textbf{$\text{No Doc}$} & \textbf{$\text{Doc}^{-}$} & \textbf{BM25} & \textbf{Cont.} & \textbf{ReRa.} & \textbf{$\text{Doc}^{+}$} \\
\hline\hline

\multirow{10}{*}{\rotatebox[origin=c]{90}{Llama2-chat}} &
PE   &
\colorbox{orange!50}{1.29~\textsuperscript{ }} &
\colorbox{orange!40}{1.11~\textsuperscript{$\ast$}} &
\colorbox{orange!30}{0.54~\textsuperscript{$\ast$}} &
\colorbox{orange!20}{0.46~\textsuperscript{$\ast$}} &
\colorbox{orange!10}{0.35~\textsuperscript{$\ast$}} &
\colorbox{orange!5}{0.34~\textsuperscript{$\ast$}} \\ 

& SE   & 
\colorbox{orange!50}{4.86~\textsuperscript{ }} &
\colorbox{orange!40}{4.37~\textsuperscript{$\ast$}} &
\colorbox{orange!30}{3.45~\textsuperscript{$\ast$}} &
\colorbox{orange!20}{3.30~\textsuperscript{$\ast$}} &
\colorbox{orange!5}{3.13~\textsuperscript{$\ast$}} &
\colorbox{orange!10}{3.19~\textsuperscript{$\ast$}} \\ 

& PE+M &
\colorbox{orange!50}{1.59~\textsuperscript{ }} &
\colorbox{orange!40}{1.34~\textsuperscript{$\ast$}} &
\colorbox{orange!30}{0.65~\textsuperscript{$\ast$}} &
\colorbox{orange!20}{0.55~\textsuperscript{$\ast$}} &
\colorbox{orange!5}{0.44~\textsuperscript{$\ast$}} &
\colorbox{orange!10}{0.45~\textsuperscript{$\ast$}} \\ 

& SE+M &
\colorbox{orange!50}{5.38~\textsuperscript{ }} &
\colorbox{orange!40}{4.71~\textsuperscript{$\ast$}} &
\colorbox{orange!30}{3.62~\textsuperscript{$\ast$}} &
\colorbox{orange!20}{3.43~\textsuperscript{$\ast$}} &
\colorbox{orange!5}{3.23~\textsuperscript{$\ast$}} &
\colorbox{orange!10}{3.27~\textsuperscript{$\ast$}} \\

& Deg & 
\colorbox{orange!50}{0.52~\textsuperscript{ }} &
\colorbox{orange!40}{0.32~\textsuperscript{$\ast$}} &
\colorbox{orange!30}{0.12~\textsuperscript{$\ast$}} &
\colorbox{orange!20}{0.09~\textsuperscript{$\ast$}} &
\colorbox{orange!10}{0.06~\textsuperscript{$\ast$}} &
\colorbox{orange!5}{0.05~\textsuperscript{$\ast$}} \\

& ECC &
\colorbox{orange!50}{0.71~\textsuperscript{ }} &
\colorbox{orange!40}{0.54~\textsuperscript{$\ast$}} &
\colorbox{orange!30}{0.22~\textsuperscript{$\ast$}} &
\colorbox{orange!20}{0.17~\textsuperscript{$\ast$}} &
\colorbox{orange!10}{0.12~\textsuperscript{$\ast$}} &
\colorbox{orange!5}{0.10~\textsuperscript{$\ast$}} \\

& EigV &
\colorbox{orange!50}{4.25~\textsuperscript{ }} &
\colorbox{orange!40}{2.28~\textsuperscript{$\ast$}} &
\colorbox{orange!30}{1.42~\textsuperscript{$\ast$}} &
\colorbox{orange!20}{1.31~\textsuperscript{$\ast$}} &
\colorbox{orange!10}{1.18~\textsuperscript{$\ast$}} &
\colorbox{orange!5}{1.17~\textsuperscript{$\ast$}} \\

\hline
\multirow{10}{*}{\rotatebox[origin=c]{90}{Mistral-v0.3}} &
PE   & 
\colorbox{orange!50}{1.51~\textsuperscript{ }} &
\colorbox{orange!40}{0.94~\textsuperscript{$\ast$}} &
\colorbox{orange!30}{0.84~\textsuperscript{$\ast$}} &
\colorbox{orange!20}{0.69~\textsuperscript{$\ast$}} &
\colorbox{orange!10}{0.62~\textsuperscript{$\ast$}} &
\colorbox{orange!5}{0.51~\textsuperscript{$\ast$}} \\

& SE   &
\colorbox{orange!50}{5.66~\textsuperscript{ }} &
\colorbox{orange!40}{3.73~\textsuperscript{$\ast$}} &
\colorbox{orange!30}{3.68~\textsuperscript{$\ast$}} &
\colorbox{orange!20}{3.53~\textsuperscript{$\ast$}} &
\colorbox{orange!10}{3.41~\textsuperscript{$\ast$}} &
\colorbox{orange!5}{3.26~\textsuperscript{$\ast$}} \\

& PE+M &
\colorbox{orange!50}{2.35~\textsuperscript{ }} &
\colorbox{orange!40}{1.42~\textsuperscript{$\ast$}}& 
\colorbox{orange!30}{1.26~\textsuperscript{$\ast$}} & 
\colorbox{orange!20}{1.05~\textsuperscript{$\ast$}} & 
\colorbox{orange!10}{0.92~\textsuperscript{$\ast$}} & 
\colorbox{orange!5}{0.80~\textsuperscript{$\ast$}} \\

& SE+M &
\colorbox{orange!50}{6.47~\textsuperscript{ }} &
\colorbox{orange!40}{4.05~\textsuperscript{$\ast$}} &
\colorbox{orange!30}{3.98~\textsuperscript{$\ast$}} &
\colorbox{orange!20}{3.77~\textsuperscript{$\ast$}} &
\colorbox{orange!10}{3.60~\textsuperscript{$\ast$}} &
\colorbox{orange!5}{3.45~\textsuperscript{$\ast$}} \\

& Deg  &
\colorbox{orange!50}{0.48~\textsuperscript{ }} &
\colorbox{orange!20}{0.05~\textsuperscript{$\ast$}} &
\colorbox{orange!40}{0.07~\textsuperscript{$\ast$}} &
\colorbox{orange!30}{0.06~\textsuperscript{$\ast$}} &
\colorbox{orange!20}{0.05~\textsuperscript{$\ast$}} &
\colorbox{orange!5}{0.03~\textsuperscript{$\ast$}} \\

& ECC  &
\colorbox{orange!50}{0.68~\textsuperscript{ }} &
\colorbox{orange!10}{0.03~\textsuperscript{$\ast$}}&
\colorbox{orange!40}{0.08~\textsuperscript{$\ast$}} &
\colorbox{orange!40}{0.08~\textsuperscript{$\ast$}} &
\colorbox{orange!20}{0.05~\textsuperscript{$\ast$}} &
\colorbox{orange!5}{0.04~\textsuperscript{$\ast$}} \\

& EigV &
\colorbox{orange!50}{4.18~\textsuperscript{ }} &
\colorbox{orange!10}{1.08~\textsuperscript{$\ast$}}&
\colorbox{orange!30}{1.16~\textsuperscript{$\ast$}} &
\colorbox{orange!40}{1.17~\textsuperscript{$\ast$}} &
\colorbox{orange!20}{1.11~\textsuperscript{$\ast$}} &
\colorbox{orange!10}{1.08~\textsuperscript{$\ast$}} \\


\hline
\end{tabular}
\caption{Average uncertainty values for various settings. Lighter colors indicate lower uncertainty. Statistically significant differences are compared to \textit{No Doc} are marked with $\ast$.}
\label{tab:uncertainty_value_rag_methods}
\shrink
\shrink
\end{table}

\begin{figure}[t]
  \centering \shrink
  \includegraphics[width=0.46\textwidth]{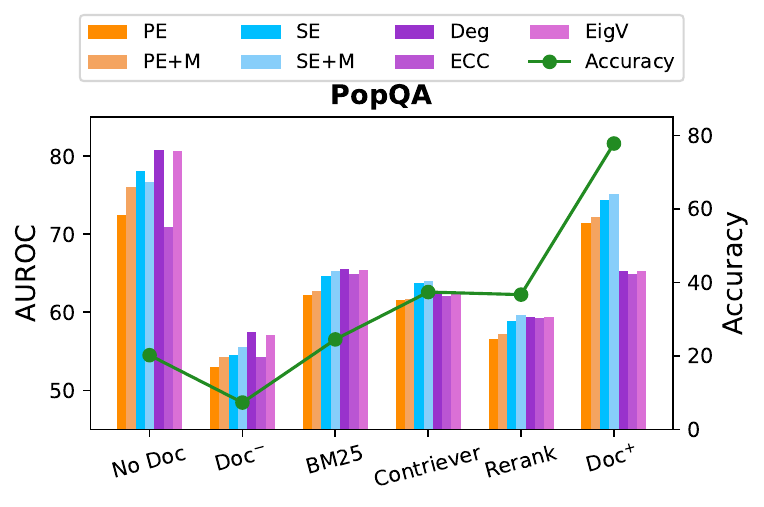}
  \shrink
  \caption{Comparison of AUROC between no-RAG and RAG settings for Llama2-chat.}
  \label{fig:auroc_all_llama2}
  \shrink
\end{figure}

\subsection{Uncertainty Changes with RAG}
\textit{\textbf{(RQ1)}} examines how the performance of UE methods and their associated uncertainty values vary with and without context in the input prompt. Figures~\ref{fig:auroc_all_llama2} and \ref{fig:auroc_all_1} present accuracy and AUROC for different RAG settings.
We observe inconsistent behavior of UE methods with and without RAG across different datasets, often displaying drop AUROC for RAG cases, except for $\text{\textit{Doc}}^{+}$. While AUROC should be independent accuracy, the results suggest a correlation between the performance of the RAG method and AUROC; especially when considering irrelevant and relevant documents. This indicates a bias of current UE methods towards RAG generations.


To assess this bias further, we report on average uncertainty values of these methods in  Tables~\ref{tab:uncertainty_value_rag_methods} and \ref{tab:uncertainty_value_rag_methods_1}. The results reveal that incorporating any context results in lower uncertainty values. Even the inclusion of irrelevant contexts, which do not enhance accuracy, leads to a significant reduction in uncertainty scores. This suggests that current UE methods produce lower uncertainty values in the RAG setup without adequately accounting for the relevance of the context. 


\begin{table}[t]
\centering \shrink
\setlength{\tabcolsep}{3pt}
\scriptsize
\begin{tabular}{c|ccc}
\hline
\textbf{UE} & \multicolumn{3}{c}{\textbf{PopQA}} \\ \hline
& \textbf{BM25} & \textbf{Contriever} & \textbf{$\text{Doc}^{+}$} \\
\hline\hline

\multicolumn{3}{l}{\textbf{Axiom 1:} Positively Consistent $\downarrow$} \\ \hline PE &\colorbox{green!10}{0.735 $\rightarrow$ 0.419~\textsuperscript{$\ast$}} &\colorbox{green!10}{0.735 $\rightarrow$ 0.408~\textsuperscript{$\ast$}} &\colorbox{green!10}{1.242 $\rightarrow$ 0.340~\textsuperscript{$\ast$}} \\SE &\colorbox{green!10}{3.781 $\rightarrow$ 3.205~\textsuperscript{$\ast$}} &\colorbox{green!10}{3.791 $\rightarrow$ 3.158~\textsuperscript{$\ast$}} &\colorbox{green!10}{4.682 $\rightarrow$ 3.113~\textsuperscript{$\ast$}} \\PE+M &\colorbox{green!10}{0.896 $\rightarrow$ 0.483~\textsuperscript{$\ast$}} &\colorbox{green!10}{0.881 $\rightarrow$ 0.458~\textsuperscript{$\ast$}} &\colorbox{green!10}{1.530 $\rightarrow$ 0.406~\textsuperscript{$\ast$}} \\SE+M &\colorbox{green!10}{4.102 $\rightarrow$ 3.286~\textsuperscript{$\ast$}} &\colorbox{green!10}{4.091 $\rightarrow$ 3.248~\textsuperscript{$\ast$}} &\colorbox{green!10}{5.146 $\rightarrow$ 3.173~\textsuperscript{$\ast$}} \\EigV &\colorbox{green!10}{1.951 $\rightarrow$ 1.166~\textsuperscript{$\ast$}} &\colorbox{green!10}{2.025 $\rightarrow$ 1.143~\textsuperscript{$\ast$}} &\colorbox{green!10}{4.074 $\rightarrow$ 1.078~\textsuperscript{$\ast$}} \\ECC &\colorbox{green!10}{0.417 $\rightarrow$ 0.110~\textsuperscript{$\ast$}} &\colorbox{green!10}{0.426 $\rightarrow$ 0.094~\textsuperscript{$\ast$}} &\colorbox{green!10}{0.710 $\rightarrow$ 0.055~\textsuperscript{$\ast$}} \\Deg &\colorbox{green!10}{0.220 $\rightarrow$ 0.048~\textsuperscript{$\ast$}} &\colorbox{green!10}{0.230 $\rightarrow$ 0.043~\textsuperscript{$\ast$}} &\colorbox{green!10}{0.496 $\rightarrow$ 0.022~\textsuperscript{$\ast$}} \\\hline\multicolumn{3}{l}{\textbf{Axiom 2:} Negatively Consistent $\uparrow$} \\ \hline PE &\colorbox{magenta!10}{1.068 $\rightarrow$ 0.746~\textsuperscript{ }} &\colorbox{magenta!10}{0.820 $\rightarrow$ 0.593~\textsuperscript{ }} &\colorbox{magenta!10}{1.083 $\rightarrow$ 0.597~\textsuperscript{ }} \\SE &\colorbox{magenta!20}{4.163 $\rightarrow$ 3.548~\textsuperscript{$\ast$}} &\colorbox{magenta!20}{4.104 $\rightarrow$ 3.381~\textsuperscript{$\ast$}} &\colorbox{magenta!10}{4.388 $\rightarrow$ 4.107~\textsuperscript{ }} \\PE+M &\colorbox{magenta!10}{1.309 $\rightarrow$ 0.844~\textsuperscript{ }} &\colorbox{magenta!10}{1.016 $\rightarrow$ 0.782~\textsuperscript{ }} &\colorbox{magenta!10}{1.328 $\rightarrow$ 0.684~\textsuperscript{ }} \\SE+M &\colorbox{magenta!20}{4.599 $\rightarrow$ 3.700~\textsuperscript{$\ast$}} &\colorbox{magenta!20}{4.481 $\rightarrow$ 3.610~\textsuperscript{$\ast$}} &\colorbox{magenta!10}{4.764 $\rightarrow$ 4.221~\textsuperscript{ }} \\EigV &\colorbox{magenta!20}{2.453 $\rightarrow$ 1.338~\textsuperscript{$\ast$}} &\colorbox{magenta!20}{2.088 $\rightarrow$ 1.274~\textsuperscript{$\ast$}} &\colorbox{magenta!10}{2.758 $\rightarrow$ 1.910~\textsuperscript{ }} \\ECC &\colorbox{magenta!20}{0.541 $\rightarrow$ 0.197~\textsuperscript{$\ast$}} &\colorbox{magenta!20}{0.477 $\rightarrow$ 0.152~\textsuperscript{$\ast$}} &\colorbox{magenta!10}{0.503 $\rightarrow$ 0.443~\textsuperscript{ }} \\Deg &\colorbox{magenta!20}{0.286 $\rightarrow$ 0.101~\textsuperscript{$\ast$}} &\colorbox{magenta!20}{0.228 $\rightarrow$ 0.073~\textsuperscript{$\ast$}} &\colorbox{magenta!10}{0.343 $\rightarrow$ 0.254~\textsuperscript{ }} \\\hline\multicolumn{3}{l}{\textbf{Axiom 3:} Positively Changed $\downarrow$} \\ \hline PE &\colorbox{green!10}{1.375 $\rightarrow$ 0.347~\textsuperscript{$\ast$}} &\colorbox{green!10}{1.416 $\rightarrow$ 0.298~\textsuperscript{$\ast$}} &\colorbox{green!10}{1.342 $\rightarrow$ 0.268~\textsuperscript{$\ast$}} \\SE &\colorbox{green!10}{4.889 $\rightarrow$ 3.015~\textsuperscript{$\ast$}} &\colorbox{green!10}{5.091 $\rightarrow$ 3.013~\textsuperscript{$\ast$}} &\colorbox{green!10}{4.884 $\rightarrow$ 3.051~\textsuperscript{$\ast$}} \\PE+M &\colorbox{green!10}{1.708 $\rightarrow$ 0.398~\textsuperscript{$\ast$}} &\colorbox{green!10}{1.735 $\rightarrow$ 0.374~\textsuperscript{$\ast$}} &\colorbox{green!10}{1.604 $\rightarrow$ 0.340~\textsuperscript{$\ast$}} \\SE+M &\colorbox{green!10}{5.514 $\rightarrow$ 3.072~\textsuperscript{$\ast$}} &\colorbox{green!10}{5.681 $\rightarrow$ 3.082~\textsuperscript{$\ast$}} &\colorbox{green!10}{5.379 $\rightarrow$ 3.099~\textsuperscript{$\ast$}} \\EigV &\colorbox{green!10}{4.131 $\rightarrow$ 1.139~\textsuperscript{$\ast$}} &\colorbox{green!10}{4.733 $\rightarrow$ 1.114~\textsuperscript{$\ast$}} &\colorbox{green!10}{4.449 $\rightarrow$ 1.102~\textsuperscript{$\ast$}} \\ECC &\colorbox{green!10}{0.790 $\rightarrow$ 0.085~\textsuperscript{$\ast$}} &\colorbox{green!10}{0.823 $\rightarrow$ 0.081~\textsuperscript{$\ast$}} &\colorbox{green!10}{0.780 $\rightarrow$ 0.072~\textsuperscript{$\ast$}} \\Deg &\colorbox{green!10}{0.547 $\rightarrow$ 0.044~\textsuperscript{$\ast$}} &\colorbox{green!10}{0.588 $\rightarrow$ 0.035~\textsuperscript{$\ast$}} &\colorbox{green!10}{0.544 $\rightarrow$ 0.032~\textsuperscript{$\ast$}} \\\hline\multicolumn{3}{l}{\textbf{Axiom 4:} Negatively Changed $\uparrow$} \\ \hline PE &\colorbox{magenta!10}{0.933 $\rightarrow$ 0.636~\textsuperscript{ }} &\colorbox{magenta!10}{1.006 $\rightarrow$ 0.558~\textsuperscript{ }} &\colorbox{magenta!10}{1.252 $\rightarrow$ 0.463~\textsuperscript{ }} \\SE &\colorbox{magenta!20}{4.152 $\rightarrow$ 3.552~\textsuperscript{$\ast$}} &\colorbox{magenta!20}{4.192 $\rightarrow$ 3.409~\textsuperscript{$\ast$}} &\colorbox{magenta!20}{4.830 $\rightarrow$ 3.690~\textsuperscript{$\ast$}} \\PE+M &\colorbox{magenta!20}{1.164 $\rightarrow$ 0.714~\textsuperscript{$\ast$}} &\colorbox{magenta!20}{1.298 $\rightarrow$ 0.748~\textsuperscript{$\ast$}} &\colorbox{magenta!10}{1.689 $\rightarrow$ 0.747~\textsuperscript{ }} \\SE+M &\colorbox{magenta!20}{4.553 $\rightarrow$ 3.690~\textsuperscript{$\ast$}} &\colorbox{magenta!20}{4.653 $\rightarrow$ 3.608~\textsuperscript{$\ast$}} &\colorbox{magenta!20}{5.381 $\rightarrow$ 4.007~\textsuperscript{$\ast$}} \\EigV &\colorbox{magenta!20}{2.593 $\rightarrow$ 1.449~\textsuperscript{$\ast$}} &\colorbox{magenta!20}{2.557 $\rightarrow$ 1.412~\textsuperscript{$\ast$}} &\colorbox{magenta!20}{3.567 $\rightarrow$ 1.449~\textsuperscript{$\ast$}} \\ECC &\colorbox{magenta!20}{0.540 $\rightarrow$ 0.262~\textsuperscript{$\ast$}} &\colorbox{magenta!20}{0.548 $\rightarrow$ 0.220~\textsuperscript{$\ast$}} &\colorbox{magenta!20}{0.707 $\rightarrow$ 0.237~\textsuperscript{$\ast$}} \\Deg &\colorbox{magenta!20}{0.320 $\rightarrow$ 0.128~\textsuperscript{$\ast$}} &\colorbox{magenta!20}{0.320 $\rightarrow$ 0.115~\textsuperscript{$\ast$}} &\colorbox{magenta!20}{0.463 $\rightarrow$ 0.140~\textsuperscript{$\ast$}} \\\hline

\end{tabular}
\caption{
Comparison of changes in average uncertainty values for Axioms 1--4 before (left) and after (right) applying RAG with Llama2-chat.
Colors \colorbox{green!10}{green} and \colorbox{magenta!20}{deep red} indicate significant changes aligning or conflicting with axioms, respectively. 
Color \colorbox{magenta!10}{shallow red} represents non-significant changes conflicting with axioms. Significance is marked by $\ast$.}
\label{tab:axioms_unc_changes_correctness_llama2_pqa}
\shrink
\shrink
\end{table}

\subsection{Axiomatic Evaluation}
The second research question (\textbf{RQ2}) investigates properties (i.e., axioms) of UE methods that guarantee optimal performance, and assesses how these axioms are satisfied by current UE methods. 
%
Tables~\ref{tab:axioms_unc_changes_correctness_llama2_pqa} and \ref{tab:axioms_unc_changes_correctness_llama2} present the change in the average uncertainty value of Llama2-chat, without and with RAG, for Axioms 1--4 using the \emph{Reference-based} implementation. 
The results indicate that Axioms 2 and 4 are largely unmet. 
Furthermore, MARS, although being a state-of-the-art white-box method, does not demonstrate improved axiom compliance.
Similar trends are observed for Mistral and other datasets (see Table~\ref{tab:axioms_rel_correctness_mistral}), underscoring the generalizability of these findings. Additionally, the \emph{Reference-free} implementation of axioms (Table~\ref{tab:axioms_rel_nli_llama}) strongly correlate with the \emph{Reference-based} findings, confirming that UE methods completely fail to satisfy Axioms 2 and 4. This further shows the reliability of reference-free implementation for axiomatic evaluation of UE methods.

To evaluate Axiom 5, we add irrelevant context (\(\text{\textit{Doc}}^{-}\)) for each query. Table~\ref{tab:axiom5_unc_changes_correctness_llama2} shows that only PE+M and SE+M partially satisfy Axiom 5 for Llama2.
For Mistral (Table~\ref{tab:Axiom_irrelevant_mistral}), all methods pass Axiom 5 for \textsc{PopQA} but not for the other datasets. These findings suggest that none of the existing UE methods fully satisfy Axiom 2, 4, and 5.
\begin{table}[t]
\centering \shrink
\setlength{\tabcolsep}{3pt}
\scriptsize
\begin{tabular}{c|ccc}
\hline
\textbf{Unc.} & \textbf{NQ-open} & \textbf{TriviaQA} & \textbf{PopQA} \\
\hline\hline

PE &
\colorbox{magenta!20}{2.072 $\rightarrow$ 2.248~\textsuperscript{$\ast$}} &
\colorbox{magenta!20}{0.872 $\rightarrow$ 1.155~\textsuperscript{$\ast$}} &
\colorbox{magenta!20}{0.897 $\rightarrow$ 0.909~\textsuperscript{$\ast$}} \\
SE &
\colorbox{magenta!20}{5.253 $\rightarrow$ 5.471~\textsuperscript{$\ast$}} &
\colorbox{magenta!20}{3.863 $\rightarrow$ 4.158~\textsuperscript{$\ast$}} &
\colorbox{magenta!20}{3.897 $\rightarrow$ 4.319~\textsuperscript{$\ast$}} \\
PE+M &
\colorbox{green!10}{4.791 $\rightarrow$ 4.805~\textsuperscript{ }} &
\colorbox{magenta!20}{1.415 $\rightarrow$ 1.699~\textsuperscript{$\ast$}} &
\colorbox{magenta!20}{1.031 $\rightarrow$ 1.130~\textsuperscript{$\ast$}} \\
SE+M &
\colorbox{green!10}{7.993 $\rightarrow$ 7.933~\textsuperscript{ }} &
\colorbox{magenta!20}{4.540 $\rightarrow$ 4.817~\textsuperscript{$\ast$}} &
\colorbox{green!10}{4.297 $\rightarrow$ 4.591~\textsuperscript{ }} \\
EigV &
\colorbox{magenta!20}{2.211 $\rightarrow$ 2.446~\textsuperscript{$\ast$}} &
\colorbox{magenta!20}{1.757 $\rightarrow$ 1.870~\textsuperscript{$\ast$}} &
\colorbox{green!10}{2.270 $\rightarrow$ 2.218~\textsuperscript{ }} \\
ECC &
\colorbox{magenta!20}{0.512 $\rightarrow$ 0.625~\textsuperscript{$\ast$}} &
\colorbox{magenta!20}{0.382 $\rightarrow$ 0.448~\textsuperscript{$\ast$}} &
\colorbox{green!10}{0.490 $\rightarrow$ 0.507~\textsuperscript{ }} \\
Deg &
\colorbox{magenta!20}{0.265 $\rightarrow$ 0.333~\textsuperscript{$\ast$}} &
\colorbox{magenta!20}{0.171 $\rightarrow$ 0.211~\textsuperscript{$\ast$}} &
\colorbox{green!10}{0.256 $\rightarrow$ 0.309~\textsuperscript{ }} \\ \hline

\end{tabular}
\caption{
Comparison of changes in average uncertainty values for Axiom 5 before (left) and after (right) applying RAG with Llama2-chat.
Color coding and significance markers follow those in Table~\ref{tab:axioms_unc_changes_correctness_llama2_pqa}.}
\shrink
\shrink
\label{tab:axiom5_unc_changes_correctness_llama2}  
\end{table}

\renewcommand{\arraystretch}{1.25}
\begin{table*}[t]
\centering 
\setlength{\tabcolsep}{3pt}
\scriptsize
\begin{tabular}{l|ccccc|ccccc|ccccc}
\hline
\textbf{UE} &
\multicolumn{5}{c}{\textbf{NQ-open}} & \multicolumn{5}{c}{\textbf{TriviaQA}} & \multicolumn{5}{c}{\textbf{PopQA}} \\ \hline
& \textbf{A1} (\%) & \textbf{A2} (\%) & \textbf{A3} (\%) & \textbf{A4} (\%) & \textbf{AUROC}
& \textbf{A1} (\%) & \textbf{A2} (\%) & \textbf{A3} (\%) & \textbf{A4} (\%) & \textbf{AUROC}
& \textbf{A1} (\%) & \textbf{A2} (\%) & \textbf{A3} (\%) & \textbf{A4} (\%) & \textbf{AUROC} \\ \hline\hline

PE   & 
60.19 & \underline{47.85} & 77.35 & 51.16 & 64.87 &
45.53 & 43.78 & 70.26 & 66.88 & 68.18 &
66.19 & \underline{42.46} & 87.57 & 38.17 & 61.59 \\

+CTI &
61.49 & 44.17 & 76.43 & 53.88 & 65.38 &
46.00 & 43.78 & 69.23 & 68.47 & 69.29 &
\underline{69.63} & 39.73 & 87.95 & 38.17 & 63.04 \\

+NLI &
66.02 & 47.24 & 77.57 & \underline{55.43} & 67.21 &
48.45 & 45.77 & 71.28 & 68.47 & 69.40 &
68.77 & 41.10 & 88.15 & \underline{41.22} & 63.09 \\

+MCH & 
\underline{76.05} & 37.42 & \underline{83.75} & 51.93 & \underline{69.85} &
\underline{51.36} & \underline{49.25} & \underline{74.10} & \underline{69.75} & \underline{71.92} &
69.34 & 39.73 & \underline{89.48} & 39.70 & \underline{64.31} \\ \hline

SE   &
77.35 & 33.75 & 91.53 & 36.05 & 67.49 &
50.14 & 35.82 & \underline{84.62} & 54.78 & 73.44 &
71.92 & 31.51 & 94.07 & 29.01 & 63.79 \\

+CTI & 
77.02 & 25.76 & 89.47 & 40.31 & 67.09 &
56.54 & 39.30 & 79.74 & 56.69 & 72.65 &
\underline{78.51} & 26.03 & 91.21 & 26.72 & 62.58 \\

+NLI & 
79.61 & \underline{40.49} & 86.72 & \underline{50.00} & 69.77 &
68.96 & 46.77 & 80.77 & 62.74 & 74.72 &
71.63 & \underline{38.36} & 92.73 & 41.22 & 67.86 \\

+MCH &
\underline{88.02} & 32.52 & \underline{91.53} & 46.90 & \textbf{\underline{75.88}} &
\underline{73.28} & \underline{49.75} & 82.82 & \underline{67.20} & \textbf{\underline{79.79}} &
77.94 & 31.51 & \underline{94.07} & \underline{41.22} & \textbf{\underline{72.49}} \\ \hline

EigV & 
65.37 & 12.88 & 88.56 & 24.42 & 63.94 &
37.16 & 24.38 & 86.15 & 39.17 & 70.00 &
55.30 & 6.85 & 92.93 & 20.61 & 62.42 \\

+CTI &
77.35 & 20.25 & 90.16 & 34.50 & 66.82 &
66.89 & 30.85 & 86.15 & 48.41 & 72.54 & 
80.80 & 19.18 & 93.50 & 29.77 & 61.51 \\

+NLI & 
80.91 & \underline{27.61} & 91.76 & \underline{35.27} & 69.44 &
60.21 & \underline{41.79} & 87.69 & 51.59 & 73.58 &
73.35 & \underline{35.62} & 95.60 & 38.17 & 67.60 \\

+MCH &
\underline{88.67} & 23.93 & \underline{93.82} & 34.88 & \underline{73.60} &
\underline{74.88} & 40.30 & \underline{90.00} & \underline{55.41} & \underline{78.34} &
\underline{83.09} & 24.66 & \underline{96.75} & \underline{32.82} & \underline{72.18} \\ \hline

ECC &
61.49 & 9.82 & 83.06 & 18.99 & 63.57 &
34.24 & 14.43 & 73.59 & 30.89 & 68.23 & 
52.44 & 6.84 & 87.38 & 18.32 & 62.06 \\

+CTI & 
75.73 & 23.31 & 87.18 & 37.98 & 67.37 &
65.47 & 31.84 & 77.69 & 53.19 & 69.92 & 
78.80 & 23.29 & 90.82 & 34.35 & 61.75 \\

+NLI & 
78.64 & \underline{32.52} & 87.18 & \underline{42.64} & 68.96 &
58.04 & \underline{42.79} & 77.44 & \underline{59.87} & 71.31 & 
71.35 & \underline{32.88} & 92.16 & \underline{42.75} & 66.44 \\

+MCH & 
\underline{86.08} & 26.99 & \underline{89.93} & 39.54 & \underline{71.81} &
\underline{72.44} & 41.29 & \underline{82.31} & 58.92 & \underline{74.94} & 
\underline{79.37} & 21.92 & \underline{94.84} & 35.87 & \underline{71.39} \\ \hline





\hline
\end{tabular}
\caption{Percentage of samples passing the axioms before and after calibration for Contriver with Llama2-chat. The results show that as the number of samples passing the axioms increases, the AUROC also improves. Bold values indicate the best performance for each dataset, while underlined values represent the best performance achieved by a UE method and its calibrated variants.}
\label{tab:axioms_calibration_llama2}
\shrink
\end{table*}

\begin{figure*}[t]
  \centering
  \includegraphics[width=0.98\textwidth]{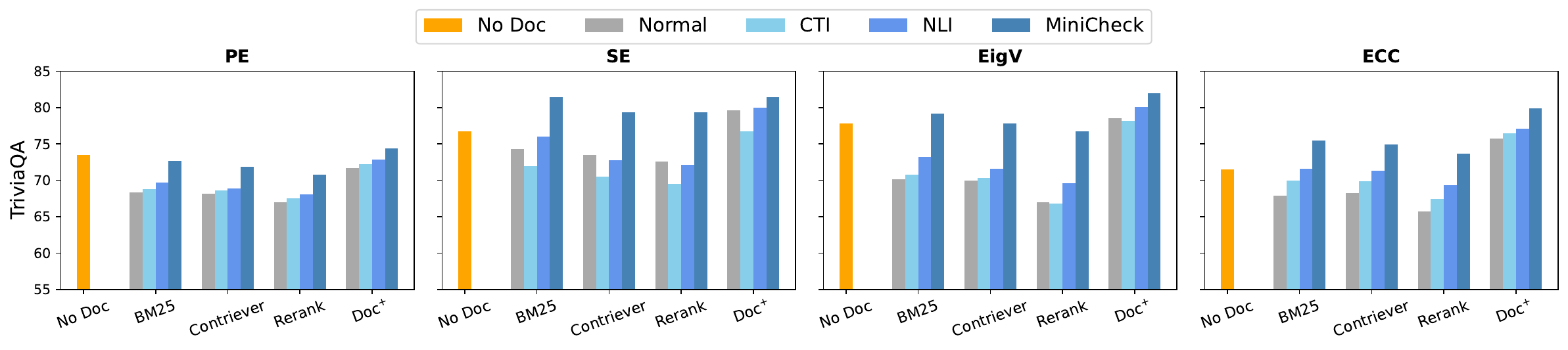}
  \shrink
  \caption{Comparison of AUROC between the no-RAG and calibrated RAG settings for Llama2-chat for TriviaQA. AUROC improves significantly, either surpassing the no-RAG setting or reducing the gap between them. 
  }
  \label{fig:auroc_cal_llama2}
  \shrink
\end{figure*}
\subsection{Axiomatic Calibration}
Our third research question (\textbf{RQ3}) examines how our axiomatic framework can lead to designing an optimal UE method. 
Tables~\ref{tab:axioms_calibration_llama2} and \ref{tab:axioms_sample_percentage_mistral} present AUROC and percentage of samples passing the axioms 1--4 before and after applying our calibration method. Axiom 5 is not assessed, as retrievers tend to retrieve relevant documents. We perform the experiments on four representative (and not cherry-picked) UE methods, as the results generalize to other methods as well.
The calibration function is implemented using the three models described in Section~\ref{sec:calibration_function:inst}, and Contriever is employed for RAG.
 
 The results show that calibration MiniCheck outperforms all implementations, improving percentages of all axioms for EigV and ECC and for most axioms in open-box methods. 
 Most importantly, the results show as the percentage of samples satisfying the axioms increases, the AUROC improves, showing the empirical validity of our axioms in improving  UE methods.
Moreover, Figures~\ref{fig:auroc_cal_llama2} and \ref{fig:auroc_cal_nq_popqa_llama2} show that after calibration, the RAG AUROC becomes comparable to or even better than the \textit{No Doc} baseline, suggesting that our calibration method successfully compensates for the inefficiencies of existing UE methods in RAG.

%% file: sections/7_conclusions.tex
\section{Discussion and Conclusions}~\label{sec:conclusions}
\shrink

In this paper, we examined existing uncertainty estimation (UE) for the RAG setup and showed they systematically generated low uncertainty values in the RAG setup without considering the relevance of the given context to the query.
We further proposed an axiomatic evaluation framework for UE in the RAG setup and defined five formal constraints that a UE method should satisfy when processing both parametric and non-parametric knowledge. 
These axioms were empirically validated across multiple representative datasets, UE methods, and LLMs. Our results showed that none of the existing UE methods pass all the axiom, pinpointing the problem in these methods. We further derived a calibration function for adjusting UE methods in the RAG setup and improvements in both axiomatic evaluation and correlation with correctness.
Future work includes developing a UE method designed to naturally conform to the established axioms. Another direction is assessing these axioms in long-form responses and uncertainty-based applications, such as Active RAG.


%% file: sections/8_limitations.tex
\section*{Limitations}

\noindent
\textbf{Axiomatic Uncertainty Estimator.}
In this study, we evaluate existing uncertainty estimation (UE) methods within the RAG setup and delineate the optimal behaviors that these methods should exhibit. Although we introduce a calibration function in Section~\ref{sec:calibration_function}, it may be more effective to develop an axiomatic UE model that inherently adheres to the prescribed axioms. Future research should leverage these principles in the construction of UE methods.

\medskip 
\noindent\textbf{Comprehensiveness of the Axioms.} 
As discussed in Section~\ref{sec:axiom}, while our current axioms address most cases, additional axioms may be needed to cover all sample types. For example, consider when an LLM produces a different output after incorporating a context, and both the initial and augmented responses contradict the context. In this scenario, our framework does not specify a change in uncertainty, though supplementary axioms might address this gap. Future research should develop axioms for such cases.

\medskip
\noindent\textbf{Scalability and Applications.} 
We investigated the impact of incorporating context into the input prompt on uncertainty measures. However, we did not explore other input modalities, such as multi-modal RAG, or alternative response formats, such as long-form responses, each of which presents unique challenges. Furthermore, applications of uncertainty estimation, such as Adaptive RAG~\cite{Cheng24Unified, Tao24When}, hallucination detection~\cite{Geng24Survey}, reasoning monitoring~\cite{Yin24Reasoning}, and LLM-as-Judgment~\cite{Lee24Are, dietz2025llm}, fall outside the scope of this study. Future research should extend these findings to encompass diverse input types, response formats, and UE applications.



%% file: sections/appendix.tex
\section*{Appendix}

\begin{figure*}[t]
  \centering
  \includegraphics[width=0.99\textwidth]{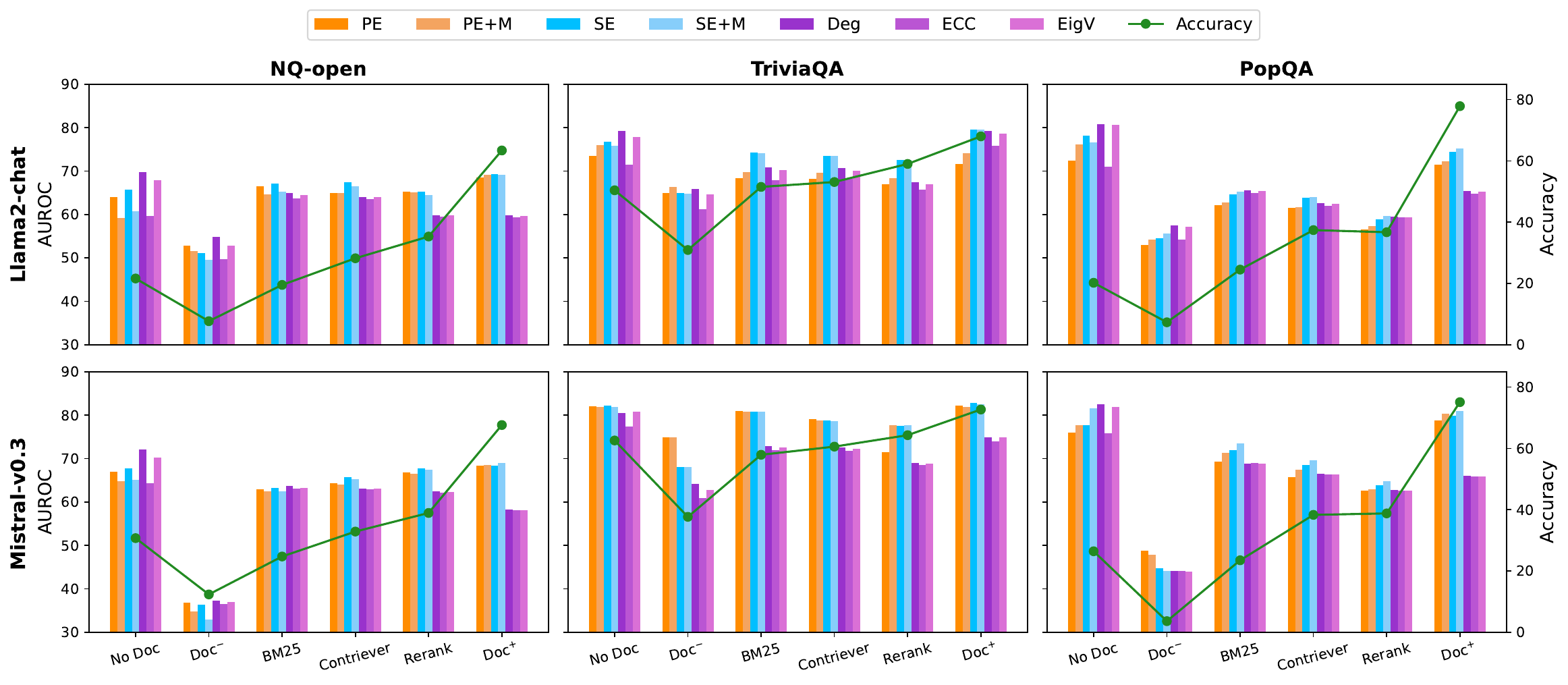}
  \shrink
  \caption{Comparison of AUROC between no-RAG and RAG settings.}
  \label{fig:auroc_all_1}
\end{figure*}
\begin{table*}[t]
\tiny
\centering
\setlength{\tabcolsep}{0.9pt}
\begin{tabular}{@{~}cc|cccccc|cccccc|cccccc}
\hline
\textbf{LM} & \textbf{Unc.} &
\multicolumn{6}{c}{\textbf{NQ-open}} & \multicolumn{6}{c}{\textbf{TriviaQA}}  & \multicolumn{6}{c}{\textbf{PopQA}} \\ \hline

\hline
& &
\textbf{$\text{No Doc}$} & \textbf{$\text{Doc}^{-}$} & \textbf{BM25} & \textbf{Cont.} & \textbf{ReRa.} & \textbf{$\text{Doc}^{+}$} &
\textbf{$\text{No Doc}$} & \textbf{$\text{Doc}^{-}$} & \textbf{BM25} & \textbf{Cont.} & \textbf{ReRa.} & \textbf{$\text{Doc}^{+}$} &
\textbf{$\text{No Doc}$} & \textbf{$\text{Doc}^{-}$} & \textbf{BM25} & \textbf{Cont.} & \textbf{ReRa.} & \textbf{$\text{Doc}^{+}$} \\
\hline\hline

\multirow{8}{*}{\rotatebox[origin=c]{90}{Llama2-chat}} &
PE   &
\colorbox{orange!50}{1.98~\textsuperscript{ }} &
\colorbox{orange!40}{1.92~\textsuperscript{ }} &
\colorbox{orange!30}{1.53~\textsuperscript{$\ast$}} &
\colorbox{orange!20}{1.41~\textsuperscript{$\ast$}} &
\colorbox{orange!10}{1.31~\textsuperscript{$\ast$}} &
\colorbox{orange!5}{1.19~\textsuperscript{$\ast$}} &
\colorbox{orange!40}{1.14~\textsuperscript{ }} &
\colorbox{orange!50}{1.42~\textsuperscript{$\ast$}} &
\colorbox{orange!30}{1.05~\textsuperscript{$\ast$}} &
\colorbox{orange!20}{1.03~\textsuperscript{ }} &
\colorbox{orange!5}{0.90~\textsuperscript{$\ast$}} &
\colorbox{orange!10}{0.96~\textsuperscript{$\ast$}} &
\colorbox{orange!50}{1.29~\textsuperscript{ }} &
\colorbox{orange!40}{1.11~\textsuperscript{$\ast$}} &
\colorbox{orange!30}{0.54~\textsuperscript{$\ast$}} &
\colorbox{orange!20}{0.46~\textsuperscript{$\ast$}} &
\colorbox{orange!10}{0.35~\textsuperscript{$\ast$}} &
\colorbox{orange!5}{0.34~\textsuperscript{$\ast$}} \\

& SE   & 
\colorbox{orange!50}{5.40~\textsuperscript{ }} &
\colorbox{orange!40}{5.09~\textsuperscript{$\ast$}} &
\colorbox{orange!30}{4.29~\textsuperscript{$\ast$}} &
\colorbox{orange!20}{4.20~\textsuperscript{$\ast$}} &
\colorbox{orange!10}{3.99~\textsuperscript{$\ast$}} &
\colorbox{orange!5}{3.88~\textsuperscript{$\ast$}} &
\colorbox{orange!40}{4.39~\textsuperscript{ }} &
\colorbox{orange!50}{4.48~\textsuperscript{$\ast$}} &
\colorbox{orange!30}{3.89~\textsuperscript{$\ast$}} &
\colorbox{orange!20}{3.85~\textsuperscript{$\ast$}} &
\colorbox{orange!5}{3.66~\textsuperscript{$\ast$}} &
\colorbox{orange!10}{3.73~\textsuperscript{$\ast$}} &
\colorbox{orange!50}{4.86~\textsuperscript{ }} &
\colorbox{orange!40}{4.37~\textsuperscript{$\ast$}} &
\colorbox{orange!30}{3.45~\textsuperscript{$\ast$}} &
\colorbox{orange!20}{3.30~\textsuperscript{$\ast$}} &
\colorbox{orange!5}{3.13~\textsuperscript{$\ast$}} &
\colorbox{orange!10}{3.19~\textsuperscript{$\ast$}} \\

& PEM &
\colorbox{orange!50}{3.90~\textsuperscript{ }} &
\colorbox{orange!40}{3.89~\textsuperscript{ }} &
\colorbox{orange!30}{3.33~\textsuperscript{$\ast$}} &
\colorbox{orange!20}{3.26~\textsuperscript{$\ast$}} &
\colorbox{orange!10}{3.12~\textsuperscript{$\ast$}} &
\colorbox{orange!5}{2.97~\textsuperscript{$\ast$}} &
\colorbox{orange!40}{1.74~\textsuperscript{ }} &
\colorbox{orange!50}{2.06~\textsuperscript{$\ast$}} &
\colorbox{orange!30}{1.64~\textsuperscript{ }} &
\colorbox{orange!30}{1.64~\textsuperscript{ }} &
\colorbox{orange!5}{1.46~\textsuperscript{$\ast$}} &
\colorbox{orange!10}{1.51~\textsuperscript{$\ast$}} &
\colorbox{orange!50}{1.59~\textsuperscript{ }} &
\colorbox{orange!40}{1.34~\textsuperscript{$\ast$}} &
\colorbox{orange!30}{0.65~\textsuperscript{$\ast$}} &
\colorbox{orange!20}{0.55~\textsuperscript{$\ast$}} &
\colorbox{orange!5}{0.44~\textsuperscript{$\ast$}} &
\colorbox{orange!10}{0.45~\textsuperscript{$\ast$}} \\

& SEM &
\colorbox{orange!50}{7.41~\textsuperscript{ }} &
\colorbox{orange!40}{6.93~\textsuperscript{$\ast$}} &
\colorbox{orange!30}{5.97~\textsuperscript{$\ast$}} &
\colorbox{orange!20}{5.88~\textsuperscript{$\ast$}} &
\colorbox{orange!10}{5.62~\textsuperscript{$\ast$}} &
\colorbox{orange!5}{5.49~\textsuperscript{$\ast$}} &
\colorbox{orange!40}{5.16~\textsuperscript{ }} &
\colorbox{orange!50}{5.21~\textsuperscript{ }} &
\colorbox{orange!30}{4.51~\textsuperscript{$\ast$}} &
\colorbox{orange!30}{4.50~\textsuperscript{$\ast$}} &
\colorbox{orange!5}{4.22~\textsuperscript{$\ast$}} &
\colorbox{orange!10}{4.30~\textsuperscript{$\ast$}} &
\colorbox{orange!50}{5.38~\textsuperscript{ }} &
\colorbox{orange!40}{4.71~\textsuperscript{$\ast$}} &
\colorbox{orange!30}{3.62~\textsuperscript{$\ast$}} &
\colorbox{orange!20}{3.43~\textsuperscript{$\ast$}} &
\colorbox{orange!5}{3.23~\textsuperscript{$\ast$}} &
\colorbox{orange!10}{3.27~\textsuperscript{$\ast$}} \\

& Deg & 
\colorbox{orange!50}{0.52~\textsuperscript{ }} &
\colorbox{orange!40}{0.36~\textsuperscript{$\ast$}} &
\colorbox{orange!30}{0.16~\textsuperscript{$\ast$}} &
\colorbox{orange!20}{0.13~\textsuperscript{$\ast$}} &
\colorbox{orange!10}{0.09~\textsuperscript{$\ast$}} &
\colorbox{orange!5}{0.07~\textsuperscript{$\ast$}} & 
\colorbox{orange!50}{0.31~\textsuperscript{ }} &
\colorbox{orange!40}{0.29~\textsuperscript{$\ast$}} &
\colorbox{orange!30}{0.17~\textsuperscript{$\ast$}} &
\colorbox{orange!20}{0.15~\textsuperscript{$\ast$}} &
\colorbox{orange!10}{0.11~\textsuperscript{$\ast$}} &
\colorbox{orange!5}{0.16~\textsuperscript{$\ast$}} &
\colorbox{orange!50}{0.52~\textsuperscript{ }} &
\colorbox{orange!40}{0.32~\textsuperscript{$\ast$}} &
\colorbox{orange!30}{0.12~\textsuperscript{$\ast$}} &
\colorbox{orange!20}{0.09~\textsuperscript{$\ast$}} &
\colorbox{orange!10}{0.06~\textsuperscript{$\ast$}} &
\colorbox{orange!5}{0.05~\textsuperscript{$\ast$}} \\

& ECC &
\colorbox{orange!50}{0.64~\textsuperscript{ }} &
\colorbox{orange!40}{0.60~\textsuperscript{$\ast$}} &
\colorbox{orange!30}{0.29~\textsuperscript{$\ast$}} &
\colorbox{orange!20}{0.23~\textsuperscript{$\ast$}} &
\colorbox{orange!10}{0.17~\textsuperscript{$\ast$}} &
\colorbox{orange!5}{0.14~\textsuperscript{$\ast$}} & 
\colorbox{orange!50}{0.56~\textsuperscript{ }} &
\colorbox{orange!40}{0.53~\textsuperscript{$\ast$}} &
\colorbox{orange!30}{0.33~\textsuperscript{$\ast$}} &
\colorbox{orange!10}{0.29~\textsuperscript{$\ast$}} &
\colorbox{orange!5}{0.23~\textsuperscript{$\ast$}} &
\colorbox{orange!20}{0.31~\textsuperscript{$\ast$}} &
\colorbox{orange!50}{0.71~\textsuperscript{ }} &
\colorbox{orange!40}{0.54~\textsuperscript{$\ast$}} &
\colorbox{orange!30}{0.22~\textsuperscript{$\ast$}} &
\colorbox{orange!20}{0.17~\textsuperscript{$\ast$}} &
\colorbox{orange!10}{0.12~\textsuperscript{$\ast$}} &
\colorbox{orange!5}{0.10~\textsuperscript{$\ast$}} \\

& EigV & 
\colorbox{orange!50}{3.06~\textsuperscript{ }} &
\colorbox{orange!40}{2.48~\textsuperscript{$\ast$}} &
\colorbox{orange!30}{1.57~\textsuperscript{$\ast$}} &
\colorbox{orange!20}{1.42~\textsuperscript{$\ast$}} &
\colorbox{orange!10}{1.28~\textsuperscript{$\ast$}} &
\colorbox{orange!5}{1.21~\textsuperscript{$\ast$}} &
\colorbox{orange!50}{2.52~\textsuperscript{ }} &
\colorbox{orange!40}{2.21~\textsuperscript{$\ast$}} &
\colorbox{orange!20}{1.65~\textsuperscript{$\ast$}} &
\colorbox{orange!10}{1.57~\textsuperscript{$\ast$}} &
\colorbox{orange!5}{1.41~\textsuperscript{$\ast$}} &
\colorbox{orange!30}{1.68~\textsuperscript{$\ast$}} &
\colorbox{orange!50}{4.25~\textsuperscript{ }} &
\colorbox{orange!40}{2.28~\textsuperscript{$\ast$}} &
\colorbox{orange!30}{1.42~\textsuperscript{$\ast$}} &
\colorbox{orange!20}{1.31~\textsuperscript{$\ast$}} &
\colorbox{orange!10}{1.18~\textsuperscript{$\ast$}} &
\colorbox{orange!5}{1.17~\textsuperscript{$\ast$}} \\

\hline
\multirow{8}{*}{\rotatebox[origin=c]{90}{Mistral-v0.3}} &
PE   &
\colorbox{orange!50}{1.98~\textsuperscript{ }} &
\colorbox{orange!5}{1.28~\textsuperscript{$\ast$}} &
\colorbox{orange!30}{1.40~\textsuperscript{$\ast$}} &
\colorbox{orange!40}{1.46~\textsuperscript{$\ast$}} &
\colorbox{orange!20}{1.39~\textsuperscript{$\ast$}} &
\colorbox{orange!10}{1.32~\textsuperscript{$\ast$}} &
\colorbox{orange!40}{0.96~\textsuperscript{ }} &
\colorbox{orange!50}{1.08~\textsuperscript{$\ast$}} &
\colorbox{orange!30}{0.83~\textsuperscript{$\ast$}} &
\colorbox{orange!20}{0.81~\textsuperscript{$\ast$}} &
\colorbox{orange!5}{0.72~\textsuperscript{$\ast$}} &
\colorbox{orange!10}{0.74~\textsuperscript{$\ast$}} & 
\colorbox{orange!50}{1.51~\textsuperscript{ }} &
\colorbox{orange!40}{0.94~\textsuperscript{$\ast$}} &
\colorbox{orange!30}{0.84~\textsuperscript{$\ast$}} &
\colorbox{orange!20}{0.69~\textsuperscript{$\ast$}} &
\colorbox{orange!10}{0.62~\textsuperscript{$\ast$}} &
\colorbox{orange!5}{0.51~\textsuperscript{$\ast$}} \\

& SE   &
\colorbox{orange!50}{5.61~\textsuperscript{ }} &
\colorbox{orange!40}{4.37~\textsuperscript{$\ast$}} &
\colorbox{orange!20}{4.32~\textsuperscript{$\ast$}} &
\colorbox{orange!30}{4.33~\textsuperscript{$\ast$}} &
\colorbox{orange!10}{4.19~\textsuperscript{$\ast$}} &
\colorbox{orange!5}{4.05~\textsuperscript{$\ast$}} &
\colorbox{orange!50}{4.29~\textsuperscript{ }} &
\colorbox{orange!40}{4.27~\textsuperscript{$\ast$}} &
\colorbox{orange!30}{3.76~\textsuperscript{$\ast$}} &
\colorbox{orange!20}{3.74~\textsuperscript{$\ast$}} &
\colorbox{orange!5}{3.57~\textsuperscript{$\ast$}} &
\colorbox{orange!10}{3.67~\textsuperscript{$\ast$}} &
\colorbox{orange!50}{5.66~\textsuperscript{ }} &
\colorbox{orange!40}{3.73~\textsuperscript{$\ast$}} &
\colorbox{orange!30}{3.68~\textsuperscript{$\ast$}} &
\colorbox{orange!20}{3.53~\textsuperscript{$\ast$}} &
\colorbox{orange!10}{3.41~\textsuperscript{$\ast$}} &
\colorbox{orange!5}{3.26~\textsuperscript{$\ast$}} \\

& PEM &
\colorbox{orange!50}{4.25~\textsuperscript{ }} &
\colorbox{orange!5}{2.51~\textsuperscript{$\ast$}}&
\colorbox{orange!10}{3.29~\textsuperscript{$\ast$}} &
\colorbox{orange!40}{3.61~\textsuperscript{$\ast$}} &
\colorbox{orange!30}{3.48~\textsuperscript{$\ast$}} &
\colorbox{orange!20}{3.36~\textsuperscript{$\ast$}} & 
\colorbox{orange!40}{1.73~\textsuperscript{ }} &
\colorbox{orange!50}{1.88~\textsuperscript{$\ast$}} &
\colorbox{orange!20}{1.51~\textsuperscript{$\ast$}} &
\colorbox{orange!30}{1.54~\textsuperscript{$\ast$}} &
\colorbox{orange!5}{1.36~\textsuperscript{$\ast$}} &
\colorbox{orange!10}{1.41~\textsuperscript{$\ast$}} &
\colorbox{orange!50}{2.35~\textsuperscript{ }} &
\colorbox{orange!40}{1.42~\textsuperscript{$\ast$}}& 
\colorbox{orange!30}{1.26~\textsuperscript{$\ast$}} & 
\colorbox{orange!20}{1.05~\textsuperscript{$\ast$}} & 
\colorbox{orange!10}{0.92~\textsuperscript{$\ast$}} & 
\colorbox{orange!5}{0.80~\textsuperscript{$\ast$}} \\

& SEM &
\colorbox{orange!50}{7.65~\textsuperscript{ }} &
\colorbox{orange!5}{5.42~\textsuperscript{$\ast$}}&
\colorbox{orange!20}{5.94~\textsuperscript{$\ast$}} &
\colorbox{orange!40}{6.19~\textsuperscript{$\ast$}} &
\colorbox{orange!30}{6.01~\textsuperscript{$\ast$}} &
\colorbox{orange!10}{5.85~\textsuperscript{$\ast$}} &
\colorbox{orange!50}{4.99~\textsuperscript{ }} &
\colorbox{orange!40}{4.98~\textsuperscript{$\ast$}} &
\colorbox{orange!20}{4.35~\textsuperscript{$\ast$}} &
\colorbox{orange!30}{4.37~\textsuperscript{$\ast$}} &
\colorbox{orange!5}{4.12~\textsuperscript{$\ast$}} &
\colorbox{orange!10}{4.27~\textsuperscript{$\ast$}} & 
\colorbox{orange!50}{6.47~\textsuperscript{ }} &
\colorbox{orange!40}{4.05~\textsuperscript{$\ast$}} &
\colorbox{orange!30}{3.98~\textsuperscript{$\ast$}} &
\colorbox{orange!20}{3.77~\textsuperscript{$\ast$}} &
\colorbox{orange!10}{3.60~\textsuperscript{$\ast$}} &
\colorbox{orange!5}{3.45~\textsuperscript{$\ast$}} \\

& Deg  &
\colorbox{orange!50}{0.37~\textsuperscript{ }} &
\colorbox{orange!40}{0.16~\textsuperscript{$\ast$}}&
\colorbox{orange!30}{0.13~\textsuperscript{$\ast$}} &
\colorbox{orange!20}{0.10~\textsuperscript{$\ast$}} &
\colorbox{orange!10}{0.07~\textsuperscript{$\ast$}} &
\colorbox{orange!5}{0.05~\textsuperscript{$\ast$}} &
\colorbox{orange!50}{0.20~\textsuperscript{ }} &
\colorbox{orange!30}{0.18~\textsuperscript{$\ast$}} &
\colorbox{orange!20}{0.10~\textsuperscript{$\ast$}} &
\colorbox{orange!20}{0.10~\textsuperscript{$\ast$}} &
\colorbox{orange!5}{0.07~\textsuperscript{$\ast$}} &
\colorbox{orange!40}{0.19~\textsuperscript{$\ast$}} &
\colorbox{orange!50}{0.48~\textsuperscript{ }} &
\colorbox{orange!20}{0.05~\textsuperscript{$\ast$}} &
\colorbox{orange!40}{0.07~\textsuperscript{$\ast$}} &
\colorbox{orange!30}{0.06~\textsuperscript{$\ast$}} &
\colorbox{orange!20}{0.05~\textsuperscript{$\ast$}} &
\colorbox{orange!5}{0.03~\textsuperscript{$\ast$}} \\

& ECC  &
\colorbox{orange!50}{0.54~\textsuperscript{ }} &
\colorbox{orange!40}{0.20~\textsuperscript{$\ast$}}&
\colorbox{orange!30}{0.18~\textsuperscript{$\ast$}} &
\colorbox{orange!20}{0.15~\textsuperscript{$\ast$}} &
\colorbox{orange!10}{0.11~\textsuperscript{$\ast$}} &
\colorbox{orange!5}{0.08~\textsuperscript{$\ast$}} &
\colorbox{orange!50}{0.37~\textsuperscript{ }} &
\colorbox{orange!40}{0.32~\textsuperscript{$\ast$}} &
\colorbox{orange!20}{0.17~\textsuperscript{$\ast$}} &
\colorbox{orange!30}{0.19~\textsuperscript{$\ast$}} &
\colorbox{orange!5}{0.13~\textsuperscript{$\ast$}} &
\colorbox{orange!20}{0.17~\textsuperscript{$\ast$}} &
\colorbox{orange!50}{0.68~\textsuperscript{ }} &
\colorbox{orange!10}{0.03~\textsuperscript{$\ast$}}&
\colorbox{orange!40}{0.08~\textsuperscript{$\ast$}} &
\colorbox{orange!40}{0.08~\textsuperscript{$\ast$}} &
\colorbox{orange!20}{0.05~\textsuperscript{$\ast$}} &
\colorbox{orange!5}{0.04~\textsuperscript{$\ast$}} \\

& EigV &
\colorbox{orange!50}{2.83~\textsuperscript{ }} &
\colorbox{orange!40}{1.49~\textsuperscript{$\ast$}}&
\colorbox{orange!30}{1.40~\textsuperscript{$\ast$}} &
\colorbox{orange!10}{1.32~\textsuperscript{$\ast$}} &
\colorbox{orange!5}{1.23~\textsuperscript{$\ast$}} &
\colorbox{orange!20}{1.37~\textsuperscript{$\ast$}} & 
\colorbox{orange!50}{2.04~\textsuperscript{ }} &
\colorbox{orange!40}{1.65~\textsuperscript{$\ast$}} &
\colorbox{orange!10}{1.36~\textsuperscript{$\ast$}} &
\colorbox{orange!20}{1.39~\textsuperscript{$\ast$}} &
\colorbox{orange!5}{1.25~\textsuperscript{$\ast$}} &
\colorbox{orange!30}{1.41~\textsuperscript{$\ast$}} &
\colorbox{orange!50}{4.18~\textsuperscript{ }} &
\colorbox{orange!10}{1.08~\textsuperscript{$\ast$}}&
\colorbox{orange!30}{1.16~\textsuperscript{$\ast$}} &
\colorbox{orange!40}{1.17~\textsuperscript{$\ast$}} &
\colorbox{orange!20}{1.11~\textsuperscript{$\ast$}} &
\colorbox{orange!10}{1.08~\textsuperscript{$\ast$}} \\


\hline
\end{tabular}
\shrink
\caption{Average uncertainty values for various settings. Lighter colors indicate lower uncertainty. Statistically significant differences are compared to \textit{No Doc} are marked with $\ast$.}
\label{tab:uncertainty_value_rag_methods_1}
\shrink
\end{table*}

\begin{table*}[t]
\centering
\setlength{\tabcolsep}{0.0pt}
\tiny
\begin{tabular}{c|ccc|ccc|ccc}
\hline
\textbf{UE} &
\multicolumn{3}{c}{\textbf{NQ-open}} & \multicolumn{3}{c}{\textbf{TriviaQA}} & \multicolumn{3}{c}{\textbf{PopQA}} \\ \hline %
& \textbf{BM25} & \textbf{Contriever} & \textbf{$\text{Doc}^{+}$} & 
\textbf{BM25} & \textbf{Contriever} & \textbf{$\text{Doc}^{+}$} & 
\textbf{BM25} & \textbf{Contriever} & \textbf{$\text{Doc}^{+}$} \\
\hline\hline

\multicolumn{6}{l}{\textbf{Axiom 1:} Positively Consistent $\downarrow$} \\ \hline PE &\colorbox{green!10}{1.445 $\rightarrow$ 1.194~\textsuperscript{$\ast$}} &\colorbox{green!10}{1.535 $\rightarrow$ 1.216~\textsuperscript{$\ast$}} &\colorbox{green!10}{1.549 $\rightarrow$ 1.159~\textsuperscript{$\ast$}} &\colorbox{magenta!20}{0.700 $\rightarrow$ 0.753~\textsuperscript{$\ast$}} &\colorbox{magenta!20}{0.718 $\rightarrow$ 0.743~\textsuperscript{$\ast$}} &\colorbox{green!10}{0.731 $\rightarrow$ 0.724~\textsuperscript{$\ast$}} &\colorbox{green!10}{0.735 $\rightarrow$ 0.419~\textsuperscript{$\ast$}} &\colorbox{green!10}{0.735 $\rightarrow$ 0.408~\textsuperscript{$\ast$}} &\colorbox{green!10}{1.242 $\rightarrow$ 0.340~\textsuperscript{$\ast$}} \\SE &\colorbox{green!10}{4.656 $\rightarrow$ 3.933~\textsuperscript{$\ast$}} &\colorbox{green!10}{4.756 $\rightarrow$ 3.907~\textsuperscript{$\ast$}} &\colorbox{green!10}{4.800 $\rightarrow$ 3.823~\textsuperscript{$\ast$}} &\colorbox{magenta!5}{3.644 $\rightarrow$ 3.412~\textsuperscript{ }} &\colorbox{green!10}{3.664 $\rightarrow$ 3.424~\textsuperscript{$\ast$}} &\colorbox{green!10}{3.738 $\rightarrow$ 3.388~\textsuperscript{$\ast$}} &\colorbox{green!10}{3.781 $\rightarrow$ 3.205~\textsuperscript{$\ast$}} &\colorbox{green!10}{3.791 $\rightarrow$ 3.158~\textsuperscript{$\ast$}} &\colorbox{green!10}{4.682 $\rightarrow$ 3.113~\textsuperscript{$\ast$}} \\PE+M &\colorbox{magenta!5}{3.389 $\rightarrow$ 3.124~\textsuperscript{ }} &\colorbox{green!10}{3.412 $\rightarrow$ 3.052~\textsuperscript{$\ast$}} &\colorbox{green!10}{3.437 $\rightarrow$ 3.069~\textsuperscript{$\ast$}} &\colorbox{magenta!20}{1.051 $\rightarrow$ 1.110~\textsuperscript{$\ast$}} &\colorbox{magenta!20}{1.131 $\rightarrow$ 1.178~\textsuperscript{$\ast$}} &\colorbox{magenta!5}{1.141 $\rightarrow$ 1.120~\textsuperscript{ }} &\colorbox{green!10}{0.896 $\rightarrow$ 0.483~\textsuperscript{$\ast$}} &\colorbox{green!10}{0.881 $\rightarrow$ 0.458~\textsuperscript{$\ast$}} &\colorbox{green!10}{1.530 $\rightarrow$ 0.406~\textsuperscript{$\ast$}} \\SE+M &\colorbox{green!10}{6.640 $\rightarrow$ 5.778~\textsuperscript{$\ast$}} &\colorbox{green!10}{6.705 $\rightarrow$ 5.632~\textsuperscript{$\ast$}} &\colorbox{green!10}{6.740 $\rightarrow$ 5.667~\textsuperscript{$\ast$}} &\colorbox{green!10}{4.142 $\rightarrow$ 3.832~\textsuperscript{$\ast$}} &\colorbox{green!10}{4.212 $\rightarrow$ 3.898~\textsuperscript{$\ast$}} &\colorbox{green!10}{4.293 $\rightarrow$ 3.824~\textsuperscript{$\ast$}} &\colorbox{green!10}{4.102 $\rightarrow$ 3.286~\textsuperscript{$\ast$}} &\colorbox{green!10}{4.091 $\rightarrow$ 3.248~\textsuperscript{$\ast$}} &\colorbox{green!10}{5.146 $\rightarrow$ 3.173~\textsuperscript{$\ast$}} \\EigV &\colorbox{green!10}{2.030 $\rightarrow$ 1.270~\textsuperscript{$\ast$}} &\colorbox{green!10}{2.129 $\rightarrow$ 1.189~\textsuperscript{$\ast$}} &\colorbox{green!10}{2.166 $\rightarrow$ 1.112~\textsuperscript{$\ast$}} &\colorbox{green!10}{1.622 $\rightarrow$ 1.318~\textsuperscript{$\ast$}} &\colorbox{green!10}{1.617 $\rightarrow$ 1.234~\textsuperscript{$\ast$}} &\colorbox{green!10}{1.679 $\rightarrow$ 1.254~\textsuperscript{$\ast$}} &\colorbox{green!10}{1.951 $\rightarrow$ 1.166~\textsuperscript{$\ast$}} &\colorbox{green!10}{2.025 $\rightarrow$ 1.143~\textsuperscript{$\ast$}} &\colorbox{green!10}{4.074 $\rightarrow$ 1.078~\textsuperscript{$\ast$}} \\ECC &\colorbox{green!10}{0.479 $\rightarrow$ 0.149~\textsuperscript{$\ast$}} &\colorbox{green!10}{0.538 $\rightarrow$ 0.120~\textsuperscript{$\ast$}} &\colorbox{green!10}{0.557 $\rightarrow$ 0.071~\textsuperscript{$\ast$}} &\colorbox{green!10}{0.346 $\rightarrow$ 0.228~\textsuperscript{$\ast$}} &\colorbox{green!10}{0.338 $\rightarrow$ 0.169~\textsuperscript{$\ast$}} &\colorbox{green!10}{0.367 $\rightarrow$ 0.180~\textsuperscript{$\ast$}} &\colorbox{green!10}{0.417 $\rightarrow$ 0.110~\textsuperscript{$\ast$}} &\colorbox{green!10}{0.426 $\rightarrow$ 0.094~\textsuperscript{$\ast$}} &\colorbox{green!10}{0.710 $\rightarrow$ 0.055~\textsuperscript{$\ast$}} \\Deg &\colorbox{green!10}{0.227 $\rightarrow$ 0.084~\textsuperscript{$\ast$}} &\colorbox{green!10}{0.262 $\rightarrow$ 0.061~\textsuperscript{$\ast$}} &\colorbox{green!10}{0.270 $\rightarrow$ 0.035~\textsuperscript{$\ast$}} &\colorbox{green!10}{0.144 $\rightarrow$ 0.087~\textsuperscript{$\ast$}} &\colorbox{green!10}{0.142 $\rightarrow$ 0.066~\textsuperscript{$\ast$}} &\colorbox{green!10}{0.155 $\rightarrow$ 0.067~\textsuperscript{$\ast$}} &\colorbox{green!10}{0.220 $\rightarrow$ 0.048~\textsuperscript{$\ast$}} &\colorbox{green!10}{0.230 $\rightarrow$ 0.043~\textsuperscript{$\ast$}} &\colorbox{green!10}{0.496 $\rightarrow$ 0.022~\textsuperscript{$\ast$}} \\\hline\multicolumn{6}{l}{\textbf{Axiom 2:} Negatively Consistent $\uparrow$} \\ \hline PE &\colorbox{magenta!10}{2.317 $\rightarrow$ 2.261~\textsuperscript{ }} &\colorbox{magenta!10}{2.230 $\rightarrow$ 2.153~\textsuperscript{ }} &\colorbox{magenta!10}{2.232 $\rightarrow$ 2.194~\textsuperscript{ }} &\colorbox{magenta!10}{1.543 $\rightarrow$ 1.478~\textsuperscript{ }} &\colorbox{magenta!10}{1.534 $\rightarrow$ 1.438~\textsuperscript{ }} &\colorbox{magenta!5}{1.495 $\rightarrow$ 1.528~\textsuperscript{ }} &\colorbox{magenta!10}{1.068 $\rightarrow$ 0.746~\textsuperscript{ }} &\colorbox{magenta!10}{0.820 $\rightarrow$ 0.593~\textsuperscript{ }} &\colorbox{magenta!10}{1.083 $\rightarrow$ 0.597~\textsuperscript{ }} \\SE &\colorbox{magenta!20}{5.626 $\rightarrow$ 4.989~\textsuperscript{$\ast$}} &\colorbox{magenta!20}{5.515 $\rightarrow$ 4.848~\textsuperscript{$\ast$}} &\colorbox{magenta!20}{5.572 $\rightarrow$ 4.841~\textsuperscript{$\ast$}} &\colorbox{magenta!10}{4.715 $\rightarrow$ 4.460~\textsuperscript{ }} &\colorbox{magenta!20}{4.672 $\rightarrow$ 4.291~\textsuperscript{$\ast$}} &\colorbox{magenta!20}{4.897 $\rightarrow$ 4.638~\textsuperscript{$\ast$}} &\colorbox{magenta!20}{4.163 $\rightarrow$ 3.548~\textsuperscript{$\ast$}} &\colorbox{magenta!20}{4.104 $\rightarrow$ 3.381~\textsuperscript{$\ast$}} &\colorbox{magenta!10}{4.388 $\rightarrow$ 4.107~\textsuperscript{ }} \\PE+M &\colorbox{magenta!20}{5.284 $\rightarrow$ 4.891~\textsuperscript{$\ast$}} &\colorbox{magenta!10}{5.052 $\rightarrow$ 4.904~\textsuperscript{ }} &\colorbox{magenta!10}{5.665 $\rightarrow$ 5.652~\textsuperscript{ }} &\colorbox{magenta!10}{2.716 $\rightarrow$ 2.633~\textsuperscript{ }} &\colorbox{magenta!10}{2.381 $\rightarrow$ 2.249~\textsuperscript{ }} &\colorbox{magenta!5}{2.594 $\rightarrow$ 2.597~\textsuperscript{ }} &\colorbox{magenta!10}{1.309 $\rightarrow$ 0.844~\textsuperscript{ }} &\colorbox{magenta!10}{1.016 $\rightarrow$ 0.782~\textsuperscript{ }} &\colorbox{magenta!10}{1.328 $\rightarrow$ 0.684~\textsuperscript{ }} \\SE+M &\colorbox{magenta!20}{8.566 $\rightarrow$ 7.579~\textsuperscript{$\ast$}} &\colorbox{magenta!20}{8.377 $\rightarrow$ 7.471~\textsuperscript{$\ast$}} &\colorbox{magenta!20}{8.914 $\rightarrow$ 7.962~\textsuperscript{$\ast$}} &\colorbox{magenta!20}{5.978 $\rightarrow$ 5.521~\textsuperscript{$\ast$}} &\colorbox{magenta!20}{5.737 $\rightarrow$ 5.170~\textsuperscript{$\ast$}} &\colorbox{magenta!20}{6.109 $\rightarrow$ 5.733~\textsuperscript{$\ast$}} &\colorbox{magenta!20}{4.599 $\rightarrow$ 3.700~\textsuperscript{$\ast$}} &\colorbox{magenta!20}{4.481 $\rightarrow$ 3.610~\textsuperscript{$\ast$}} &\colorbox{magenta!10}{4.764 $\rightarrow$ 4.221~\textsuperscript{ }} \\EigV &\colorbox{magenta!20}{2.410 $\rightarrow$ 1.694~\textsuperscript{$\ast$}} &\colorbox{magenta!20}{2.454 $\rightarrow$ 1.375~\textsuperscript{$\ast$}} &\colorbox{magenta!20}{2.340 $\rightarrow$ 1.216~\textsuperscript{$\ast$}} &\colorbox{magenta!20}{2.147 $\rightarrow$ 1.802~\textsuperscript{$\ast$}} &\colorbox{magenta!20}{2.271 $\rightarrow$ 1.700~\textsuperscript{$\ast$}} &\colorbox{magenta!10}{2.654 $\rightarrow$ 2.508~\textsuperscript{ }} &\colorbox{magenta!20}{2.453 $\rightarrow$ 1.338~\textsuperscript{$\ast$}} &\colorbox{magenta!20}{2.088 $\rightarrow$ 1.274~\textsuperscript{$\ast$}} &\colorbox{magenta!10}{2.758 $\rightarrow$ 1.910~\textsuperscript{ }} \\ECC &\colorbox{magenta!20}{0.564 $\rightarrow$ 0.302~\textsuperscript{$\ast$}} &\colorbox{magenta!20}{0.600 $\rightarrow$ 0.240~\textsuperscript{$\ast$}} &\colorbox{magenta!20}{0.542 $\rightarrow$ 0.166~\textsuperscript{$\ast$}} &\colorbox{magenta!20}{0.554 $\rightarrow$ 0.382~\textsuperscript{$\ast$}} &\colorbox{magenta!20}{0.561 $\rightarrow$ 0.331~\textsuperscript{$\ast$}} &\colorbox{magenta!10}{0.617 $\rightarrow$ 0.600~\textsuperscript{ }} &\colorbox{magenta!20}{0.541 $\rightarrow$ 0.197~\textsuperscript{$\ast$}} &\colorbox{magenta!20}{0.477 $\rightarrow$ 0.152~\textsuperscript{$\ast$}} &\colorbox{magenta!10}{0.503 $\rightarrow$ 0.443~\textsuperscript{ }} \\Deg &\colorbox{magenta!20}{0.304 $\rightarrow$ 0.172~\textsuperscript{$\ast$}} &\colorbox{magenta!20}{0.314 $\rightarrow$ 0.113~\textsuperscript{$\ast$}} &\colorbox{magenta!20}{0.299 $\rightarrow$ 0.069~\textsuperscript{$\ast$}} &\colorbox{magenta!20}{0.274 $\rightarrow$ 0.194~\textsuperscript{$\ast$}} &\colorbox{magenta!20}{0.294 $\rightarrow$ 0.186~\textsuperscript{$\ast$}} &\colorbox{magenta!10}{0.353 $\rightarrow$ 0.325~\textsuperscript{ }} &\colorbox{magenta!20}{0.286 $\rightarrow$ 0.101~\textsuperscript{$\ast$}} &\colorbox{magenta!20}{0.228 $\rightarrow$ 0.073~\textsuperscript{$\ast$}} &\colorbox{magenta!10}{0.343 $\rightarrow$ 0.254~\textsuperscript{ }} \\\hline\multicolumn{6}{l}{\textbf{Axiom 3:} Positively Changed $\downarrow$} \\ \hline PE &\colorbox{green!10}{2.113 $\rightarrow$ 0.909~\textsuperscript{$\ast$}} &\colorbox{green!10}{1.989 $\rightarrow$ 0.939~\textsuperscript{$\ast$}} &\colorbox{green!10}{2.006 $\rightarrow$ 0.847~\textsuperscript{$\ast$}} &\colorbox{green!10}{1.481 $\rightarrow$ 0.665~\textsuperscript{$\ast$}} &\colorbox{green!10}{1.413 $\rightarrow$ 0.702~\textsuperscript{$\ast$}} &\colorbox{green!10}{1.403 $\rightarrow$ 0.653~\textsuperscript{$\ast$}} &\colorbox{green!10}{1.375 $\rightarrow$ 0.347~\textsuperscript{$\ast$}} &\colorbox{green!10}{1.416 $\rightarrow$ 0.298~\textsuperscript{$\ast$}} &\colorbox{green!10}{1.342 $\rightarrow$ 0.268~\textsuperscript{$\ast$}} \\SE &\colorbox{green!10}{5.606 $\rightarrow$ 3.589~\textsuperscript{$\ast$}} &\colorbox{green!10}{5.459 $\rightarrow$ 3.589~\textsuperscript{$\ast$}} &\colorbox{green!10}{5.500 $\rightarrow$ 3.544~\textsuperscript{$\ast$}} &\colorbox{green!10}{4.970 $\rightarrow$ 3.347~\textsuperscript{$\ast$}} &\colorbox{green!10}{4.966 $\rightarrow$ 3.469~\textsuperscript{$\ast$}} &\colorbox{green!10}{4.972 $\rightarrow$ 3.287~\textsuperscript{$\ast$}} &\colorbox{green!10}{4.889 $\rightarrow$ 3.015~\textsuperscript{$\ast$}} &\colorbox{green!10}{5.091 $\rightarrow$ 3.013~\textsuperscript{$\ast$}} &\colorbox{green!10}{4.884 $\rightarrow$ 3.051~\textsuperscript{$\ast$}} \\PE+M &\colorbox{green!10}{3.479 $\rightarrow$ 2.056~\textsuperscript{$\ast$}} &\colorbox{green!10}{3.420 $\rightarrow$ 1.991~\textsuperscript{$\ast$}} &\colorbox{green!10}{3.416 $\rightarrow$ 2.012~\textsuperscript{$\ast$}} &\colorbox{green!10}{2.001 $\rightarrow$ 0.917~\textsuperscript{$\ast$}} &\colorbox{green!10}{2.026 $\rightarrow$ 1.020~\textsuperscript{$\ast$}} &\colorbox{green!10}{1.930 $\rightarrow$ 0.938~\textsuperscript{$\ast$}} &\colorbox{green!10}{1.708 $\rightarrow$ 0.398~\textsuperscript{$\ast$}} &\colorbox{green!10}{1.735 $\rightarrow$ 0.374~\textsuperscript{$\ast$}} &\colorbox{green!10}{1.604 $\rightarrow$ 0.340~\textsuperscript{$\ast$}} \\SE+M &\colorbox{green!10}{7.268 $\rightarrow$ 4.703~\textsuperscript{$\ast$}} &\colorbox{green!10}{7.069 $\rightarrow$ 4.616~\textsuperscript{$\ast$}} &\colorbox{green!10}{7.101 $\rightarrow$ 4.637~\textsuperscript{$\ast$}} &\colorbox{green!10}{5.790 $\rightarrow$ 3.648~\textsuperscript{$\ast$}} &\colorbox{green!10}{5.804 $\rightarrow$ 3.825~\textsuperscript{$\ast$}} &\colorbox{green!10}{5.760 $\rightarrow$ 3.579~\textsuperscript{$\ast$}} &\colorbox{green!10}{5.514 $\rightarrow$ 3.072~\textsuperscript{$\ast$}} &\colorbox{green!10}{5.681 $\rightarrow$ 3.082~\textsuperscript{$\ast$}} &\colorbox{green!10}{5.379 $\rightarrow$ 3.099~\textsuperscript{$\ast$}} \\EigV &\colorbox{green!10}{3.692 $\rightarrow$ 1.220~\textsuperscript{$\ast$}} &\colorbox{green!10}{3.561 $\rightarrow$ 1.182~\textsuperscript{$\ast$}} &\colorbox{green!10}{3.551 $\rightarrow$ 1.159~\textsuperscript{$\ast$}} &\colorbox{green!10}{3.588 $\rightarrow$ 1.245~\textsuperscript{$\ast$}} &\colorbox{green!10}{3.625 $\rightarrow$ 1.346~\textsuperscript{$\ast$}} &\colorbox{green!10}{3.650 $\rightarrow$ 1.277~\textsuperscript{$\ast$}} &\colorbox{green!10}{4.131 $\rightarrow$ 1.139~\textsuperscript{$\ast$}} &\colorbox{green!10}{4.733 $\rightarrow$ 1.114~\textsuperscript{$\ast$}} &\colorbox{green!10}{4.449 $\rightarrow$ 1.102~\textsuperscript{$\ast$}} \\ECC &\colorbox{green!10}{0.756 $\rightarrow$ 0.144~\textsuperscript{$\ast$}} &\colorbox{green!10}{0.701 $\rightarrow$ 0.111~\textsuperscript{$\ast$}} &\colorbox{green!10}{0.714 $\rightarrow$ 0.115~\textsuperscript{$\ast$}} &\colorbox{green!10}{0.801 $\rightarrow$ 0.163~\textsuperscript{$\ast$}} &\colorbox{green!10}{0.807 $\rightarrow$ 0.218~\textsuperscript{$\ast$}} &\colorbox{green!10}{0.810 $\rightarrow$ 0.179~\textsuperscript{$\ast$}} &\colorbox{green!10}{0.790 $\rightarrow$ 0.085~\textsuperscript{$\ast$}} &\colorbox{green!10}{0.823 $\rightarrow$ 0.081~\textsuperscript{$\ast$}} &\colorbox{green!10}{0.780 $\rightarrow$ 0.072~\textsuperscript{$\ast$}} \\Deg &\colorbox{green!10}{0.507 $\rightarrow$ 0.065~\textsuperscript{$\ast$}} &\colorbox{green!10}{0.484 $\rightarrow$ 0.057~\textsuperscript{$\ast$}} &\colorbox{green!10}{0.488 $\rightarrow$ 0.051~\textsuperscript{$\ast$}} &\colorbox{green!10}{0.497 $\rightarrow$ 0.076~\textsuperscript{$\ast$}} &\colorbox{green!10}{0.502 $\rightarrow$ 0.093~\textsuperscript{$\ast$}} &\colorbox{green!10}{0.504 $\rightarrow$ 0.079~\textsuperscript{$\ast$}} &\colorbox{green!10}{0.547 $\rightarrow$ 0.044~\textsuperscript{$\ast$}} &\colorbox{green!10}{0.588 $\rightarrow$ 0.035~\textsuperscript{$\ast$}} &\colorbox{green!10}{0.544 $\rightarrow$ 0.032~\textsuperscript{$\ast$}} \\\hline\multicolumn{6}{l}{\textbf{Axiom 4:} Negatively Changed $\uparrow$} \\ \hline PE &\colorbox{magenta!5}{1.609 $\rightarrow$ 1.695~\textsuperscript{ }} &\colorbox{magenta!5}{1.621 $\rightarrow$ 1.635~\textsuperscript{ }} &\colorbox{magenta!5}{1.598 $\rightarrow$ 1.688~\textsuperscript{ }} &\colorbox{green!10}{0.945 $\rightarrow$ 1.325~\textsuperscript{$\ast$}} &\colorbox{green!10}{0.889 $\rightarrow$ 1.364~\textsuperscript{$\ast$}} &\colorbox{green!10}{1.034 $\rightarrow$ 1.396~\textsuperscript{$\ast$}} &\colorbox{magenta!10}{0.933 $\rightarrow$ 0.636~\textsuperscript{ }} &\colorbox{magenta!10}{1.006 $\rightarrow$ 0.558~\textsuperscript{ }} &\colorbox{magenta!10}{1.252 $\rightarrow$ 0.463~\textsuperscript{ }} \\SE &\colorbox{magenta!20}{4.899 $\rightarrow$ 4.457~\textsuperscript{$\ast$}} &\colorbox{magenta!20}{4.899 $\rightarrow$ 4.437~\textsuperscript{$\ast$}} &\colorbox{magenta!10}{4.915 $\rightarrow$ 4.497~\textsuperscript{ }} &\colorbox{magenta!5}{4.160 $\rightarrow$ 4.312~\textsuperscript{ }} &\colorbox{magenta!5}{4.157 $\rightarrow$ 4.273~\textsuperscript{ }} &\colorbox{magenta!5}{4.297 $\rightarrow$ 4.339~\textsuperscript{ }} &\colorbox{magenta!20}{4.152 $\rightarrow$ 3.552~\textsuperscript{$\ast$}} &\colorbox{magenta!20}{4.192 $\rightarrow$ 3.409~\textsuperscript{$\ast$}} &\colorbox{magenta!20}{4.830 $\rightarrow$ 3.690~\textsuperscript{$\ast$}} \\PE+M &\colorbox{magenta!5}{3.446 $\rightarrow$ 3.653~\textsuperscript{ }} &\colorbox{magenta!5}{3.522 $\rightarrow$ 3.692~\textsuperscript{ }} &\colorbox{magenta!5}{3.465 $\rightarrow$ 4.158~\textsuperscript{ }} &\colorbox{green!10}{1.566 $\rightarrow$ 2.123~\textsuperscript{$\ast$}} &\colorbox{green!10}{1.306 $\rightarrow$ 1.946~\textsuperscript{$\ast$}} &\colorbox{green!10}{1.486 $\rightarrow$ 2.178~\textsuperscript{$\ast$}} &\colorbox{magenta!20}{1.164 $\rightarrow$ 0.714~\textsuperscript{$\ast$}} &\colorbox{magenta!20}{1.298 $\rightarrow$ 0.748~\textsuperscript{$\ast$}} &\colorbox{magenta!10}{1.689 $\rightarrow$ 0.747~\textsuperscript{ }} \\SE+M &\colorbox{magenta!20}{6.764 $\rightarrow$ 6.286~\textsuperscript{$\ast$}} &\colorbox{magenta!20}{6.803 $\rightarrow$ 6.377~\textsuperscript{$\ast$}} &\colorbox{magenta!10}{6.643 $\rightarrow$ 6.442~\textsuperscript{ }} &\colorbox{magenta!5}{4.953 $\rightarrow$ 5.121~\textsuperscript{ }} &\colorbox{magenta!5}{4.769 $\rightarrow$ 4.933~\textsuperscript{ }} &\colorbox{magenta!5}{4.983 $\rightarrow$ 5.088~\textsuperscript{ }} &\colorbox{magenta!20}{4.553 $\rightarrow$ 3.690~\textsuperscript{$\ast$}} &\colorbox{magenta!20}{4.653 $\rightarrow$ 3.608~\textsuperscript{$\ast$}} &\colorbox{magenta!20}{5.381 $\rightarrow$ 4.007~\textsuperscript{$\ast$}} \\EigV &\colorbox{magenta!20}{2.262 $\rightarrow$ 1.582~\textsuperscript{$\ast$}} &\colorbox{magenta!20}{2.244 $\rightarrow$ 1.503~\textsuperscript{$\ast$}} &\colorbox{magenta!20}{2.233 $\rightarrow$ 1.367~\textsuperscript{$\ast$}} &\colorbox{magenta!10}{2.089 $\rightarrow$ 1.908~\textsuperscript{ }} &\colorbox{magenta!10}{2.141 $\rightarrow$ 1.908~\textsuperscript{ }} &\colorbox{magenta!10}{2.399 $\rightarrow$ 2.131~\textsuperscript{ }} &\colorbox{magenta!20}{2.593 $\rightarrow$ 1.449~\textsuperscript{$\ast$}} &\colorbox{magenta!20}{2.557 $\rightarrow$ 1.412~\textsuperscript{$\ast$}} &\colorbox{magenta!20}{3.567 $\rightarrow$ 1.449~\textsuperscript{$\ast$}} \\ECC &\colorbox{magenta!20}{0.594 $\rightarrow$ 0.332~\textsuperscript{$\ast$}} &\colorbox{magenta!20}{0.565 $\rightarrow$ 0.295~\textsuperscript{$\ast$}} &\colorbox{magenta!20}{0.490 $\rightarrow$ 0.270~\textsuperscript{$\ast$}} &\colorbox{magenta!10}{0.501 $\rightarrow$ 0.453~\textsuperscript{ }} &\colorbox{magenta!10}{0.542 $\rightarrow$ 0.456~\textsuperscript{ }} &\colorbox{magenta!10}{0.614 $\rightarrow$ 0.555~\textsuperscript{ }} &\colorbox{magenta!20}{0.540 $\rightarrow$ 0.262~\textsuperscript{$\ast$}} &\colorbox{magenta!20}{0.548 $\rightarrow$ 0.220~\textsuperscript{$\ast$}} &\colorbox{magenta!20}{0.707 $\rightarrow$ 0.237~\textsuperscript{$\ast$}} \\Deg &\colorbox{magenta!20}{0.301 $\rightarrow$ 0.163~\textsuperscript{$\ast$}} &\colorbox{magenta!20}{0.294 $\rightarrow$ 0.148~\textsuperscript{$\ast$}} &\colorbox{magenta!20}{0.308 $\rightarrow$ 0.123~\textsuperscript{$\ast$}} &\colorbox{magenta!10}{0.239 $\rightarrow$ 0.237~\textsuperscript{ }} &\colorbox{magenta!10}{0.253 $\rightarrow$ 0.251~\textsuperscript{ }} &\colorbox{magenta!10}{0.313 $\rightarrow$ 0.299~\textsuperscript{ }} &\colorbox{magenta!20}{0.320 $\rightarrow$ 0.128~\textsuperscript{$\ast$}} &\colorbox{magenta!20}{0.320 $\rightarrow$ 0.115~\textsuperscript{$\ast$}} &\colorbox{magenta!20}{0.463 $\rightarrow$ 0.140~\textsuperscript{$\ast$}} \\\hline

\end{tabular}
\shrink
\caption{
Comparison of changes in average uncertainty values for Axioms 1--4 before (left) and after (right) applying RAG with Llama2-chat. Axioms are implemented using the \textit{Reference-based} method.
Colors \colorbox{green!10}{green} and \colorbox{magenta!20}{deep red} indicate significant changes aligning or conflicting with axioms, respectively. 
Color \colorbox{magenta!10}{shallow red} represents non-significant changes conflicting with axioms. Significance is marked by $\ast$.}
\label{tab:axioms_unc_changes_correctness_llama2}
\shrink
\end{table*}

\begin{table*}[t]
\centering
\setlength{\tabcolsep}{0.0pt}
\tiny
\begin{tabular}{c|ccc|ccc|ccc}
\hline
\textbf{UE} &
\multicolumn{3}{c}{\textbf{NQ-open}} & \multicolumn{3}{c}{\textbf{TriviaQA}} & \multicolumn{3}{c}{\textbf{PopQA}} \\ \hline %
& \textbf{BM25} & \textbf{Contriever} & \textbf{$\text{Doc}^{+}$} & 
\textbf{BM25} & \textbf{Contriever} & \textbf{$\text{Doc}^{+}$} & 
\textbf{BM25} & \textbf{Contriever} & \textbf{$\text{Doc}^{+}$} \\
\hline\hline

\multicolumn{6}{l}{\textbf{Axiom 1:} Positively Consistent $\downarrow$} \\ \hline PE &\colorbox{green!10}{1.620 $\rightarrow$ 1.332~\textsuperscript{$\ast$}} &\colorbox{green!10}{1.538 $\rightarrow$ 1.288~\textsuperscript{$\ast$}} &\colorbox{green!10}{1.544 $\rightarrow$ 1.232~\textsuperscript{$\ast$}} &\colorbox{green!10}{0.531 $\rightarrow$ 0.483~\textsuperscript{$\ast$}} &\colorbox{green!10}{0.549 $\rightarrow$ 0.494~\textsuperscript{$\ast$}} &\colorbox{green!10}{0.570 $\rightarrow$ 0.456~\textsuperscript{$\ast$}} &\colorbox{green!10}{0.893 $\rightarrow$ 0.673~\textsuperscript{$\ast$}} &\colorbox{green!10}{0.886 $\rightarrow$ 0.638~\textsuperscript{$\ast$}} &\colorbox{green!10}{1.368 $\rightarrow$ 0.419~\textsuperscript{$\ast$}} \\SE &\colorbox{green!10}{4.874 $\rightarrow$ 4.164~\textsuperscript{$\ast$}} &\colorbox{green!10}{4.876 $\rightarrow$ 4.060~\textsuperscript{$\ast$}} &\colorbox{green!10}{4.941 $\rightarrow$ 3.922~\textsuperscript{$\ast$}} &\colorbox{green!10}{3.460 $\rightarrow$ 3.265~\textsuperscript{$\ast$}} &\colorbox{green!10}{3.508 $\rightarrow$ 3.299~\textsuperscript{$\ast$}} &\colorbox{green!10}{3.565 $\rightarrow$ 3.241~\textsuperscript{$\ast$}} &\colorbox{green!10}{4.063 $\rightarrow$ 3.361~\textsuperscript{$\ast$}} &\colorbox{green!10}{4.162 $\rightarrow$ 3.354~\textsuperscript{$\ast$}} &\colorbox{green!10}{5.379 $\rightarrow$ 3.112~\textsuperscript{$\ast$}} \\PE+M &\colorbox{green!10}{3.682 $\rightarrow$ 3.283~\textsuperscript{$\ast$}} &\colorbox{green!10}{3.380 $\rightarrow$ 3.188~\textsuperscript{$\ast$}} &\colorbox{green!10}{3.395 $\rightarrow$ 3.080~\textsuperscript{$\ast$}} &\colorbox{green!10}{0.943 $\rightarrow$ 0.903~\textsuperscript{$\ast$}} &\colorbox{green!10}{0.987 $\rightarrow$ 0.948~\textsuperscript{$\ast$}} &\colorbox{green!10}{1.010 $\rightarrow$ 0.869~\textsuperscript{$\ast$}} &\colorbox{green!10}{1.220 $\rightarrow$ 0.942~\textsuperscript{$\ast$}} &\colorbox{green!10}{1.195 $\rightarrow$ 0.836~\textsuperscript{$\ast$}} &\colorbox{green!10}{2.115 $\rightarrow$ 0.615~\textsuperscript{$\ast$}} \\SE+M &\colorbox{green!10}{6.710 $\rightarrow$ 5.863~\textsuperscript{$\ast$}} &\colorbox{green!10}{6.531 $\rightarrow$ 5.746~\textsuperscript{$\ast$}} &\colorbox{green!10}{6.594 $\rightarrow$ 5.602~\textsuperscript{$\ast$}} &\colorbox{green!10}{3.839 $\rightarrow$ 3.638~\textsuperscript{$\ast$}} &\colorbox{green!10}{3.913 $\rightarrow$ 3.715~\textsuperscript{$\ast$}} &\colorbox{green!10}{3.971 $\rightarrow$ 3.615~\textsuperscript{$\ast$}} &\colorbox{green!10}{4.315 $\rightarrow$ 3.539~\textsuperscript{$\ast$}} &\colorbox{green!10}{4.424 $\rightarrow$ 3.459~\textsuperscript{$\ast$}} &\colorbox{green!10}{6.087 $\rightarrow$ 3.224~\textsuperscript{$\ast$}} \\EigV &\colorbox{green!10}{1.724 $\rightarrow$ 1.285~\textsuperscript{$\ast$}} &\colorbox{green!10}{1.788 $\rightarrow$ 1.172~\textsuperscript{$\ast$}} &\colorbox{green!10}{1.901 $\rightarrow$ 1.069~\textsuperscript{$\ast$}} &\colorbox{green!10}{1.277 $\rightarrow$ 1.129~\textsuperscript{$\ast$}} &\colorbox{green!10}{1.294 $\rightarrow$ 1.162~\textsuperscript{$\ast$}} &\colorbox{green!10}{1.344 $\rightarrow$ 1.114~\textsuperscript{$\ast$}} &\colorbox{green!10}{1.614 $\rightarrow$ 1.119~\textsuperscript{$\ast$}} &\colorbox{green!10}{1.837 $\rightarrow$ 1.095~\textsuperscript{$\ast$}} &\colorbox{green!10}{3.781 $\rightarrow$ 1.041~\textsuperscript{$\ast$}} \\ECC &\colorbox{green!10}{0.356 $\rightarrow$ 0.169~\textsuperscript{$\ast$}} &\colorbox{green!10}{0.381 $\rightarrow$ 0.104~\textsuperscript{$\ast$}} &\colorbox{green!10}{0.405 $\rightarrow$ 0.043~\textsuperscript{$\ast$}} &\colorbox{green!10}{0.155 $\rightarrow$ 0.082~\textsuperscript{$\ast$}} &\colorbox{green!10}{0.164 $\rightarrow$ 0.094~\textsuperscript{$\ast$}} &\colorbox{green!10}{0.187 $\rightarrow$ 0.076~\textsuperscript{$\ast$}} &\colorbox{green!10}{0.260 $\rightarrow$ 0.050~\textsuperscript{$\ast$}} &\colorbox{green!10}{0.288 $\rightarrow$ 0.052~\textsuperscript{$\ast$}} &\colorbox{green!10}{0.621 $\rightarrow$ 0.021~\textsuperscript{$\ast$}} \\Deg &\colorbox{green!10}{0.175 $\rightarrow$ 0.088~\textsuperscript{$\ast$}} &\colorbox{green!10}{0.185 $\rightarrow$ 0.053~\textsuperscript{$\ast$}} &\colorbox{green!10}{0.208 $\rightarrow$ 0.023~\textsuperscript{$\ast$}} &\colorbox{green!10}{0.063 $\rightarrow$ 0.038~\textsuperscript{$\ast$}} &\colorbox{green!10}{0.067 $\rightarrow$ 0.042~\textsuperscript{$\ast$}} &\colorbox{green!10}{0.076 $\rightarrow$ 0.030~\textsuperscript{$\ast$}} &\colorbox{green!10}{0.129 $\rightarrow$ 0.045~\textsuperscript{$\ast$}} &\colorbox{green!10}{0.157 $\rightarrow$ 0.033~\textsuperscript{$\ast$}} &\colorbox{green!10}{0.426 $\rightarrow$ 0.016~\textsuperscript{$\ast$}} \\\hline\multicolumn{6}{l}{\textbf{Axiom 2:} Negatively Consistent $\uparrow$} \\ \hline PE &\colorbox{magenta!20}{2.460 $\rightarrow$ 2.303~\textsuperscript{$\ast$}} &\colorbox{magenta!10}{2.353 $\rightarrow$ 2.321~\textsuperscript{ }} &\colorbox{magenta!10}{2.377 $\rightarrow$ 2.374~\textsuperscript{ }} &\colorbox{magenta!10}{1.512 $\rightarrow$ 1.397~\textsuperscript{ }} &\colorbox{magenta!5}{1.226 $\rightarrow$ 1.228~\textsuperscript{ }} &\colorbox{magenta!20}{1.477 $\rightarrow$ 1.421~\textsuperscript{$\ast$}} &\colorbox{magenta!20}{0.933 $\rightarrow$ 0.589~\textsuperscript{$\ast$}} &\colorbox{magenta!20}{0.804 $\rightarrow$ 0.450~\textsuperscript{$\ast$}} &\colorbox{magenta!10}{1.196 $\rightarrow$ 0.570~\textsuperscript{ }} \\SE &\colorbox{magenta!20}{5.846 $\rightarrow$ 5.233~\textsuperscript{$\ast$}} &\colorbox{magenta!20}{5.614 $\rightarrow$ 5.074~\textsuperscript{$\ast$}} &\colorbox{magenta!20}{5.619 $\rightarrow$ 4.966~\textsuperscript{$\ast$}} &\colorbox{magenta!20}{4.697 $\rightarrow$ 4.384~\textsuperscript{$\ast$}} &\colorbox{magenta!20}{4.449 $\rightarrow$ 4.133~\textsuperscript{$\ast$}} &\colorbox{magenta!20}{4.936 $\rightarrow$ 4.699~\textsuperscript{$\ast$}} &\colorbox{magenta!20}{4.407 $\rightarrow$ 3.314~\textsuperscript{$\ast$}} &\colorbox{magenta!20}{4.290 $\rightarrow$ 3.215~\textsuperscript{$\ast$}} &\colorbox{magenta!10}{4.620 $\rightarrow$ 3.442~\textsuperscript{ }} \\PE+M &\colorbox{magenta!10}{6.014 $\rightarrow$ 5.908~\textsuperscript{ }} &\colorbox{magenta!5}{5.523 $\rightarrow$ 5.664~\textsuperscript{ }} &\colorbox{magenta!5}{5.783 $\rightarrow$ 5.920~\textsuperscript{ }} &\colorbox{magenta!10}{2.917 $\rightarrow$ 2.797~\textsuperscript{ }} &\colorbox{magenta!5}{2.260 $\rightarrow$ 2.313~\textsuperscript{ }} &\colorbox{magenta!5}{2.777 $\rightarrow$ 2.833~\textsuperscript{ }} &\colorbox{magenta!20}{1.376 $\rightarrow$ 0.901~\textsuperscript{$\ast$}} &\colorbox{magenta!20}{1.230 $\rightarrow$ 0.762~\textsuperscript{$\ast$}} &\colorbox{magenta!10}{1.631 $\rightarrow$ 0.686~\textsuperscript{ }} \\SE+M &\colorbox{magenta!20}{9.087 $\rightarrow$ 8.557~\textsuperscript{$\ast$}} &\colorbox{magenta!10}{8.493 $\rightarrow$ 8.092~\textsuperscript{ }} &\colorbox{magenta!10}{8.728 $\rightarrow$ 8.147~\textsuperscript{ }} &\colorbox{magenta!10}{6.033 $\rightarrow$ 5.699~\textsuperscript{ }} &\colorbox{magenta!20}{5.462 $\rightarrow$ 5.100~\textsuperscript{$\ast$}} &\colorbox{magenta!10}{6.121 $\rightarrow$ 6.009~\textsuperscript{ }} &\colorbox{magenta!20}{4.819 $\rightarrow$ 3.551~\textsuperscript{$\ast$}} &\colorbox{magenta!20}{4.702 $\rightarrow$ 3.456~\textsuperscript{$\ast$}} &\colorbox{magenta!10}{4.875 $\rightarrow$ 3.504~\textsuperscript{ }} \\EigV &\colorbox{magenta!20}{2.177 $\rightarrow$ 1.529~\textsuperscript{$\ast$}} &\colorbox{magenta!20}{2.047 $\rightarrow$ 1.303~\textsuperscript{$\ast$}} &\colorbox{magenta!20}{1.869 $\rightarrow$ 1.071~\textsuperscript{$\ast$}} &\colorbox{magenta!10}{1.648 $\rightarrow$ 1.472~\textsuperscript{ }} &\colorbox{magenta!10}{1.655 $\rightarrow$ 1.489~\textsuperscript{ }} &\colorbox{magenta!20}{2.284 $\rightarrow$ 2.025~\textsuperscript{$\ast$}} &\colorbox{magenta!20}{2.041 $\rightarrow$ 1.098~\textsuperscript{$\ast$}} &\colorbox{magenta!20}{2.055 $\rightarrow$ 1.181~\textsuperscript{$\ast$}} &\colorbox{magenta!10}{2.188 $\rightarrow$ 1.143~\textsuperscript{ }} \\ECC &\colorbox{magenta!20}{0.507 $\rightarrow$ 0.256~\textsuperscript{$\ast$}} &\colorbox{magenta!20}{0.437 $\rightarrow$ 0.166~\textsuperscript{$\ast$}} &\colorbox{magenta!20}{0.453 $\rightarrow$ 0.040~\textsuperscript{$\ast$}} &\colorbox{magenta!20}{0.367 $\rightarrow$ 0.223~\textsuperscript{$\ast$}} &\colorbox{magenta!20}{0.394 $\rightarrow$ 0.243~\textsuperscript{$\ast$}} &\colorbox{magenta!20}{0.476 $\rightarrow$ 0.394~\textsuperscript{$\ast$}} &\colorbox{magenta!20}{0.338 $\rightarrow$ 0.065~\textsuperscript{$\ast$}} &\colorbox{magenta!20}{0.411 $\rightarrow$ 0.069~\textsuperscript{$\ast$}} &\colorbox{magenta!10}{0.514 $\rightarrow$ 0.041~\textsuperscript{ }} \\Deg &\colorbox{magenta!20}{0.260 $\rightarrow$ 0.134~\textsuperscript{$\ast$}} &\colorbox{magenta!20}{0.227 $\rightarrow$ 0.080~\textsuperscript{$\ast$}} &\colorbox{magenta!20}{0.210 $\rightarrow$ 0.022~\textsuperscript{$\ast$}} &\colorbox{magenta!10}{0.153 $\rightarrow$ 0.127~\textsuperscript{ }} &\colorbox{magenta!10}{0.152 $\rightarrow$ 0.120~\textsuperscript{ }} &\colorbox{magenta!20}{0.254 $\rightarrow$ 0.205~\textsuperscript{$\ast$}} &\colorbox{magenta!20}{0.200 $\rightarrow$ 0.030~\textsuperscript{$\ast$}} &\colorbox{magenta!20}{0.194 $\rightarrow$ 0.044~\textsuperscript{$\ast$}} &\colorbox{magenta!10}{0.260 $\rightarrow$ 0.055~\textsuperscript{ }} \\\hline\multicolumn{6}{l}{\textbf{Axiom 3:} Positively Changed $\downarrow$} \\ \hline PE &\colorbox{green!10}{1.972 $\rightarrow$ 1.038~\textsuperscript{$\ast$}} &\colorbox{green!10}{1.972 $\rightarrow$ 1.120~\textsuperscript{$\ast$}} &\colorbox{green!10}{2.020 $\rightarrow$ 1.097~\textsuperscript{$\ast$}} &\colorbox{green!10}{1.492 $\rightarrow$ 0.531~\textsuperscript{$\ast$}} &\colorbox{green!10}{1.446 $\rightarrow$ 0.515~\textsuperscript{$\ast$}} &\colorbox{green!10}{1.452 $\rightarrow$ 0.510~\textsuperscript{$\ast$}} &\colorbox{green!10}{1.837 $\rightarrow$ 0.734~\textsuperscript{$\ast$}} &\colorbox{green!10}{1.727 $\rightarrow$ 0.566~\textsuperscript{$\ast$}} &\colorbox{green!10}{1.458 $\rightarrow$ 0.403~\textsuperscript{$\ast$}} \\SE &\colorbox{green!10}{5.861 $\rightarrow$ 3.808~\textsuperscript{$\ast$}} &\colorbox{green!10}{5.813 $\rightarrow$ 3.855~\textsuperscript{$\ast$}} &\colorbox{green!10}{5.898 $\rightarrow$ 3.838~\textsuperscript{$\ast$}} &\colorbox{green!10}{5.527 $\rightarrow$ 3.337~\textsuperscript{$\ast$}} &\colorbox{green!10}{5.569 $\rightarrow$ 3.326~\textsuperscript{$\ast$}} &\colorbox{green!10}{5.497 $\rightarrow$ 3.364~\textsuperscript{$\ast$}} &\colorbox{green!10}{6.309 $\rightarrow$ 3.349~\textsuperscript{$\ast$}} &\colorbox{green!10}{6.227 $\rightarrow$ 3.276~\textsuperscript{$\ast$}} &\colorbox{green!10}{5.662 $\rightarrow$ 3.104~\textsuperscript{$\ast$}} \\PE+M &\colorbox{green!10}{3.917 $\rightarrow$ 2.599~\textsuperscript{$\ast$}} &\colorbox{green!10}{4.061 $\rightarrow$ 2.810~\textsuperscript{$\ast$}} &\colorbox{green!10}{4.063 $\rightarrow$ 2.762~\textsuperscript{$\ast$}} &\colorbox{green!10}{2.544 $\rightarrow$ 0.959~\textsuperscript{$\ast$}} &\colorbox{green!10}{2.545 $\rightarrow$ 1.033~\textsuperscript{$\ast$}} &\colorbox{green!10}{2.480 $\rightarrow$ 0.927~\textsuperscript{$\ast$}} &\colorbox{green!10}{2.935 $\rightarrow$ 0.970~\textsuperscript{$\ast$}} &\colorbox{green!10}{2.686 $\rightarrow$ 0.867~\textsuperscript{$\ast$}} &\colorbox{green!10}{2.244 $\rightarrow$ 0.594~\textsuperscript{$\ast$}} \\SE+M &\colorbox{green!10}{7.587 $\rightarrow$ 5.162~\textsuperscript{$\ast$}} &\colorbox{green!10}{7.662 $\rightarrow$ 5.299~\textsuperscript{$\ast$}} &\colorbox{green!10}{7.746 $\rightarrow$ 5.303~\textsuperscript{$\ast$}} &\colorbox{green!10}{6.542 $\rightarrow$ 3.690~\textsuperscript{$\ast$}} &\colorbox{green!10}{6.606 $\rightarrow$ 3.798~\textsuperscript{$\ast$}} &\colorbox{green!10}{6.465 $\rightarrow$ 3.716~\textsuperscript{$\ast$}} &\colorbox{green!10}{7.365 $\rightarrow$ 3.467~\textsuperscript{$\ast$}} &\colorbox{green!10}{7.156 $\rightarrow$ 3.436~\textsuperscript{$\ast$}} &\colorbox{green!10}{6.439 $\rightarrow$ 3.211~\textsuperscript{$\ast$}} \\EigV &\colorbox{green!10}{3.745 $\rightarrow$ 1.168~\textsuperscript{$\ast$}} &\colorbox{green!10}{3.449 $\rightarrow$ 1.131~\textsuperscript{$\ast$}} &\colorbox{green!10}{3.547 $\rightarrow$ 1.119~\textsuperscript{$\ast$}} &\colorbox{green!10}{3.575 $\rightarrow$ 1.191~\textsuperscript{$\ast$}} &\colorbox{green!10}{3.611 $\rightarrow$ 1.179~\textsuperscript{$\ast$}} &\colorbox{green!10}{3.470 $\rightarrow$ 1.210~\textsuperscript{$\ast$}} &\colorbox{green!10}{5.124 $\rightarrow$ 1.054~\textsuperscript{$\ast$}} &\colorbox{green!10}{5.217 $\rightarrow$ 1.055~\textsuperscript{$\ast$}} &\colorbox{green!10}{4.323 $\rightarrow$ 1.040~\textsuperscript{$\ast$}} \\ECC &\colorbox{green!10}{0.653 $\rightarrow$ 0.089~\textsuperscript{$\ast$}} &\colorbox{green!10}{0.633 $\rightarrow$ 0.072~\textsuperscript{$\ast$}} &\colorbox{green!10}{0.661 $\rightarrow$ 0.069~\textsuperscript{$\ast$}} &\colorbox{green!10}{0.756 $\rightarrow$ 0.110~\textsuperscript{$\ast$}} &\colorbox{green!10}{0.752 $\rightarrow$ 0.104~\textsuperscript{$\ast$}} &\colorbox{green!10}{0.747 $\rightarrow$ 0.131~\textsuperscript{$\ast$}} &\colorbox{green!10}{0.854 $\rightarrow$ 0.024~\textsuperscript{$\ast$}} &\colorbox{green!10}{0.841 $\rightarrow$ 0.024~\textsuperscript{$\ast$}} &\colorbox{green!10}{0.700 $\rightarrow$ 0.025~\textsuperscript{$\ast$}} \\Deg &\colorbox{green!10}{0.471 $\rightarrow$ 0.053~\textsuperscript{$\ast$}} &\colorbox{green!10}{0.450 $\rightarrow$ 0.048~\textsuperscript{$\ast$}} &\colorbox{green!10}{0.466 $\rightarrow$ 0.037~\textsuperscript{$\ast$}} &\colorbox{green!10}{0.462 $\rightarrow$ 0.053~\textsuperscript{$\ast$}} &\colorbox{green!10}{0.471 $\rightarrow$ 0.047~\textsuperscript{$\ast$}} &\colorbox{green!10}{0.454 $\rightarrow$ 0.063~\textsuperscript{$\ast$}} &\colorbox{green!10}{0.614 $\rightarrow$ 0.022~\textsuperscript{$\ast$}} &\colorbox{green!10}{0.615 $\rightarrow$ 0.021~\textsuperscript{$\ast$}} &\colorbox{green!10}{0.492 $\rightarrow$ 0.016~\textsuperscript{$\ast$}} \\\hline\multicolumn{6}{l}{\textbf{Axiom 4:} Negatively Changed $\uparrow$} \\ \hline PE &\colorbox{magenta!20}{1.450 $\rightarrow$ 1.284~\textsuperscript{$\ast$}} &\colorbox{magenta!10}{1.570 $\rightarrow$ 1.490~\textsuperscript{ }} &\colorbox{magenta!20}{1.518 $\rightarrow$ 1.256~\textsuperscript{$\ast$}} &\colorbox{green!10}{0.791 $\rightarrow$ 1.173~\textsuperscript{$\ast$}} &\colorbox{green!10}{0.833 $\rightarrow$ 1.144~\textsuperscript{$\ast$}} &\colorbox{magenta!5}{0.881 $\rightarrow$ 1.021~\textsuperscript{ }} &\colorbox{magenta!10}{0.941 $\rightarrow$ 0.881~\textsuperscript{ }} &\colorbox{magenta!10}{1.014 $\rightarrow$ 0.807~\textsuperscript{ }} &\colorbox{magenta!20}{1.660 $\rightarrow$ 0.913~\textsuperscript{$\ast$}} \\SE &\colorbox{magenta!20}{4.957 $\rightarrow$ 4.252~\textsuperscript{$\ast$}} &\colorbox{magenta!20}{5.039 $\rightarrow$ 4.543~\textsuperscript{$\ast$}} &\colorbox{magenta!20}{4.775 $\rightarrow$ 4.116~\textsuperscript{$\ast$}} &\colorbox{green!10}{4.173 $\rightarrow$ 4.356~\textsuperscript{$\ast$}} &\colorbox{magenta!5}{4.212 $\rightarrow$ 4.319~\textsuperscript{ }} &\colorbox{magenta!10}{4.392 $\rightarrow$ 4.174~\textsuperscript{ }} &\colorbox{magenta!20}{4.569 $\rightarrow$ 3.875~\textsuperscript{$\ast$}} &\colorbox{magenta!20}{4.739 $\rightarrow$ 3.709~\textsuperscript{$\ast$}} &\colorbox{magenta!20}{5.853 $\rightarrow$ 3.783~\textsuperscript{$\ast$}} \\PE+M &\colorbox{magenta!10}{3.045 $\rightarrow$ 2.901~\textsuperscript{ }} &\colorbox{magenta!5}{3.421 $\rightarrow$ 3.597~\textsuperscript{ }} &\colorbox{magenta!10}{3.159 $\rightarrow$ 2.924~\textsuperscript{ }} &\colorbox{green!10}{1.383 $\rightarrow$ 1.989~\textsuperscript{$\ast$}} &\colorbox{green!10}{1.361 $\rightarrow$ 2.107~\textsuperscript{$\ast$}} &\colorbox{magenta!5}{1.424 $\rightarrow$ 1.849~\textsuperscript{ }} &\colorbox{magenta!5}{1.323 $\rightarrow$ 1.354~\textsuperscript{ }} &\colorbox{magenta!10}{1.447 $\rightarrow$ 1.303~\textsuperscript{ }} &\colorbox{magenta!20}{2.705 $\rightarrow$ 1.735~\textsuperscript{$\ast$}} \\SE+M &\colorbox{magenta!20}{6.368 $\rightarrow$ 5.630~\textsuperscript{$\ast$}} &\colorbox{magenta!10}{6.674 $\rightarrow$ 6.349~\textsuperscript{ }} &\colorbox{magenta!10}{6.181 $\rightarrow$ 5.549~\textsuperscript{ }} &\colorbox{green!10}{4.743 $\rightarrow$ 5.076~\textsuperscript{$\ast$}} &\colorbox{magenta!5}{4.720 $\rightarrow$ 5.081~\textsuperscript{ }} &\colorbox{magenta!10}{4.954 $\rightarrow$ 4.796~\textsuperscript{ }} &\colorbox{magenta!20}{4.958 $\rightarrow$ 4.241~\textsuperscript{$\ast$}} &\colorbox{magenta!20}{5.184 $\rightarrow$ 4.062~\textsuperscript{$\ast$}} &\colorbox{magenta!20}{6.835 $\rightarrow$ 4.446~\textsuperscript{$\ast$}} \\EigV &\colorbox{magenta!20}{2.087 $\rightarrow$ 1.415~\textsuperscript{$\ast$}} &\colorbox{magenta!20}{2.115 $\rightarrow$ 1.497~\textsuperscript{$\ast$}} &\colorbox{magenta!20}{1.906 $\rightarrow$ 1.375~\textsuperscript{$\ast$}} &\colorbox{magenta!10}{1.850 $\rightarrow$ 1.593~\textsuperscript{ }} &\colorbox{magenta!10}{1.944 $\rightarrow$ 1.710~\textsuperscript{ }} &\colorbox{magenta!20}{2.103 $\rightarrow$ 1.594~\textsuperscript{$\ast$}} &\colorbox{magenta!20}{2.522 $\rightarrow$ 1.200~\textsuperscript{$\ast$}} &\colorbox{magenta!20}{2.565 $\rightarrow$ 1.222~\textsuperscript{$\ast$}} &\colorbox{magenta!20}{4.209 $\rightarrow$ 1.159~\textsuperscript{$\ast$}} \\ECC &\colorbox{magenta!20}{0.440 $\rightarrow$ 0.208~\textsuperscript{$\ast$}} &\colorbox{magenta!20}{0.438 $\rightarrow$ 0.238~\textsuperscript{$\ast$}} &\colorbox{magenta!20}{0.351 $\rightarrow$ 0.151~\textsuperscript{$\ast$}} &\colorbox{magenta!10}{0.362 $\rightarrow$ 0.293~\textsuperscript{ }} &\colorbox{magenta!10}{0.378 $\rightarrow$ 0.323~\textsuperscript{ }} &\colorbox{magenta!10}{0.420 $\rightarrow$ 0.298~\textsuperscript{ }} &\colorbox{magenta!20}{0.437 $\rightarrow$ 0.091~\textsuperscript{$\ast$}} &\colorbox{magenta!20}{0.492 $\rightarrow$ 0.119~\textsuperscript{$\ast$}} &\colorbox{magenta!20}{0.700 $\rightarrow$ 0.068~\textsuperscript{$\ast$}} \\Deg &\colorbox{magenta!20}{0.243 $\rightarrow$ 0.138~\textsuperscript{$\ast$}} &\colorbox{magenta!20}{0.252 $\rightarrow$ 0.149~\textsuperscript{$\ast$}} &\colorbox{magenta!20}{0.222 $\rightarrow$ 0.109~\textsuperscript{$\ast$}} &\colorbox{magenta!5}{0.175 $\rightarrow$ 0.183~\textsuperscript{ }} &\colorbox{magenta!5}{0.190 $\rightarrow$ 0.203~\textsuperscript{ }} &\colorbox{magenta!10}{0.233 $\rightarrow$ 0.180~\textsuperscript{ }} &\colorbox{magenta!20}{0.259 $\rightarrow$ 0.091~\textsuperscript{$\ast$}} &\colorbox{magenta!20}{0.280 $\rightarrow$ 0.092~\textsuperscript{$\ast$}} &\colorbox{magenta!20}{0.479 $\rightarrow$ 0.065~\textsuperscript{$\ast$}} \\\hline

\end{tabular}
\caption{
Comparison of changes in average uncertainty values for Axioms 1--4 before (left) and after (right) applying RAG with Mistral-v0.3.
Axioms are implemented using the \textit{Reference-based} method.
Colors \colorbox{green!10}{green} and \colorbox{magenta!20}{deep red} indicate significant changes aligning or conflicting with axioms, respectively. 
Color \colorbox{magenta!10}{shallow red} represents non-significant changes conflicting with axioms. 
Significance is marked by $\ast$.}
\label{tab:axioms_rel_correctness_mistral}
\end{table*}

\begin{table}[h!]
\centering
\setlength{\tabcolsep}{3pt}
\scriptsize
\begin{tabular}{c|ccc}
\hline
\textbf{Unc.} & \textbf{NQ-open} & \textbf{TriviaQA} & \textbf{PopQA} \\
\hline\hline

PE &
\colorbox{magenta!20}{2.227 $\rightarrow$ 1.778~\textsuperscript{$\ast$}} &
\colorbox{magenta!20}{0.657 $\rightarrow$ 0.780~\textsuperscript{$\ast$}} &
\colorbox{green!10}{1.014 $\rightarrow$ 1.087~\textsuperscript{ }} \\

SE &
\colorbox{magenta!20}{5.453 $\rightarrow$ 4.964~\textsuperscript{$\ast$}} &
\colorbox{magenta!20}{3.570 $\rightarrow$ 3.892~\textsuperscript{$\ast$}} &
\colorbox{green!10}{3.976 $\rightarrow$ 4.021~\textsuperscript{ }} \\

PE+M &
\colorbox{magenta!20}{5.634 $\rightarrow$ 4.293~\textsuperscript{$\ast$}} &
\colorbox{magenta!20}{1.223 $\rightarrow$ 1.374~\textsuperscript{$\ast$}} &
\colorbox{green!10}{1.686 $\rightarrow$ 1.759~\textsuperscript{ }} \\

SE+M &
\colorbox{magenta!20}{8.543 $\rightarrow$ 7.216~\textsuperscript{$\ast$}} &
\colorbox{magenta!20}{4.089 $\rightarrow$ 4.463~\textsuperscript{$\ast$}} &
\colorbox{green!10}{4.310 $\rightarrow$ 4.521~\textsuperscript{ }} \\

EigV &
\colorbox{green!10}{1.696 $\rightarrow$ 1.637~\textsuperscript{ }} &
\colorbox{magenta!20}{1.256 $\rightarrow$ 1.496~\textsuperscript{$\ast$}} &
\colorbox{green!10}{1.215 $\rightarrow$ 1.452~\textsuperscript{ }} \\

ECC &
\colorbox{green!10}{0.357 $\rightarrow$ 0.335~\textsuperscript{ }} &
\colorbox{magenta!20}{0.154 $\rightarrow$ 0.300~\textsuperscript{$\ast$}} &
\colorbox{green!10}{0.059 $\rightarrow$ 0.362~\textsuperscript{ }} \\

Deg &
\colorbox{magenta!20}{0.160 $\rightarrow$ 0.206~\textsuperscript{$\ast$}} &
\colorbox{magenta!20}{0.056 $\rightarrow$ 0.128~\textsuperscript{$\ast$}} &
\colorbox{green!10}{0.093 $\rightarrow$ 0.140~\textsuperscript{ }} \\\hline

\end{tabular}
\shrink
\caption{
Comparison of changes in average uncertainty values for Axiom 5 before (left) and after(right) applying RAG with Mistral-v0.3.
Color coding and significance markers follow those in Table~\ref{tab:axioms_rel_correctness_mistral}.}
\label{tab:Axiom_irrelevant_mistral}
\shrink
\shrink
\end{table}

\begin{table*}[t]
\centering
\setlength{\tabcolsep}{0.0pt}
\tiny
\begin{tabular}{c|ccc|ccc|ccc}
\hline
\textbf{UE} &
\multicolumn{3}{c}{\textbf{NQ-open}} & \multicolumn{3}{c}{\textbf{TriviaQA}} & \multicolumn{3}{c}{\textbf{PopQA}} \\ \hline %
& \textbf{BM25} & \textbf{Contriever} & \textbf{$\text{Doc}^{+}$} & 
\textbf{BM25} & \textbf{Contriever} & \textbf{$\text{Doc}^{+}$} & 
\textbf{BM25} & \textbf{Contriever} & \textbf{$\text{Doc}^{+}$} \\
\hline\hline

\multicolumn{6}{l}{\textbf{Axiom 1:} Positively Consistent $\downarrow$} \\ \hline PE &\colorbox{magenta!5}{1.896 $\rightarrow$ 1.802~\textsuperscript{ }} &\colorbox{magenta!5}{1.801 $\rightarrow$ 1.642~\textsuperscript{ }} &\colorbox{green!10}{1.684 $\rightarrow$ 1.500~\textsuperscript{$\ast$}} &\colorbox{magenta!20}{0.796 $\rightarrow$ 0.848~\textsuperscript{$\ast$}} &\colorbox{magenta!20}{0.844 $\rightarrow$ 0.877~\textsuperscript{$\ast$}} &\colorbox{magenta!20}{0.929 $\rightarrow$ 0.952~\textsuperscript{$\ast$}} &\colorbox{magenta!5}{0.798 $\rightarrow$ 0.418~\textsuperscript{ }} &\colorbox{green!10}{0.715 $\rightarrow$ 0.416~\textsuperscript{$\ast$}} &\colorbox{magenta!5}{0.818 $\rightarrow$ 0.191~\textsuperscript{ }} \\SE &\colorbox{green!10}{5.174 $\rightarrow$ 4.524~\textsuperscript{$\ast$}} &\colorbox{green!10}{5.071 $\rightarrow$ 4.344~\textsuperscript{$\ast$}} &\colorbox{green!10}{4.957 $\rightarrow$ 4.145~\textsuperscript{$\ast$}} &\colorbox{magenta!5}{3.779 $\rightarrow$ 3.533~\textsuperscript{ }} &\colorbox{green!10}{3.823 $\rightarrow$ 3.569~\textsuperscript{$\ast$}} &\colorbox{green!10}{4.019 $\rightarrow$ 3.725~\textsuperscript{$\ast$}} &\colorbox{green!10}{3.869 $\rightarrow$ 3.260~\textsuperscript{$\ast$}} &\colorbox{green!10}{3.805 $\rightarrow$ 3.152~\textsuperscript{$\ast$}} &\colorbox{green!10}{3.700 $\rightarrow$ 3.053~\textsuperscript{$\ast$}} \\PE+M &\colorbox{green!10}{4.445 $\rightarrow$ 4.152~\textsuperscript{$\ast$}} &\colorbox{green!10}{4.162 $\rightarrow$ 4.013~\textsuperscript{$\ast$}} &\colorbox{magenta!5}{4.090 $\rightarrow$ 4.039~\textsuperscript{ }} &\colorbox{magenta!20}{1.307 $\rightarrow$ 1.331~\textsuperscript{$\ast$}} &\colorbox{magenta!10}{1.368 $\rightarrow$ 1.392~\textsuperscript{ }} &\colorbox{magenta!5}{1.564 $\rightarrow$ 1.559~\textsuperscript{ }} &\colorbox{magenta!5}{0.930 $\rightarrow$ 0.490~\textsuperscript{ }} &\colorbox{green!10}{0.820 $\rightarrow$ 0.486~\textsuperscript{$\ast$}} &\colorbox{green!10}{0.846 $\rightarrow$ 0.213~\textsuperscript{$\ast$}} \\SE+M &\colorbox{green!10}{7.716 $\rightarrow$ 6.855~\textsuperscript{$\ast$}} &\colorbox{green!10}{7.483 $\rightarrow$ 6.619~\textsuperscript{$\ast$}} &\colorbox{green!10}{7.380 $\rightarrow$ 6.577~\textsuperscript{$\ast$}} &\colorbox{green!10}{4.422 $\rightarrow$ 4.062~\textsuperscript{$\ast$}} &\colorbox{green!10}{4.495 $\rightarrow$ 4.142~\textsuperscript{$\ast$}} &\colorbox{green!10}{4.783 $\rightarrow$ 4.381~\textsuperscript{$\ast$}} &\colorbox{green!10}{4.185 $\rightarrow$ 3.354~\textsuperscript{$\ast$}} &\colorbox{green!10}{4.090 $\rightarrow$ 3.265~\textsuperscript{$\ast$}} &\colorbox{green!10}{3.909 $\rightarrow$ 3.080~\textsuperscript{$\ast$}} \\EigV &\colorbox{green!10}{2.248 $\rightarrow$ 1.451~\textsuperscript{$\ast$}} &\colorbox{green!10}{2.264 $\rightarrow$ 1.236~\textsuperscript{$\ast$}} &\colorbox{green!10}{2.183 $\rightarrow$ 1.105~\textsuperscript{$\ast$}} &\colorbox{green!10}{1.656 $\rightarrow$ 1.375~\textsuperscript{$\ast$}} &\colorbox{green!10}{1.704 $\rightarrow$ 1.296~\textsuperscript{$\ast$}} &\colorbox{green!10}{1.913 $\rightarrow$ 1.583~\textsuperscript{$\ast$}} &\colorbox{green!10}{2.088 $\rightarrow$ 1.215~\textsuperscript{$\ast$}} &\colorbox{green!10}{2.030 $\rightarrow$ 1.175~\textsuperscript{$\ast$}} &\colorbox{green!10}{2.126 $\rightarrow$ 1.145~\textsuperscript{$\ast$}} \\ECC &\colorbox{green!10}{0.546 $\rightarrow$ 0.224~\textsuperscript{$\ast$}} &\colorbox{green!10}{0.583 $\rightarrow$ 0.163~\textsuperscript{$\ast$}} &\colorbox{green!10}{0.548 $\rightarrow$ 0.080~\textsuperscript{$\ast$}} &\colorbox{green!10}{0.353 $\rightarrow$ 0.237~\textsuperscript{$\ast$}} &\colorbox{green!10}{0.369 $\rightarrow$ 0.187~\textsuperscript{$\ast$}} &\colorbox{green!10}{0.422 $\rightarrow$ 0.289~\textsuperscript{$\ast$}} &\colorbox{green!10}{0.447 $\rightarrow$ 0.133~\textsuperscript{$\ast$}} &\colorbox{green!10}{0.432 $\rightarrow$ 0.093~\textsuperscript{$\ast$}} &\colorbox{green!10}{0.405 $\rightarrow$ 0.071~\textsuperscript{$\ast$}} \\Deg &\colorbox{green!10}{0.264 $\rightarrow$ 0.123~\textsuperscript{$\ast$}} &\colorbox{green!10}{0.277 $\rightarrow$ 0.075~\textsuperscript{$\ast$}} &\colorbox{green!10}{0.262 $\rightarrow$ 0.032~\textsuperscript{$\ast$}} &\colorbox{green!10}{0.153 $\rightarrow$ 0.097~\textsuperscript{$\ast$}} &\colorbox{green!10}{0.161 $\rightarrow$ 0.081~\textsuperscript{$\ast$}} &\colorbox{green!10}{0.201 $\rightarrow$ 0.134~\textsuperscript{$\ast$}} &\colorbox{green!10}{0.222 $\rightarrow$ 0.058~\textsuperscript{$\ast$}} &\colorbox{green!10}{0.218 $\rightarrow$ 0.051~\textsuperscript{$\ast$}} &\colorbox{green!10}{0.236 $\rightarrow$ 0.040~\textsuperscript{$\ast$}} \\\hline\multicolumn{6}{l}{\textbf{Axiom 2:} Negatively Consistent $\uparrow$} \\ \hline PE &\colorbox{magenta!5}{1.978 $\rightarrow$ 2.037~\textsuperscript{ }} &\colorbox{magenta!10}{1.717 $\rightarrow$ 1.507~\textsuperscript{ }} &\colorbox{magenta!20}{1.705 $\rightarrow$ 1.173~\textsuperscript{$\ast$}} &\colorbox{magenta!5}{0.783 $\rightarrow$ 0.833~\textsuperscript{ }} &\colorbox{magenta!10}{0.794 $\rightarrow$ 0.718~\textsuperscript{ }} &\colorbox{magenta!10}{0.795 $\rightarrow$ 0.749~\textsuperscript{ }} &\colorbox{magenta!10}{0.817 $\rightarrow$ 0.583~\textsuperscript{ }} &\colorbox{magenta!10}{0.698 $\rightarrow$ 0.310~\textsuperscript{ }} &\colorbox{magenta!10}{1.296 $\rightarrow$ 0.528~\textsuperscript{ }} \\SE &\colorbox{magenta!10}{5.499 $\rightarrow$ 5.108~\textsuperscript{ }} &\colorbox{magenta!20}{5.039 $\rightarrow$ 4.210~\textsuperscript{$\ast$}} &\colorbox{magenta!20}{5.034 $\rightarrow$ 3.845~\textsuperscript{$\ast$}} &\colorbox{magenta!5}{3.707 $\rightarrow$ 3.740~\textsuperscript{ }} &\colorbox{magenta!10}{3.744 $\rightarrow$ 3.521~\textsuperscript{ }} &\colorbox{magenta!10}{3.897 $\rightarrow$ 3.568~\textsuperscript{ }} &\colorbox{magenta!10}{3.570 $\rightarrow$ 3.238~\textsuperscript{ }} &\colorbox{magenta!20}{3.698 $\rightarrow$ 3.119~\textsuperscript{$\ast$}} &\colorbox{magenta!10}{4.233 $\rightarrow$ 2.986~\textsuperscript{ }} \\PE+M &\colorbox{magenta!10}{4.707 $\rightarrow$ 4.579~\textsuperscript{ }} &\colorbox{magenta!10}{3.438 $\rightarrow$ 3.295~\textsuperscript{ }} &\colorbox{magenta!10}{3.483 $\rightarrow$ 3.027~\textsuperscript{ }} &\colorbox{magenta!5}{1.029 $\rightarrow$ 1.185~\textsuperscript{ }} &\colorbox{magenta!10}{1.047 $\rightarrow$ 0.987~\textsuperscript{ }} &\colorbox{magenta!5}{1.024 $\rightarrow$ 1.030~\textsuperscript{ }} &\colorbox{magenta!10}{0.817 $\rightarrow$ 0.542~\textsuperscript{ }} &\colorbox{magenta!10}{0.640 $\rightarrow$ 0.332~\textsuperscript{ }} &\colorbox{magenta!10}{1.262 $\rightarrow$ 0.637~\textsuperscript{ }} \\SE+M &\colorbox{magenta!10}{8.217 $\rightarrow$ 7.579~\textsuperscript{ }} &\colorbox{magenta!20}{6.970 $\rightarrow$ 6.014~\textsuperscript{$\ast$}} &\colorbox{magenta!20}{6.959 $\rightarrow$ 5.525~\textsuperscript{$\ast$}} &\colorbox{magenta!5}{4.103 $\rightarrow$ 4.149~\textsuperscript{ }} &\colorbox{magenta!10}{4.198 $\rightarrow$ 3.894~\textsuperscript{ }} &\colorbox{magenta!20}{4.351 $\rightarrow$ 3.892~\textsuperscript{$\ast$}} &\colorbox{magenta!20}{3.771 $\rightarrow$ 3.321~\textsuperscript{$\ast$}} &\colorbox{magenta!20}{3.854 $\rightarrow$ 3.181~\textsuperscript{$\ast$}} &\colorbox{magenta!10}{4.458 $\rightarrow$ 3.029~\textsuperscript{ }} \\EigV &\colorbox{magenta!10}{2.563 $\rightarrow$ 2.233~\textsuperscript{ }} &\colorbox{magenta!20}{2.610 $\rightarrow$ 1.464~\textsuperscript{$\ast$}} &\colorbox{magenta!20}{2.236 $\rightarrow$ 1.192~\textsuperscript{$\ast$}} &\colorbox{magenta!20}{1.811 $\rightarrow$ 1.537~\textsuperscript{$\ast$}} &\colorbox{magenta!20}{1.804 $\rightarrow$ 1.426~\textsuperscript{$\ast$}} &\colorbox{magenta!20}{1.970 $\rightarrow$ 1.632~\textsuperscript{$\ast$}} &\colorbox{magenta!20}{1.911 $\rightarrow$ 1.217~\textsuperscript{$\ast$}} &\colorbox{magenta!20}{2.015 $\rightarrow$ 1.175~\textsuperscript{$\ast$}} &\colorbox{magenta!10}{2.998 $\rightarrow$ 1.214~\textsuperscript{ }} \\ECC &\colorbox{magenta!10}{0.580 $\rightarrow$ 0.403~\textsuperscript{ }} &\colorbox{magenta!20}{0.612 $\rightarrow$ 0.294~\textsuperscript{$\ast$}} &\colorbox{magenta!20}{0.626 $\rightarrow$ 0.169~\textsuperscript{$\ast$}} &\colorbox{magenta!10}{0.419 $\rightarrow$ 0.359~\textsuperscript{ }} &\colorbox{magenta!10}{0.390 $\rightarrow$ 0.283~\textsuperscript{ }} &\colorbox{magenta!20}{0.450 $\rightarrow$ 0.319~\textsuperscript{$\ast$}} &\colorbox{magenta!20}{0.406 $\rightarrow$ 0.142~\textsuperscript{$\ast$}} &\colorbox{magenta!20}{0.473 $\rightarrow$ 0.145~\textsuperscript{$\ast$}} &\colorbox{magenta!10}{0.667 $\rightarrow$ 0.058~\textsuperscript{ }} \\Deg &\colorbox{magenta!10}{0.308 $\rightarrow$ 0.255~\textsuperscript{ }} &\colorbox{magenta!20}{0.331 $\rightarrow$ 0.125~\textsuperscript{$\ast$}} &\colorbox{magenta!20}{0.269 $\rightarrow$ 0.063~\textsuperscript{$\ast$}} &\colorbox{magenta!10}{0.177 $\rightarrow$ 0.141~\textsuperscript{ }} &\colorbox{magenta!20}{0.173 $\rightarrow$ 0.111~\textsuperscript{$\ast$}} &\colorbox{magenta!20}{0.215 $\rightarrow$ 0.139~\textsuperscript{$\ast$}} &\colorbox{magenta!20}{0.217 $\rightarrow$ 0.077~\textsuperscript{$\ast$}} &\colorbox{magenta!20}{0.217 $\rightarrow$ 0.052~\textsuperscript{$\ast$}} &\colorbox{magenta!10}{0.407 $\rightarrow$ 0.093~\textsuperscript{ }} \\\hline\multicolumn{6}{l}{\textbf{Axiom 3:} Positively Changed $\downarrow$} \\ \hline PE &\colorbox{green!10}{1.800 $\rightarrow$ 1.239~\textsuperscript{$\ast$}} &\colorbox{green!10}{1.860 $\rightarrow$ 1.261~\textsuperscript{$\ast$}} &\colorbox{green!10}{1.816 $\rightarrow$ 1.063~\textsuperscript{$\ast$}} &\colorbox{green!10}{1.239 $\rightarrow$ 0.749~\textsuperscript{$\ast$}} &\colorbox{green!10}{1.287 $\rightarrow$ 0.851~\textsuperscript{$\ast$}} &\colorbox{green!10}{1.332 $\rightarrow$ 0.686~\textsuperscript{$\ast$}} &\colorbox{green!10}{1.348 $\rightarrow$ 0.397~\textsuperscript{$\ast$}} &\colorbox{green!10}{1.386 $\rightarrow$ 0.368~\textsuperscript{$\ast$}} &\colorbox{green!10}{1.358 $\rightarrow$ 0.264~\textsuperscript{$\ast$}} \\SE &\colorbox{green!10}{5.575 $\rightarrow$ 4.025~\textsuperscript{$\ast$}} &\colorbox{green!10}{5.603 $\rightarrow$ 4.055~\textsuperscript{$\ast$}} &\colorbox{green!10}{5.685 $\rightarrow$ 3.742~\textsuperscript{$\ast$}} &\colorbox{green!10}{4.773 $\rightarrow$ 3.511~\textsuperscript{$\ast$}} &\colorbox{green!10}{4.908 $\rightarrow$ 3.641~\textsuperscript{$\ast$}} &\colorbox{green!10}{5.029 $\rightarrow$ 3.312~\textsuperscript{$\ast$}} &\colorbox{green!10}{5.092 $\rightarrow$ 3.161~\textsuperscript{$\ast$}} &\colorbox{green!10}{5.203 $\rightarrow$ 3.135~\textsuperscript{$\ast$}} &\colorbox{green!10}{4.987 $\rightarrow$ 3.050~\textsuperscript{$\ast$}} \\PE+M &\colorbox{green!10}{3.630 $\rightarrow$ 3.003~\textsuperscript{$\ast$}} &\colorbox{green!10}{3.770 $\rightarrow$ 3.028~\textsuperscript{$\ast$}} &\colorbox{green!10}{3.704 $\rightarrow$ 2.785~\textsuperscript{$\ast$}} &\colorbox{green!10}{1.809 $\rightarrow$ 1.112~\textsuperscript{$\ast$}} &\colorbox{green!10}{1.943 $\rightarrow$ 1.297~\textsuperscript{$\ast$}} &\colorbox{green!10}{1.835 $\rightarrow$ 1.061~\textsuperscript{$\ast$}} &\colorbox{green!10}{1.723 $\rightarrow$ 0.470~\textsuperscript{$\ast$}} &\colorbox{green!10}{1.766 $\rightarrow$ 0.436~\textsuperscript{$\ast$}} &\colorbox{green!10}{1.655 $\rightarrow$ 0.331~\textsuperscript{$\ast$}} \\SE+M &\colorbox{green!10}{7.504 $\rightarrow$ 5.709~\textsuperscript{$\ast$}} &\colorbox{green!10}{7.565 $\rightarrow$ 5.705~\textsuperscript{$\ast$}} &\colorbox{green!10}{7.693 $\rightarrow$ 5.285~\textsuperscript{$\ast$}} &\colorbox{green!10}{5.504 $\rightarrow$ 3.907~\textsuperscript{$\ast$}} &\colorbox{green!10}{5.740 $\rightarrow$ 4.134~\textsuperscript{$\ast$}} &\colorbox{green!10}{5.771 $\rightarrow$ 3.681~\textsuperscript{$\ast$}} &\colorbox{green!10}{5.691 $\rightarrow$ 3.262~\textsuperscript{$\ast$}} &\colorbox{green!10}{5.782 $\rightarrow$ 3.228~\textsuperscript{$\ast$}} &\colorbox{green!10}{5.506 $\rightarrow$ 3.102~\textsuperscript{$\ast$}} \\EigV &\colorbox{green!10}{3.693 $\rightarrow$ 1.335~\textsuperscript{$\ast$}} &\colorbox{green!10}{3.822 $\rightarrow$ 1.321~\textsuperscript{$\ast$}} &\colorbox{green!10}{3.947 $\rightarrow$ 1.149~\textsuperscript{$\ast$}} &\colorbox{green!10}{3.377 $\rightarrow$ 1.332~\textsuperscript{$\ast$}} &\colorbox{green!10}{3.626 $\rightarrow$ 1.411~\textsuperscript{$\ast$}} &\colorbox{green!10}{3.840 $\rightarrow$ 1.281~\textsuperscript{$\ast$}} &\colorbox{green!10}{4.772 $\rightarrow$ 1.222~\textsuperscript{$\ast$}} &\colorbox{green!10}{5.100 $\rightarrow$ 1.197~\textsuperscript{$\ast$}} &\colorbox{green!10}{4.622 $\rightarrow$ 1.102~\textsuperscript{$\ast$}} \\ECC &\colorbox{green!10}{0.762 $\rightarrow$ 0.203~\textsuperscript{$\ast$}} &\colorbox{green!10}{0.760 $\rightarrow$ 0.207~\textsuperscript{$\ast$}} &\colorbox{green!10}{0.817 $\rightarrow$ 0.098~\textsuperscript{$\ast$}} &\colorbox{green!10}{0.738 $\rightarrow$ 0.213~\textsuperscript{$\ast$}} &\colorbox{green!10}{0.796 $\rightarrow$ 0.261~\textsuperscript{$\ast$}} &\colorbox{green!10}{0.845 $\rightarrow$ 0.162~\textsuperscript{$\ast$}} &\colorbox{green!10}{0.814 $\rightarrow$ 0.135~\textsuperscript{$\ast$}} &\colorbox{green!10}{0.855 $\rightarrow$ 0.125~\textsuperscript{$\ast$}} &\colorbox{green!10}{0.806 $\rightarrow$ 0.065~\textsuperscript{$\ast$}} \\Deg &\colorbox{green!10}{0.494 $\rightarrow$ 0.105~\textsuperscript{$\ast$}} &\colorbox{green!10}{0.517 $\rightarrow$ 0.100~\textsuperscript{$\ast$}} &\colorbox{green!10}{0.538 $\rightarrow$ 0.048~\textsuperscript{$\ast$}} &\colorbox{green!10}{0.460 $\rightarrow$ 0.097~\textsuperscript{$\ast$}} &\colorbox{green!10}{0.494 $\rightarrow$ 0.121~\textsuperscript{$\ast$}} &\colorbox{green!10}{0.525 $\rightarrow$ 0.076~\textsuperscript{$\ast$}} &\colorbox{green!10}{0.593 $\rightarrow$ 0.069~\textsuperscript{$\ast$}} &\colorbox{green!10}{0.630 $\rightarrow$ 0.059~\textsuperscript{$\ast$}} &\colorbox{green!10}{0.569 $\rightarrow$ 0.032~\textsuperscript{$\ast$}} \\\hline\multicolumn{6}{l}{\textbf{Axiom 4:} Negatively Changed $\uparrow$} \\ \hline PE &\colorbox{magenta!10}{2.027 $\rightarrow$ 1.829~\textsuperscript{ }} &\colorbox{magenta!20}{2.245 $\rightarrow$ 1.342~\textsuperscript{$\ast$}} &\colorbox{magenta!20}{2.423 $\rightarrow$ 1.386~\textsuperscript{$\ast$}} &\colorbox{magenta!10}{1.139 $\rightarrow$ 1.017~\textsuperscript{ }} &\colorbox{magenta!10}{1.017 $\rightarrow$ 0.911~\textsuperscript{ }} &\colorbox{magenta!10}{1.427 $\rightarrow$ 1.067~\textsuperscript{ }} &\colorbox{magenta!10}{1.248 $\rightarrow$ 0.874~\textsuperscript{ }} &\colorbox{magenta!10}{1.600 $\rightarrow$ 0.543~\textsuperscript{ }} &\colorbox{magenta!10}{1.964 $\rightarrow$ 0.223~\textsuperscript{ }} \\SE &\colorbox{magenta!20}{5.476 $\rightarrow$ 4.683~\textsuperscript{$\ast$}} &\colorbox{magenta!20}{5.494 $\rightarrow$ 4.245~\textsuperscript{$\ast$}} &\colorbox{magenta!20}{5.689 $\rightarrow$ 4.419~\textsuperscript{$\ast$}} &\colorbox{magenta!20}{4.626 $\rightarrow$ 4.176~\textsuperscript{$\ast$}} &\colorbox{magenta!20}{4.554 $\rightarrow$ 4.028~\textsuperscript{$\ast$}} &\colorbox{magenta!20}{4.523 $\rightarrow$ 3.697~\textsuperscript{$\ast$}} &\colorbox{magenta!20}{4.941 $\rightarrow$ 3.822~\textsuperscript{$\ast$}} &\colorbox{magenta!20}{4.678 $\rightarrow$ 3.879~\textsuperscript{$\ast$}} &\colorbox{magenta!20}{5.367 $\rightarrow$ 3.435~\textsuperscript{$\ast$}} \\PE+M &\colorbox{magenta!10}{3.922 $\rightarrow$ 3.817~\textsuperscript{ }} &\colorbox{magenta!20}{3.501 $\rightarrow$ 2.822~\textsuperscript{$\ast$}} &\colorbox{magenta!20}{4.112 $\rightarrow$ 3.021~\textsuperscript{$\ast$}} &\colorbox{magenta!10}{1.649 $\rightarrow$ 1.611~\textsuperscript{ }} &\colorbox{magenta!5}{1.465 $\rightarrow$ 1.570~\textsuperscript{ }} &\colorbox{magenta!10}{1.646 $\rightarrow$ 1.313~\textsuperscript{ }} &\colorbox{magenta!10}{1.634 $\rightarrow$ 1.153~\textsuperscript{ }} &\colorbox{magenta!10}{1.784 $\rightarrow$ 0.621~\textsuperscript{ }} &\colorbox{magenta!10}{2.302 $\rightarrow$ 0.339~\textsuperscript{ }} \\SE+M &\colorbox{magenta!20}{7.532 $\rightarrow$ 6.421~\textsuperscript{$\ast$}} &\colorbox{magenta!20}{7.092 $\rightarrow$ 5.728~\textsuperscript{$\ast$}} &\colorbox{magenta!20}{7.660 $\rightarrow$ 6.024~\textsuperscript{$\ast$}} &\colorbox{magenta!20}{5.387 $\rightarrow$ 4.771~\textsuperscript{$\ast$}} &\colorbox{magenta!20}{5.256 $\rightarrow$ 4.773~\textsuperscript{$\ast$}} &\colorbox{magenta!20}{5.135 $\rightarrow$ 4.003~\textsuperscript{$\ast$}} &\colorbox{magenta!20}{5.530 $\rightarrow$ 4.164~\textsuperscript{$\ast$}} &\colorbox{magenta!20}{5.171 $\rightarrow$ 4.041~\textsuperscript{$\ast$}} &\colorbox{magenta!20}{5.972 $\rightarrow$ 3.593~\textsuperscript{$\ast$}} \\EigV &\colorbox{magenta!20}{2.876 $\rightarrow$ 1.754~\textsuperscript{$\ast$}} &\colorbox{magenta!20}{3.040 $\rightarrow$ 1.550~\textsuperscript{$\ast$}} &\colorbox{magenta!20}{2.791 $\rightarrow$ 1.729~\textsuperscript{$\ast$}} &\colorbox{magenta!20}{2.919 $\rightarrow$ 1.995~\textsuperscript{$\ast$}} &\colorbox{magenta!20}{2.983 $\rightarrow$ 1.780~\textsuperscript{$\ast$}} &\colorbox{magenta!10}{2.887 $\rightarrow$ 2.134~\textsuperscript{ }} &\colorbox{magenta!20}{3.995 $\rightarrow$ 1.683~\textsuperscript{$\ast$}} &\colorbox{magenta!20}{4.122 $\rightarrow$ 1.840~\textsuperscript{$\ast$}} &\colorbox{magenta!20}{5.520 $\rightarrow$ 1.421~\textsuperscript{$\ast$}} \\ECC &\colorbox{magenta!20}{0.685 $\rightarrow$ 0.343~\textsuperscript{$\ast$}} &\colorbox{magenta!20}{0.641 $\rightarrow$ 0.307~\textsuperscript{$\ast$}} &\colorbox{magenta!10}{0.505 $\rightarrow$ 0.395~\textsuperscript{ }} &\colorbox{magenta!20}{0.705 $\rightarrow$ 0.499~\textsuperscript{$\ast$}} &\colorbox{magenta!20}{0.700 $\rightarrow$ 0.434~\textsuperscript{$\ast$}} &\colorbox{magenta!20}{0.741 $\rightarrow$ 0.433~\textsuperscript{$\ast$}} &\colorbox{magenta!20}{0.755 $\rightarrow$ 0.333~\textsuperscript{$\ast$}} &\colorbox{magenta!20}{0.799 $\rightarrow$ 0.429~\textsuperscript{$\ast$}} &\colorbox{magenta!20}{0.917 $\rightarrow$ 0.245~\textsuperscript{$\ast$}} \\Deg &\colorbox{magenta!20}{0.417 $\rightarrow$ 0.229~\textsuperscript{$\ast$}} &\colorbox{magenta!20}{0.442 $\rightarrow$ 0.171~\textsuperscript{$\ast$}} &\colorbox{magenta!20}{0.426 $\rightarrow$ 0.199~\textsuperscript{$\ast$}} &\colorbox{magenta!20}{0.384 $\rightarrow$ 0.253~\textsuperscript{$\ast$}} &\colorbox{magenta!20}{0.397 $\rightarrow$ 0.215~\textsuperscript{$\ast$}} &\colorbox{magenta!10}{0.405 $\rightarrow$ 0.269~\textsuperscript{ }} &\colorbox{magenta!20}{0.508 $\rightarrow$ 0.215~\textsuperscript{$\ast$}} &\colorbox{magenta!20}{0.546 $\rightarrow$ 0.199~\textsuperscript{$\ast$}} &\colorbox{magenta!20}{0.688 $\rightarrow$ 0.120~\textsuperscript{$\ast$}} \\\hline

\end{tabular}
\caption{
Comparison of changes in average uncertainty values for Axioms 1--4 before (left) and after (right) applying RAG with Llama2-chat.
Axioms are implemented using the \textit{Reference-free} method.
Colors \colorbox{green!10}{green} and \colorbox{magenta!20}{deep red} indicate significant changes aligning or conflicting with axioms, respectively. 
Color \colorbox{magenta!10}{shallow red} represents non-significant changes conflicting with axioms. Significance is marked by $\ast$.}
\label{tab:axioms_rel_nli_llama}
\end{table*}

\renewcommand{\arraystretch}{1.25}
\begin{table*}[h!]
\centering
\setlength{\tabcolsep}{3pt}
\scriptsize
\begin{tabular}{l|ccccc|ccccc|ccccc}
\hline
\textbf{UE} &
\multicolumn{5}{c}{\textbf{NQ-open}} & \multicolumn{5}{c}{\textbf{TriviaQA}} & \multicolumn{5}{c}{\textbf{PopQA}} \\ \hline
& \textbf{A1} (\%) & \textbf{A2} (\%) & \textbf{A3} (\%) & \textbf{A4} (\%) & \textbf{AUROC}
& \textbf{A1} (\%) & \textbf{A2} (\%) & \textbf{A3} (\%) & \textbf{A4} (\%) & \textbf{AUROC}
& \textbf{A1} (\%) & \textbf{A2} (\%) & \textbf{A3} (\%) & \textbf{A4} (\%) & \textbf{AUROC} \\ \hline\hline

PE   & 
69.77 & \underline{42.95} & 80.54 & 45.21 & 64.40 &
63.88 & 39.05 & 87.24 & 61.30 & 79.14 &
70.57 & 22.22 & 89.17 & 44.23 & 65.72 \\

+CTI &
67.91 & 40.39 & 78.60 & 48.85 & 64.82 &
65.67 & 37.87 & 86.01 & 64.04 & 79.90 &
72.91 & 14.82 & 89.63 & 46.80 & 67.14 \\

+NLI &
70.81 & 42.31 & 80.54 & \underline{52.81} & 66.50 &
65.88 & 42.01 & 86.42 & 64.04 & 79.78 &
\underline{74.23} & \underline{24.69} & \underline{93.09} & 46.80 & \underline{68.65} \\

+MCH & 
\underline{78.05} & 32.05 & \underline{85.41} & 49.18 & \underline{67.15} &
\underline{67.53} & \underline{44.38} & \underline{87.24} & \underline{64.38} & \textbf{80.25} &
74.22 & 17.28 & 91.01 & \underline{48.72} & 68.63 \\ \hline

SE   &
76.40 & 32.69 & \underline{90.54} & 37.62 & 65.66 &
65.88 & 31.36 & \underline{91.36} & 51.71 & \underline{78.83} &
74.22 & 16.05 & \underline{95.39} & 33.97 & 68.53 \\

+CTI &
74.12 & 31.41 & 87.57 & 42.57 & 65.13 &
53.22 & 40.23 & 85.60 & 59.93 & 76.28 &
70.57 & 12.35 & 93.55 & 37.18 & 67.10 \\

+NLI &
70.39 & \underline{42.95} & 87.03 & \underline{49.51} & 66.92 &
51.57 & \underline{48.52} & 86.83 & \underline{62.67} & 77.01 &
69.01 & \underline{25.93} & 94.24 & 39.74 & \underline{70.53} \\

+MCH & 
\underline{78.47} & 30.13 & 89.73 & 42.57 & \textbf{\underline{69.66}} &
\underline{68.60} & 45.56 & 89.30 & 56.16 & 77.65 &
\underline{80.21} & 12.35 & 94.01 & \underline{42.31} & 70.10 \\ \hline

EigV & 
54.66 & 14.10 & 85.40 & 27.72 & 63.03 &
19.24 & 23.08 & 85.19 & \underline{41.44} & 72.25 &
41.15 & 6.17 & 93.08 & 26.92 & 66.35 \\

+CTI &
71.22 & 24.36 & 87.30 & 37.95 & 65.06 &
47.93 & 47.34 & 88.89 & 60.27 & 74.79 &
68.23 & 16.05 & 93.55 & \underline{39.74} & 65.23 \\

+NLI &
69.98 & \underline{38.46} & 88.65 & 41.91 & 67.29 &
48.64 & 50.89 & 87.24 & \underline{60.27} & 74.48 &
65.88 & \underline{32.10} & 95.16 & 39.10 & 68.40 \\

+MCH & 
\underline{77.85} & 28.85 & \underline{92.16} & 36.63 & \underline{68.45} &
\underline{68.81} & 45.56 & \underline{90.95} & 53.43 & \underline{75.81} &
\underline{83.07} & 12.35 & \underline{96.31} & 37.18 & \underline{70.39} \\ \hline

ECC & 
53.00 & 13.46 & 81.62 & 26.07 & 62.87 &
18.31 & 14.79 & 78.60 & 35.62 & 71.72 &
40.62 & 4.93 & 92.16 & 23.08 & 66.28 \\

+CTI & 
72.05 & 29.48 & 86.48 & 39.60 & 66.76 &
47.13 & 50.29 & 82.71 & 60.95 & 76.22 &
67.45 & 18.52 & 94.24 & 37.82 & 68.36 \\

+NLI & 
70.18 & 39.74 & 87.29 & 43.23 & 67.59 &
48.35 & 50.29 & 84.36 & 62.32 & 75.77 &
65.88 & 29.63 & 95.62 & 36.53 & 69.91 \\

+MCH & 
79.08 & 32.05 & 90.81 & 37.29 & \underline{68.80} &
68.74 & 43.19 & 90.12 & 53.08 & \underline{77.67} &
81.77 & 11.11 & 96.08 & 36.53 & \textbf{72.71} \\ \hline





\end{tabular}
\caption{Percentage of samples passing the axioms before and after calibration for Contriver with Mistral-v0.3. The results show that as the number of samples passing the axioms increases, the AUROC also improves. bold values indicate the best performance for each dataset, while underlined values represent the best performance achieved by a UE method and its calibrated variants in terms of axiomatic satisfaction.}
\label{tab:axioms_sample_percentage_mistral}
\end{table*}

\begin{figure*}[h]
  \centering
  \includegraphics[width=0.98\textwidth]{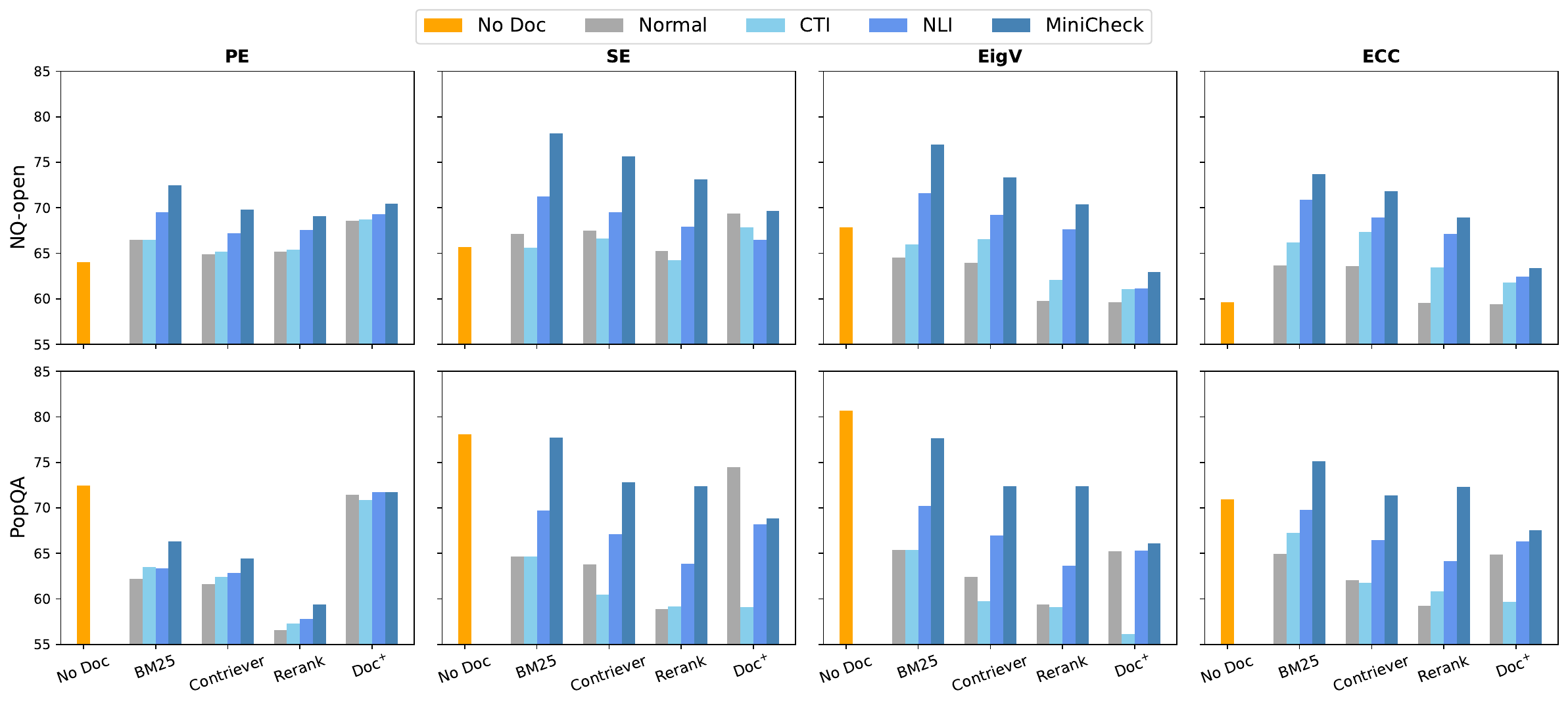}
  \caption{
  Comparison of AUROC between the no-RAG and calibrated RAG settings for Llama2-chat for NQ-open and \textsc{PopQA} datasets. AUROC improves significantly, either surpassing the no-RAG setting or reducing the gap between them.}
  \label{fig:auroc_cal_nq_popqa_llama2}
\end{figure*}

\begin{figure*}[h!]
  \centering
  \includegraphics[width=0.98\textwidth]{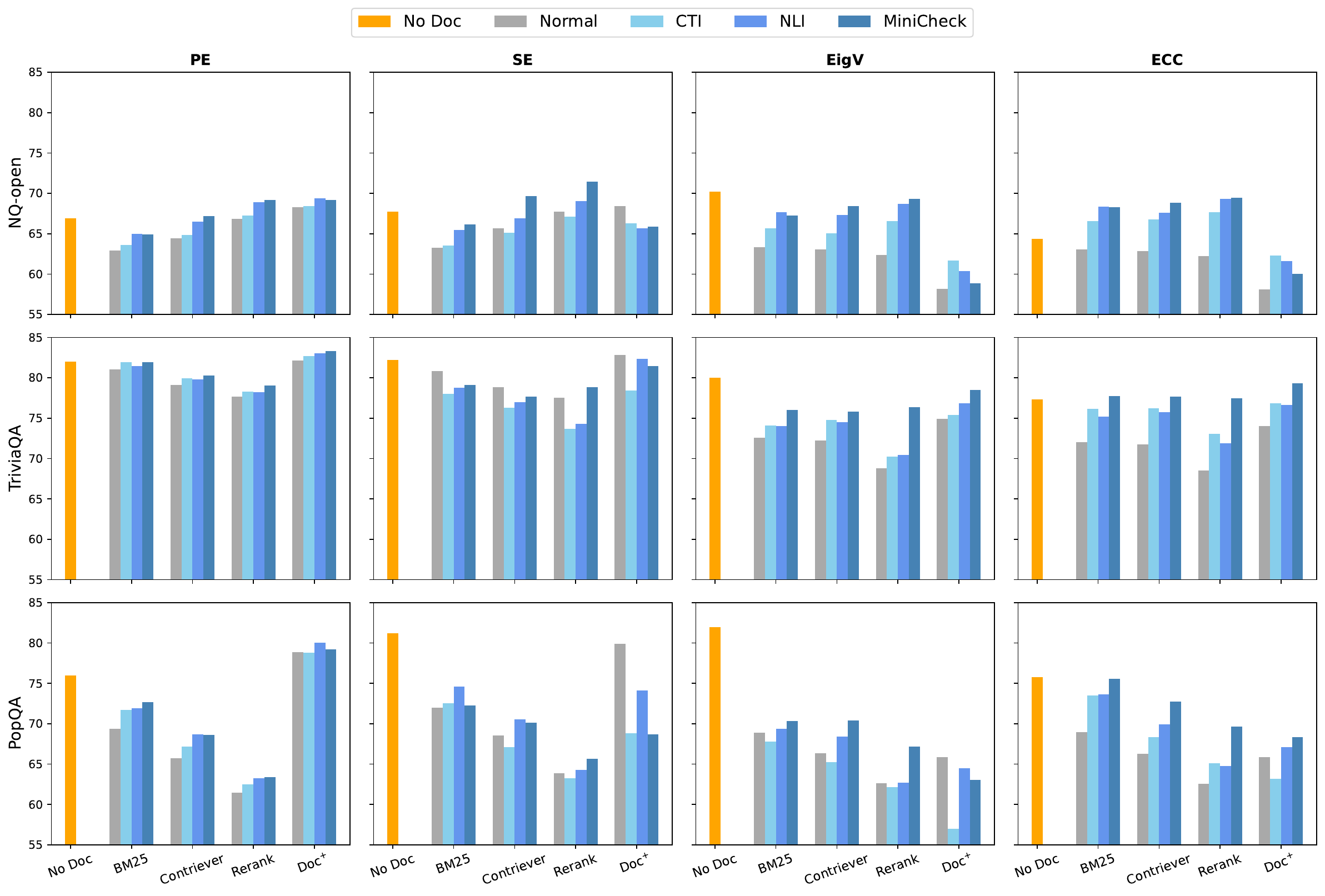}
  \shrink
  \caption{Comparison of AUROC between the no-RAG and calibrated RAG settings for Mistral-v0.3. AUROC improves significantly, either surpassing the no-RAG setting or reducing the gap between them.}
  \label{fig:auroc_cal_all_mis}
\end{figure*}

\begin{table*}[h!]
\centering
\setlength{\tabcolsep}{0.0pt}
\tiny
\begin{tabular}{c|ccc|ccc|ccc}
\hline
\textbf{UE} &
\multicolumn{3}{c}{\textbf{NQ-open}} & \multicolumn{3}{c}{\textbf{TriviaQA}} & \multicolumn{3}{c}{\textbf{PopQA}} \\ \hline %
& \textbf{BM25} & \textbf{Contriever} & \textbf{$\text{Doc}^{+}$} & 
\textbf{BM25} & \textbf{Contriever} & \textbf{$\text{Doc}^{+}$} & 
\textbf{BM25} & \textbf{Contriever} & \textbf{$\text{Doc}^{+}$} \\
\hline\hline

\multicolumn{6}{l}{\textbf{Axiom 1:} Positively Consistent $\downarrow$} \\ \hline PE &\colorbox{green!10}{55.357 $\rightarrow$ 55.357} &\colorbox{green!10}{60.194 $\rightarrow$ 61.489} &\colorbox{magenta!10}{61.793 $\rightarrow$ 61.598} &\colorbox{green!10}{41.454 $\rightarrow$ 43.124} &\colorbox{green!10}{45.532 $\rightarrow$ 46.002} &\colorbox{green!10}{45.853 $\rightarrow$ 46.401} &\colorbox{green!10}{62.369 $\rightarrow$ 68.293} &\colorbox{green!10}{66.189 $\rightarrow$ 69.628} &\colorbox{green!10}{79.511 $\rightarrow$ 81.040} \\SE &\colorbox{green!10}{66.964 $\rightarrow$ 70.982} &\colorbox{magenta!10}{77.346 $\rightarrow$ 77.023} &\colorbox{green!10}{79.337 $\rightarrow$ 83.041} &\colorbox{green!10}{47.446 $\rightarrow$ 55.894} &\colorbox{green!10}{50.141 $\rightarrow$ 56.538} &\colorbox{green!10}{52.191 $\rightarrow$ 58.059} &\colorbox{green!10}{69.338 $\rightarrow$ 79.443} &\colorbox{green!10}{71.920 $\rightarrow$ 78.510} &\colorbox{magenta!10}{88.073 $\rightarrow$ 86.544} \\PE+M &\colorbox{green!10}{57.589 $\rightarrow$ 57.589} &\colorbox{green!10}{61.165 $\rightarrow$ 63.430} &\colorbox{green!10}{62.378 $\rightarrow$ 62.378} &\colorbox{green!10}{44.008 $\rightarrow$ 44.499} &\colorbox{green!10}{43.744 $\rightarrow$ 45.720} &\colorbox{green!10}{46.870 $\rightarrow$ 48.044} &\colorbox{green!10}{62.021 $\rightarrow$ 68.293} &\colorbox{green!10}{69.628 $\rightarrow$ 73.639} &\colorbox{green!10}{81.040 $\rightarrow$ 82.875} \\SE+M &\colorbox{green!10}{64.286 $\rightarrow$ 68.304} &\colorbox{green!10}{76.375 $\rightarrow$ 76.375} &\colorbox{green!10}{72.904 $\rightarrow$ 73.099} &\colorbox{green!10}{47.348 $\rightarrow$ 55.599} &\colorbox{green!10}{48.354 $\rightarrow$ 57.008} &\colorbox{green!10}{52.504 $\rightarrow$ 58.059} &\colorbox{green!10}{69.686 $\rightarrow$ 80.488} &\colorbox{green!10}{73.639 $\rightarrow$ 79.656} &\colorbox{magenta!10}{89.602 $\rightarrow$ 88.073} \\EigV &\colorbox{green!10}{58.036 $\rightarrow$ 71.875} &\colorbox{green!10}{65.372 $\rightarrow$ 77.346} &\colorbox{green!10}{69.981 $\rightarrow$ 84.795} &\colorbox{green!10}{35.265 $\rightarrow$ 64.047} &\colorbox{green!10}{37.159 $\rightarrow$ 66.886} &\colorbox{green!10}{38.498 $\rightarrow$ 57.199} &\colorbox{green!10}{53.659 $\rightarrow$ 83.275} &\colorbox{green!10}{55.301 $\rightarrow$ 80.802} &\colorbox{green!10}{81.346 $\rightarrow$ 91.743} \\ECC &\colorbox{green!10}{54.464 $\rightarrow$ 72.321} &\colorbox{green!10}{61.489 $\rightarrow$ 75.728} &\colorbox{green!10}{68.226 $\rightarrow$ 84.016} &\colorbox{green!10}{32.122 $\rightarrow$ 62.279} &\colorbox{green!10}{34.243 $\rightarrow$ 65.475} &\colorbox{green!10}{34.977 $\rightarrow$ 55.634} &\colorbox{green!10}{50.523 $\rightarrow$ 80.836} &\colorbox{green!10}{52.436 $\rightarrow$ 78.797} &\colorbox{green!10}{77.064 $\rightarrow$ 88.379} \\Deg &\colorbox{green!10}{55.804 $\rightarrow$ 57.589} &\colorbox{green!10}{64.401 $\rightarrow$ 65.049} &\colorbox{green!10}{70.565 $\rightarrow$ 70.565} &\colorbox{green!10}{34.283 $\rightarrow$ 35.069} &\colorbox{green!10}{36.595 $\rightarrow$ 37.065} &\colorbox{green!10}{37.950 $\rightarrow$ 38.419} &\colorbox{green!10}{54.007 $\rightarrow$ 55.401} &\colorbox{green!10}{55.014 $\rightarrow$ 55.874} &\colorbox{green!10}{81.346 $\rightarrow$ 81.346} \\\hline\multicolumn{6}{l}{\textbf{Axiom 2:} Negatively Consistent $\uparrow$} \\ \hline PE &\colorbox{magenta!10}{46.237 $\rightarrow$ 44.086} &\colorbox{magenta!10}{47.853 $\rightarrow$ 44.172} &\colorbox{magenta!10}{47.059 $\rightarrow$ 44.118} &\colorbox{magenta!10}{52.299 $\rightarrow$ 50.575} &\colorbox{green!10}{43.781 $\rightarrow$ 43.781} &\colorbox{magenta!10}{56.477 $\rightarrow$ 56.218} &\colorbox{magenta!10}{49.275 $\rightarrow$ 44.928} &\colorbox{magenta!10}{42.466 $\rightarrow$ 39.726} &\colorbox{green!10}{57.143 $\rightarrow$ 57.143} \\SE &\colorbox{magenta!10}{34.409 $\rightarrow$ 31.183} &\colorbox{magenta!10}{33.742 $\rightarrow$ 26.380} &\colorbox{magenta!10}{31.618 $\rightarrow$ 28.676} &\colorbox{magenta!10}{42.529 $\rightarrow$ 41.379} &\colorbox{green!10}{35.821 $\rightarrow$ 39.303} &\colorbox{green!10}{45.078 $\rightarrow$ 50.777} &\colorbox{magenta!10}{34.783 $\rightarrow$ 21.739} &\colorbox{magenta!10}{31.507 $\rightarrow$ 26.027} &\colorbox{magenta!10}{42.857 $\rightarrow$ 28.571} \\PE+M &\colorbox{green!10}{39.247 $\rightarrow$ 39.247} &\colorbox{magenta!10}{42.945 $\rightarrow$ 38.037} &\colorbox{magenta!10}{47.794 $\rightarrow$ 46.324} &\colorbox{magenta!10}{49.425 $\rightarrow$ 47.701} &\colorbox{green!10}{41.791 $\rightarrow$ 41.791} &\colorbox{green!10}{52.332 $\rightarrow$ 53.886} &\colorbox{magenta!10}{44.928 $\rightarrow$ 43.478} &\colorbox{magenta!10}{43.836 $\rightarrow$ 42.466} &\colorbox{green!10}{57.143 $\rightarrow$ 57.143} \\SE+M &\colorbox{magenta!10}{31.720 $\rightarrow$ 30.108} &\colorbox{magenta!10}{31.288 $\rightarrow$ 26.380} &\colorbox{green!10}{35.294 $\rightarrow$ 36.029} &\colorbox{magenta!10}{41.379 $\rightarrow$ 37.356} &\colorbox{green!10}{35.323 $\rightarrow$ 38.308} &\colorbox{green!10}{44.301 $\rightarrow$ 51.295} &\colorbox{magenta!10}{33.333 $\rightarrow$ 17.391} &\colorbox{magenta!10}{30.137 $\rightarrow$ 27.397} &\colorbox{magenta!10}{42.857 $\rightarrow$ 28.571} \\EigV &\colorbox{green!10}{19.355 $\rightarrow$ 31.720} &\colorbox{green!10}{12.883 $\rightarrow$ 20.245} &\colorbox{green!10}{5.147 $\rightarrow$ 26.471} &\colorbox{green!10}{29.885 $\rightarrow$ 34.483} &\colorbox{green!10}{24.378 $\rightarrow$ 30.846} &\colorbox{green!10}{37.047 $\rightarrow$ 50.259} &\colorbox{magenta!10}{15.942 $\rightarrow$ 13.043} &\colorbox{green!10}{6.849 $\rightarrow$ 19.178} &\colorbox{magenta!10}{42.857 $\rightarrow$ 28.571} \\ECC &\colorbox{green!10}{14.516 $\rightarrow$ 37.097} &\colorbox{green!10}{9.816 $\rightarrow$ 23.313} &\colorbox{green!10}{5.882 $\rightarrow$ 30.147} &\colorbox{green!10}{19.540 $\rightarrow$ 35.057} &\colorbox{green!10}{14.428 $\rightarrow$ 31.841} &\colorbox{green!10}{21.503 $\rightarrow$ 58.031} &\colorbox{green!10}{10.145 $\rightarrow$ 20.290} &\colorbox{green!10}{6.849 $\rightarrow$ 23.288} &\colorbox{green!10}{28.571 $\rightarrow$ 28.571} \\Deg &\colorbox{green!10}{20.968 $\rightarrow$ 20.968} &\colorbox{magenta!10}{17.178 $\rightarrow$ 15.951} &\colorbox{green!10}{5.147 $\rightarrow$ 6.618} &\colorbox{green!10}{29.885 $\rightarrow$ 31.034} &\colorbox{green!10}{24.378 $\rightarrow$ 24.876} &\colorbox{green!10}{36.788 $\rightarrow$ 42.487} &\colorbox{green!10}{13.043 $\rightarrow$ 13.043} &\colorbox{green!10}{12.329 $\rightarrow$ 12.329} &\colorbox{green!10}{57.143 $\rightarrow$ 57.143} \\\hline\multicolumn{6}{l}{\textbf{Axiom 3:} Positively Changed $\downarrow$} \\ \hline PE &\colorbox{magenta!10}{82.215 $\rightarrow$ 81.544} &\colorbox{magenta!10}{77.346 $\rightarrow$ 76.430} &\colorbox{magenta!10}{82.557 $\rightarrow$ 81.541} &\colorbox{magenta!10}{73.402 $\rightarrow$ 72.634} &\colorbox{magenta!10}{70.256 $\rightarrow$ 69.231} &\colorbox{magenta!10}{74.870 $\rightarrow$ 73.830} &\colorbox{green!10}{82.331 $\rightarrow$ 83.083} &\colorbox{green!10}{87.572 $\rightarrow$ 87.954} &\colorbox{green!10}{84.314 $\rightarrow$ 84.540} \\SE &\colorbox{green!10}{93.289 $\rightarrow$ 93.289} &\colorbox{magenta!10}{91.533 $\rightarrow$ 89.703} &\colorbox{magenta!10}{93.057 $\rightarrow$ 91.194} &\colorbox{magenta!10}{86.445 $\rightarrow$ 83.632} &\colorbox{magenta!10}{84.615 $\rightarrow$ 79.744} &\colorbox{magenta!10}{88.042 $\rightarrow$ 83.882} &\colorbox{magenta!10}{93.233 $\rightarrow$ 90.226} &\colorbox{magenta!10}{94.073 $\rightarrow$ 91.205} &\colorbox{magenta!10}{92.534 $\rightarrow$ 88.235} \\PE+M &\colorbox{magenta!10}{81.544 $\rightarrow$ 79.866} &\colorbox{magenta!10}{77.574 $\rightarrow$ 76.888} &\colorbox{magenta!10}{80.271 $\rightarrow$ 78.493} &\colorbox{green!10}{76.982 $\rightarrow$ 77.749} &\colorbox{magenta!10}{73.590 $\rightarrow$ 72.308} &\colorbox{magenta!10}{80.069 $\rightarrow$ 77.296} &\colorbox{magenta!10}{88.346 $\rightarrow$ 87.594} &\colorbox{magenta!10}{90.822 $\rightarrow$ 90.440} &\colorbox{green!10}{84.389 $\rightarrow$ 84.691} \\SE+M &\colorbox{magenta!10}{90.604 $\rightarrow$ 88.591} &\colorbox{magenta!10}{88.787 $\rightarrow$ 85.812} &\colorbox{magenta!10}{88.654 $\rightarrow$ 86.198} &\colorbox{magenta!10}{86.957 $\rightarrow$ 84.143} &\colorbox{magenta!10}{84.359 $\rightarrow$ 80.000} &\colorbox{magenta!10}{88.562 $\rightarrow$ 84.922} &\colorbox{magenta!10}{93.609 $\rightarrow$ 92.857} &\colorbox{magenta!10}{94.455 $\rightarrow$ 92.543} &\colorbox{magenta!10}{93.439 $\rightarrow$ 89.668} \\EigV &\colorbox{green!10}{90.604 $\rightarrow$ 91.611} &\colorbox{green!10}{88.558 $\rightarrow$ 90.389} &\colorbox{green!10}{89.077 $\rightarrow$ 91.025} &\colorbox{magenta!10}{86.189 $\rightarrow$ 85.166} &\colorbox{green!10}{86.154 $\rightarrow$ 86.154} &\colorbox{green!10}{83.709 $\rightarrow$ 86.308} &\colorbox{magenta!10}{91.353 $\rightarrow$ 90.977} &\colorbox{green!10}{92.925 $\rightarrow$ 93.499} &\colorbox{green!10}{86.652 $\rightarrow$ 89.367} \\ECC &\colorbox{green!10}{82.886 $\rightarrow$ 87.919} &\colorbox{green!10}{83.066 $\rightarrow$ 87.185} &\colorbox{green!10}{82.557 $\rightarrow$ 86.622} &\colorbox{green!10}{79.028 $\rightarrow$ 80.563} &\colorbox{green!10}{73.590 $\rightarrow$ 77.692} &\colorbox{green!10}{75.390 $\rightarrow$ 80.243} &\colorbox{green!10}{86.466 $\rightarrow$ 89.850} &\colorbox{green!10}{87.380 $\rightarrow$ 90.822} &\colorbox{green!10}{82.730 $\rightarrow$ 86.652} \\Deg &\colorbox{green!10}{90.604 $\rightarrow$ 90.940} &\colorbox{green!10}{87.414 $\rightarrow$ 87.643} &\colorbox{green!10}{89.331 $\rightarrow$ 89.670} &\colorbox{green!10}{85.934 $\rightarrow$ 86.189} &\colorbox{magenta!10}{86.410 $\rightarrow$ 85.128} &\colorbox{green!10}{85.442 $\rightarrow$ 85.789} &\colorbox{magenta!10}{91.353 $\rightarrow$ 90.977} &\colorbox{green!10}{92.543 $\rightarrow$ 92.543} &\colorbox{magenta!10}{86.576 $\rightarrow$ 86.501} \\\hline\multicolumn{6}{l}{\textbf{Axiom 4:} Negatively Changed $\uparrow$} \\ \hline PE &\colorbox{green!10}{51.136 $\rightarrow$ 52.273} &\colorbox{green!10}{51.163 $\rightarrow$ 53.876} &\colorbox{green!10}{49.231 $\rightarrow$ 50.769} &\colorbox{magenta!10}{66.944 $\rightarrow$ 66.389} &\colorbox{green!10}{66.879 $\rightarrow$ 68.471} &\colorbox{magenta!10}{66.372 $\rightarrow$ 63.717} &\colorbox{magenta!10}{42.045 $\rightarrow$ 39.773} &\colorbox{green!10}{38.168 $\rightarrow$ 38.168} &\colorbox{green!10}{27.586 $\rightarrow$ 27.586} \\SE &\colorbox{green!10}{36.080 $\rightarrow$ 40.625} &\colorbox{green!10}{36.047 $\rightarrow$ 40.310} &\colorbox{magenta!10}{44.615 $\rightarrow$ 40.000} &\colorbox{green!10}{55.556 $\rightarrow$ 58.889} &\colorbox{green!10}{54.777 $\rightarrow$ 56.688} &\colorbox{green!10}{52.212 $\rightarrow$ 57.522} &\colorbox{green!10}{31.818 $\rightarrow$ 36.364} &\colorbox{magenta!10}{29.008 $\rightarrow$ 26.718} &\colorbox{magenta!10}{25.287 $\rightarrow$ 22.989} \\PE+M &\colorbox{green!10}{47.727 $\rightarrow$ 49.716} &\colorbox{green!10}{50.388 $\rightarrow$ 53.876} &\colorbox{green!10}{50.769 $\rightarrow$ 56.923} &\colorbox{green!10}{63.333 $\rightarrow$ 64.722} &\colorbox{magenta!10}{66.242 $\rightarrow$ 65.287} &\colorbox{green!10}{64.602 $\rightarrow$ 65.487} &\colorbox{green!10}{38.636 $\rightarrow$ 38.636} &\colorbox{magenta!10}{32.061 $\rightarrow$ 31.298} &\colorbox{green!10}{26.437 $\rightarrow$ 27.586} \\SE+M &\colorbox{green!10}{38.636 $\rightarrow$ 41.193} &\colorbox{green!10}{40.698 $\rightarrow$ 42.636} &\colorbox{green!10}{41.538 $\rightarrow$ 49.231} &\colorbox{green!10}{55.278 $\rightarrow$ 57.500} &\colorbox{green!10}{53.503 $\rightarrow$ 56.369} &\colorbox{green!10}{53.097 $\rightarrow$ 55.752} &\colorbox{green!10}{31.250 $\rightarrow$ 33.523} &\colorbox{magenta!10}{28.244 $\rightarrow$ 24.427} &\colorbox{magenta!10}{24.138 $\rightarrow$ 20.690} \\EigV &\colorbox{green!10}{24.432 $\rightarrow$ 34.091} &\colorbox{green!10}{24.419 $\rightarrow$ 35.271} &\colorbox{green!10}{16.923 $\rightarrow$ 18.462} &\colorbox{green!10}{38.333 $\rightarrow$ 51.944} &\colorbox{green!10}{39.172 $\rightarrow$ 48.408} &\colorbox{green!10}{38.938 $\rightarrow$ 51.327} &\colorbox{green!10}{21.591 $\rightarrow$ 35.795} &\colorbox{green!10}{20.611 $\rightarrow$ 29.771} &\colorbox{green!10}{8.046 $\rightarrow$ 12.644} \\ECC &\colorbox{green!10}{19.602 $\rightarrow$ 39.205} &\colorbox{green!10}{18.992 $\rightarrow$ 37.984} &\colorbox{green!10}{16.923 $\rightarrow$ 30.769} &\colorbox{green!10}{30.556 $\rightarrow$ 57.500} &\colorbox{green!10}{30.892 $\rightarrow$ 53.185} &\colorbox{green!10}{26.549 $\rightarrow$ 60.177} &\colorbox{green!10}{18.182 $\rightarrow$ 44.318} &\colorbox{green!10}{18.321 $\rightarrow$ 34.351} &\colorbox{green!10}{8.046 $\rightarrow$ 19.540} \\Deg &\colorbox{green!10}{25.284 $\rightarrow$ 26.989} &\colorbox{green!10}{24.806 $\rightarrow$ 27.132} &\colorbox{green!10}{20.000 $\rightarrow$ 23.077} &\colorbox{green!10}{42.500 $\rightarrow$ 45.278} &\colorbox{green!10}{42.357 $\rightarrow$ 44.904} &\colorbox{green!10}{42.478 $\rightarrow$ 44.248} &\colorbox{green!10}{22.727 $\rightarrow$ 23.864} &\colorbox{green!10}{19.084 $\rightarrow$ 19.084} &\colorbox{green!10}{11.494 $\rightarrow$ 11.494} \\\hline

\end{tabular}
\caption{
Changes in the percentage of samples that satisfy the axioms before and after calibration for Llama2-chat. The relation function \(\mathcal{R}\) is implemented using CTI.}
\label{tab:axioms_rel_kldiv}
\end{table*}

\begin{table*}[h]
\centering
\setlength{\tabcolsep}{0.0pt}
\tiny
\begin{tabular}{c|ccc|ccc|ccc}
\hline
\textbf{UE} &
\multicolumn{3}{c}{\textbf{NQ-open}} & \multicolumn{3}{c}{\textbf{TriviaQA}} & \multicolumn{3}{c}{\textbf{PopQA}} \\ \hline %
& \textbf{BM25} & \textbf{Contriever} & \textbf{$\text{Doc}^{+}$} & 
\textbf{BM25} & \textbf{Contriever} & \textbf{$\text{Doc}^{+}$} & 
\textbf{BM25} & \textbf{Contriever} & \textbf{$\text{Doc}^{+}$} \\
\hline\hline

\multicolumn{6}{l}{\textbf{Axiom 1:} Positively Consistent $\downarrow$} \\ \hline PE &\colorbox{green!10}{55.357 $\rightarrow$ 60.714} &\colorbox{green!10}{60.194 $\rightarrow$ 66.019} &\colorbox{green!10}{61.793 $\rightarrow$ 66.667} &\colorbox{green!10}{41.454 $\rightarrow$ 44.695} &\colorbox{green!10}{45.532 $\rightarrow$ 48.730} &\colorbox{green!10}{45.853 $\rightarrow$ 47.966} &\colorbox{green!10}{62.369 $\rightarrow$ 63.415} &\colorbox{green!10}{66.189 $\rightarrow$ 68.481} &\colorbox{green!10}{79.511 $\rightarrow$ 80.122} \\SE &\colorbox{green!10}{66.964 $\rightarrow$ 73.661} &\colorbox{green!10}{77.346 $\rightarrow$ 80.259} &\colorbox{green!10}{79.337 $\rightarrow$ 79.922} &\colorbox{green!10}{47.446 $\rightarrow$ 58.743} &\colorbox{green!10}{50.141 $\rightarrow$ 58.231} &\colorbox{green!10}{52.191 $\rightarrow$ 57.433} &\colorbox{green!10}{69.338 $\rightarrow$ 73.868} &\colorbox{green!10}{71.920 $\rightarrow$ 74.785} &\colorbox{magenta!10}{88.073 $\rightarrow$ 81.040} \\PE+M &\colorbox{green!10}{57.589 $\rightarrow$ 62.054} &\colorbox{green!10}{61.165 $\rightarrow$ 67.961} &\colorbox{green!10}{62.378 $\rightarrow$ 67.057} &\colorbox{green!10}{44.008 $\rightarrow$ 46.660} &\colorbox{green!10}{43.744 $\rightarrow$ 48.260} &\colorbox{green!10}{46.870 $\rightarrow$ 49.687} &\colorbox{green!10}{62.021 $\rightarrow$ 63.763} &\colorbox{green!10}{69.628 $\rightarrow$ 72.206} &\colorbox{green!10}{81.040 $\rightarrow$ 82.263} \\SE+M &\colorbox{green!10}{64.286 $\rightarrow$ 70.089} &\colorbox{green!10}{76.375 $\rightarrow$ 77.023} &\colorbox{green!10}{72.904 $\rightarrow$ 74.269} &\colorbox{green!10}{47.348 $\rightarrow$ 58.350} &\colorbox{green!10}{48.354 $\rightarrow$ 58.325} &\colorbox{green!10}{52.504 $\rightarrow$ 57.825} &\colorbox{green!10}{69.686 $\rightarrow$ 76.655} &\colorbox{green!10}{73.639 $\rightarrow$ 75.072} &\colorbox{magenta!10}{89.602 $\rightarrow$ 84.709} \\EigV &\colorbox{green!10}{58.036 $\rightarrow$ 71.429} &\colorbox{green!10}{65.372 $\rightarrow$ 82.201} &\colorbox{green!10}{69.981 $\rightarrow$ 86.160} &\colorbox{green!10}{35.265 $\rightarrow$ 59.136} &\colorbox{green!10}{37.159 $\rightarrow$ 60.960} &\colorbox{green!10}{38.498 $\rightarrow$ 59.077} &\colorbox{green!10}{53.659 $\rightarrow$ 73.868} &\colorbox{green!10}{55.301 $\rightarrow$ 74.785} &\colorbox{green!10}{81.346 $\rightarrow$ 85.933} \\ECC &\colorbox{green!10}{54.464 $\rightarrow$ 70.982} &\colorbox{green!10}{61.489 $\rightarrow$ 78.641} &\colorbox{green!10}{68.226 $\rightarrow$ 84.795} &\colorbox{green!10}{32.122 $\rightarrow$ 55.403} &\colorbox{green!10}{34.243 $\rightarrow$ 58.043} &\colorbox{green!10}{34.977 $\rightarrow$ 55.399} &\colorbox{green!10}{50.523 $\rightarrow$ 72.474} &\colorbox{green!10}{52.436 $\rightarrow$ 71.347} &\colorbox{green!10}{77.064 $\rightarrow$ 84.404} \\Deg &\colorbox{green!10}{55.804 $\rightarrow$ 56.250} &\colorbox{green!10}{64.401 $\rightarrow$ 64.401} &\colorbox{green!10}{70.565 $\rightarrow$ 70.955} &\colorbox{green!10}{34.283 $\rightarrow$ 35.069} &\colorbox{green!10}{36.595 $\rightarrow$ 36.877} &\colorbox{green!10}{37.950 $\rightarrow$ 38.185} &\colorbox{green!10}{54.007 $\rightarrow$ 55.052} &\colorbox{green!10}{55.014 $\rightarrow$ 57.307} &\colorbox{magenta!10}{81.346 $\rightarrow$ 81.040} \\\hline\multicolumn{6}{l}{\textbf{Axiom 2:} Negatively Consistent $\uparrow$} \\ \hline PE &\colorbox{green!10}{46.237 $\rightarrow$ 47.312} &\colorbox{magenta!10}{47.853 $\rightarrow$ 47.239} &\colorbox{magenta!10}{47.059 $\rightarrow$ 46.324} &\colorbox{green!10}{52.299 $\rightarrow$ 54.598} &\colorbox{green!10}{43.781 $\rightarrow$ 45.274} &\colorbox{green!10}{56.477 $\rightarrow$ 58.549} &\colorbox{green!10}{49.275 $\rightarrow$ 49.275} &\colorbox{magenta!10}{42.466 $\rightarrow$ 41.096} &\colorbox{green!10}{57.143 $\rightarrow$ 57.143} \\SE &\colorbox{green!10}{34.409 $\rightarrow$ 38.172} &\colorbox{green!10}{33.742 $\rightarrow$ 38.037} &\colorbox{green!10}{31.618 $\rightarrow$ 31.618} &\colorbox{green!10}{42.529 $\rightarrow$ 48.276} &\colorbox{green!10}{35.821 $\rightarrow$ 43.284} &\colorbox{green!10}{45.078 $\rightarrow$ 56.218} &\colorbox{magenta!10}{34.783 $\rightarrow$ 33.333} &\colorbox{green!10}{31.507 $\rightarrow$ 32.877} &\colorbox{green!10}{42.857 $\rightarrow$ 57.143} \\PE+M &\colorbox{green!10}{39.247 $\rightarrow$ 43.011} &\colorbox{magenta!10}{42.945 $\rightarrow$ 41.718} &\colorbox{green!10}{47.794 $\rightarrow$ 50.735} &\colorbox{green!10}{49.425 $\rightarrow$ 54.023} &\colorbox{magenta!10}{41.791 $\rightarrow$ 41.294} &\colorbox{green!10}{52.332 $\rightarrow$ 56.218} &\colorbox{green!10}{44.928 $\rightarrow$ 46.377} &\colorbox{green!10}{43.836 $\rightarrow$ 45.205} &\colorbox{green!10}{57.143 $\rightarrow$ 57.143} \\SE+M &\colorbox{green!10}{31.720 $\rightarrow$ 38.710} &\colorbox{green!10}{31.288 $\rightarrow$ 36.196} &\colorbox{magenta!10}{35.294 $\rightarrow$ 33.824} &\colorbox{green!10}{41.379 $\rightarrow$ 45.977} &\colorbox{green!10}{35.323 $\rightarrow$ 39.801} &\colorbox{green!10}{44.301 $\rightarrow$ 55.440} &\colorbox{magenta!10}{33.333 $\rightarrow$ 30.435} &\colorbox{green!10}{30.137 $\rightarrow$ 32.877} &\colorbox{green!10}{42.857 $\rightarrow$ 42.857} \\EigV &\colorbox{green!10}{19.355 $\rightarrow$ 35.484} &\colorbox{green!10}{12.883 $\rightarrow$ 26.380} &\colorbox{green!10}{5.147 $\rightarrow$ 20.588} &\colorbox{green!10}{29.885 $\rightarrow$ 46.552} &\colorbox{green!10}{24.378 $\rightarrow$ 38.806} &\colorbox{green!10}{37.047 $\rightarrow$ 58.290} &\colorbox{green!10}{15.942 $\rightarrow$ 26.087} &\colorbox{green!10}{6.849 $\rightarrow$ 32.877} &\colorbox{green!10}{42.857 $\rightarrow$ 42.857} \\ECC &\colorbox{green!10}{14.516 $\rightarrow$ 43.011} &\colorbox{green!10}{9.816 $\rightarrow$ 32.515} &\colorbox{green!10}{5.882 $\rightarrow$ 25.000} &\colorbox{green!10}{19.540 $\rightarrow$ 55.172} &\colorbox{green!10}{14.428 $\rightarrow$ 42.786} &\colorbox{green!10}{21.503 $\rightarrow$ 76.425} &\colorbox{green!10}{10.145 $\rightarrow$ 34.783} &\colorbox{green!10}{6.849 $\rightarrow$ 32.877} &\colorbox{green!10}{28.571 $\rightarrow$ 57.143} \\Deg &\colorbox{green!10}{20.968 $\rightarrow$ 23.656} &\colorbox{green!10}{17.178 $\rightarrow$ 17.791} &\colorbox{green!10}{5.147 $\rightarrow$ 8.824} &\colorbox{green!10}{29.885 $\rightarrow$ 34.483} &\colorbox{green!10}{24.378 $\rightarrow$ 26.866} &\colorbox{green!10}{36.788 $\rightarrow$ 45.596} &\colorbox{green!10}{13.043 $\rightarrow$ 13.043} &\colorbox{green!10}{12.329 $\rightarrow$ 13.699} &\colorbox{green!10}{57.143 $\rightarrow$ 57.143} \\\hline\multicolumn{6}{l}{\textbf{Axiom 3:} Positively Changed $\downarrow$} \\ \hline PE &\colorbox{green!10}{82.215 $\rightarrow$ 84.228} &\colorbox{green!10}{77.346 $\rightarrow$ 77.574} &\colorbox{magenta!10}{82.557 $\rightarrow$ 81.964} &\colorbox{green!10}{73.402 $\rightarrow$ 73.913} &\colorbox{green!10}{70.256 $\rightarrow$ 70.513} &\colorbox{magenta!10}{74.870 $\rightarrow$ 74.003} &\colorbox{green!10}{82.331 $\rightarrow$ 84.586} &\colorbox{green!10}{87.572 $\rightarrow$ 88.145} &\colorbox{green!10}{84.314 $\rightarrow$ 84.615} \\SE &\colorbox{magenta!10}{93.289 $\rightarrow$ 88.591} &\colorbox{magenta!10}{91.533 $\rightarrow$ 86.270} &\colorbox{magenta!10}{93.057 $\rightarrow$ 86.113} &\colorbox{magenta!10}{86.445 $\rightarrow$ 84.910} &\colorbox{magenta!10}{84.615 $\rightarrow$ 80.513} &\colorbox{magenta!10}{88.042 $\rightarrow$ 84.749} &\colorbox{magenta!10}{93.233 $\rightarrow$ 91.729} &\colorbox{magenta!10}{94.073 $\rightarrow$ 92.161} &\colorbox{magenta!10}{92.534 $\rightarrow$ 87.029} \\PE+M &\colorbox{green!10}{81.544 $\rightarrow$ 85.235} &\colorbox{green!10}{77.574 $\rightarrow$ 79.863} &\colorbox{green!10}{80.271 $\rightarrow$ 80.610} &\colorbox{green!10}{76.982 $\rightarrow$ 77.238} &\colorbox{magenta!10}{73.590 $\rightarrow$ 72.821} &\colorbox{magenta!10}{80.069 $\rightarrow$ 78.683} &\colorbox{green!10}{88.346 $\rightarrow$ 88.346} &\colorbox{magenta!10}{90.822 $\rightarrow$ 90.057} &\colorbox{green!10}{84.389 $\rightarrow$ 85.143} \\SE+M &\colorbox{magenta!10}{90.604 $\rightarrow$ 87.248} &\colorbox{magenta!10}{88.787 $\rightarrow$ 84.211} &\colorbox{magenta!10}{88.654 $\rightarrow$ 82.557} &\colorbox{magenta!10}{86.957 $\rightarrow$ 84.655} &\colorbox{magenta!10}{84.359 $\rightarrow$ 81.026} &\colorbox{magenta!10}{88.562 $\rightarrow$ 86.482} &\colorbox{magenta!10}{93.609 $\rightarrow$ 92.105} &\colorbox{magenta!10}{94.455 $\rightarrow$ 93.690} &\colorbox{magenta!10}{93.439 $\rightarrow$ 88.235} \\EigV &\colorbox{green!10}{90.604 $\rightarrow$ 92.617} &\colorbox{green!10}{88.558 $\rightarrow$ 91.533} &\colorbox{green!10}{89.077 $\rightarrow$ 90.517} &\colorbox{green!10}{86.189 $\rightarrow$ 87.724} &\colorbox{green!10}{86.154 $\rightarrow$ 87.436} &\colorbox{green!10}{83.709 $\rightarrow$ 86.655} &\colorbox{green!10}{91.353 $\rightarrow$ 93.609} &\colorbox{green!10}{92.925 $\rightarrow$ 95.602} &\colorbox{green!10}{86.652 $\rightarrow$ 90.875} \\ECC &\colorbox{green!10}{82.886 $\rightarrow$ 88.255} &\colorbox{green!10}{83.066 $\rightarrow$ 87.185} &\colorbox{green!10}{82.557 $\rightarrow$ 86.791} &\colorbox{green!10}{79.028 $\rightarrow$ 84.655} &\colorbox{green!10}{73.590 $\rightarrow$ 77.436} &\colorbox{green!10}{75.390 $\rightarrow$ 78.163} &\colorbox{green!10}{86.466 $\rightarrow$ 91.353} &\colorbox{green!10}{87.380 $\rightarrow$ 92.161} &\colorbox{green!10}{82.730 $\rightarrow$ 88.537} \\Deg &\colorbox{magenta!10}{90.604 $\rightarrow$ 89.933} &\colorbox{magenta!10}{87.414 $\rightarrow$ 86.270} &\colorbox{magenta!10}{89.331 $\rightarrow$ 89.162} &\colorbox{green!10}{85.934 $\rightarrow$ 86.189} &\colorbox{magenta!10}{86.410 $\rightarrow$ 86.154} &\colorbox{magenta!10}{85.442 $\rightarrow$ 84.749} &\colorbox{green!10}{91.353 $\rightarrow$ 91.353} &\colorbox{magenta!10}{92.543 $\rightarrow$ 92.352} &\colorbox{green!10}{86.576 $\rightarrow$ 86.652} \\\hline\multicolumn{6}{l}{\textbf{Axiom 4:} Negatively Changed $\uparrow$} \\ \hline PE &\colorbox{green!10}{51.136 $\rightarrow$ 56.250} &\colorbox{green!10}{51.163 $\rightarrow$ 55.426} &\colorbox{green!10}{49.231 $\rightarrow$ 58.462} &\colorbox{green!10}{66.944 $\rightarrow$ 68.611} &\colorbox{green!10}{66.879 $\rightarrow$ 68.790} &\colorbox{green!10}{66.372 $\rightarrow$ 69.027} &\colorbox{green!10}{42.045 $\rightarrow$ 42.614} &\colorbox{green!10}{38.168 $\rightarrow$ 41.221} &\colorbox{green!10}{27.586 $\rightarrow$ 31.034} \\SE &\colorbox{green!10}{36.080 $\rightarrow$ 49.432} &\colorbox{green!10}{36.047 $\rightarrow$ 50.000} &\colorbox{green!10}{44.615 $\rightarrow$ 52.308} &\colorbox{green!10}{55.556 $\rightarrow$ 65.000} &\colorbox{green!10}{54.777 $\rightarrow$ 64.013} &\colorbox{green!10}{52.212 $\rightarrow$ 64.602} &\colorbox{green!10}{31.818 $\rightarrow$ 42.045} &\colorbox{green!10}{29.008 $\rightarrow$ 41.985} &\colorbox{green!10}{25.287 $\rightarrow$ 31.034} \\PE+M &\colorbox{green!10}{47.727 $\rightarrow$ 52.273} &\colorbox{green!10}{50.388 $\rightarrow$ 56.977} &\colorbox{green!10}{50.769 $\rightarrow$ 55.385} &\colorbox{green!10}{63.333 $\rightarrow$ 66.667} &\colorbox{green!10}{66.242 $\rightarrow$ 67.834} &\colorbox{green!10}{64.602 $\rightarrow$ 67.257} &\colorbox{green!10}{38.636 $\rightarrow$ 38.636} &\colorbox{green!10}{32.061 $\rightarrow$ 35.115} &\colorbox{green!10}{26.437 $\rightarrow$ 29.885} \\SE+M &\colorbox{green!10}{38.636 $\rightarrow$ 51.136} &\colorbox{green!10}{40.698 $\rightarrow$ 53.488} &\colorbox{green!10}{41.538 $\rightarrow$ 56.923} &\colorbox{green!10}{55.278 $\rightarrow$ 62.778} &\colorbox{green!10}{53.503 $\rightarrow$ 64.013} &\colorbox{green!10}{53.097 $\rightarrow$ 61.947} &\colorbox{green!10}{31.250 $\rightarrow$ 38.068} &\colorbox{green!10}{28.244 $\rightarrow$ 39.695} &\colorbox{green!10}{24.138 $\rightarrow$ 29.885} \\EigV &\colorbox{green!10}{24.432 $\rightarrow$ 35.795} &\colorbox{green!10}{24.419 $\rightarrow$ 36.047} &\colorbox{green!10}{16.923 $\rightarrow$ 33.846} &\colorbox{green!10}{38.333 $\rightarrow$ 57.500} &\colorbox{green!10}{39.172 $\rightarrow$ 53.185} &\colorbox{green!10}{38.938 $\rightarrow$ 53.982} &\colorbox{green!10}{21.591 $\rightarrow$ 36.932} &\colorbox{green!10}{20.611 $\rightarrow$ 38.931} &\colorbox{green!10}{8.046 $\rightarrow$ 18.391} \\ECC &\colorbox{green!10}{19.602 $\rightarrow$ 43.466} &\colorbox{green!10}{18.992 $\rightarrow$ 42.636} &\colorbox{green!10}{16.923 $\rightarrow$ 36.923} &\colorbox{green!10}{30.556 $\rightarrow$ 65.556} &\colorbox{green!10}{30.892 $\rightarrow$ 59.873} &\colorbox{green!10}{26.549 $\rightarrow$ 65.487} &\colorbox{green!10}{18.182 $\rightarrow$ 46.591} &\colorbox{green!10}{18.321 $\rightarrow$ 42.748} &\colorbox{green!10}{8.046 $\rightarrow$ 25.287} \\Deg &\colorbox{green!10}{25.284 $\rightarrow$ 29.545} &\colorbox{green!10}{24.806 $\rightarrow$ 27.907} &\colorbox{green!10}{20.000 $\rightarrow$ 24.615} &\colorbox{green!10}{42.500 $\rightarrow$ 47.222} &\colorbox{green!10}{42.357 $\rightarrow$ 48.089} &\colorbox{green!10}{42.478 $\rightarrow$ 48.673} &\colorbox{green!10}{22.727 $\rightarrow$ 24.432} &\colorbox{green!10}{19.084 $\rightarrow$ 21.374} &\colorbox{green!10}{11.494 $\rightarrow$ 16.092} \\\hline

\end{tabular}
\caption{Changes in the percentage of samples that satisfy the axioms before and after calibration for Llama2-chat. The relation function \(\mathcal{R}\) is implemented using NLI.}
\label{tab:axioms_rel_nli}
\end{table*}

\begin{table*}[h]
\centering
\setlength{\tabcolsep}{0.0pt}
\tiny
\begin{tabular}{c|ccc|ccc|ccc}
\hline
\textbf{UE} &
\multicolumn{3}{c}{\textbf{NQ-open}} & \multicolumn{3}{c}{\textbf{TriviaQA}} & \multicolumn{3}{c}{\textbf{PopQA}} \\ \hline %
& \textbf{BM25} & \textbf{Contriever} & \textbf{$\text{Doc}^{+}$} & 
\textbf{BM25} & \textbf{Contriever} & \textbf{$\text{Doc}^{+}$} & 
\textbf{BM25} & \textbf{Contriever} & \textbf{$\text{Doc}^{+}$} \\
\hline\hline

\multicolumn{6}{l}{\textbf{Axiom 1:} Positively Consistent $\downarrow$} \\ \hline PE &\colorbox{green!10}{55.357 $\rightarrow$ 70.089} &\colorbox{green!10}{60.194 $\rightarrow$ 75.081} &\colorbox{green!10}{61.793 $\rightarrow$ 75.634} &\colorbox{green!10}{41.454 $\rightarrow$ 48.134} &\colorbox{green!10}{45.532 $\rightarrow$ 51.364} &\colorbox{green!10}{45.853 $\rightarrow$ 51.174} &\colorbox{green!10}{62.369 $\rightarrow$ 68.293} &\colorbox{green!10}{66.189 $\rightarrow$ 69.341} &\colorbox{green!10}{79.511 $\rightarrow$ 81.040} \\SE &\colorbox{green!10}{66.964 $\rightarrow$ 78.571} &\colorbox{green!10}{77.346 $\rightarrow$ 86.408} &\colorbox{green!10}{79.337 $\rightarrow$ 90.058} &\colorbox{green!10}{47.446 $\rightarrow$ 72.299} &\colorbox{green!10}{50.141 $\rightarrow$ 73.283} &\colorbox{green!10}{52.191 $\rightarrow$ 73.083} &\colorbox{green!10}{69.338 $\rightarrow$ 81.533} &\colorbox{green!10}{71.920 $\rightarrow$ 77.937} &\colorbox{magenta!10}{88.073 $\rightarrow$ 86.544} \\PE+M &\colorbox{green!10}{57.589 $\rightarrow$ 68.750} &\colorbox{green!10}{61.165 $\rightarrow$ 78.317} &\colorbox{green!10}{62.378 $\rightarrow$ 76.023} &\colorbox{green!10}{44.008 $\rightarrow$ 48.723} &\colorbox{green!10}{43.744 $\rightarrow$ 51.646} &\colorbox{green!10}{46.870 $\rightarrow$ 52.034} &\colorbox{green!10}{62.021 $\rightarrow$ 70.035} &\colorbox{green!10}{69.628 $\rightarrow$ 71.347} &\colorbox{green!10}{81.040 $\rightarrow$ 81.957} \\SE+M &\colorbox{green!10}{64.286 $\rightarrow$ 75.446} &\colorbox{green!10}{76.375 $\rightarrow$ 86.731} &\colorbox{green!10}{72.904 $\rightarrow$ 87.329} &\colorbox{green!10}{47.348 $\rightarrow$ 71.709} &\colorbox{green!10}{48.354 $\rightarrow$ 72.907} &\colorbox{green!10}{52.504 $\rightarrow$ 72.457} &\colorbox{green!10}{69.686 $\rightarrow$ 82.230} &\colorbox{green!10}{73.639 $\rightarrow$ 78.223} &\colorbox{magenta!10}{89.602 $\rightarrow$ 87.156} \\EigV &\colorbox{green!10}{58.036 $\rightarrow$ 77.679} &\colorbox{green!10}{65.372 $\rightarrow$ 87.702} &\colorbox{green!10}{69.981 $\rightarrow$ 93.177} &\colorbox{green!10}{35.265 $\rightarrow$ 70.432} &\colorbox{green!10}{37.159 $\rightarrow$ 74.882} &\colorbox{green!10}{38.498 $\rightarrow$ 74.257} &\colorbox{green!10}{53.659 $\rightarrow$ 82.927} &\colorbox{green!10}{55.301 $\rightarrow$ 83.095} &\colorbox{green!10}{81.346 $\rightarrow$ 95.413} \\ECC &\colorbox{green!10}{54.464 $\rightarrow$ 76.786} &\colorbox{green!10}{61.489 $\rightarrow$ 86.084} &\colorbox{green!10}{68.226 $\rightarrow$ 92.398} &\colorbox{green!10}{32.122 $\rightarrow$ 66.306} &\colorbox{green!10}{34.243 $\rightarrow$ 72.437} &\colorbox{green!10}{34.977 $\rightarrow$ 70.736} &\colorbox{green!10}{50.523 $\rightarrow$ 80.488} &\colorbox{green!10}{52.436 $\rightarrow$ 79.370} &\colorbox{green!10}{77.064 $\rightarrow$ 93.272} \\Deg &\colorbox{green!10}{55.804 $\rightarrow$ 57.143} &\colorbox{green!10}{64.401 $\rightarrow$ 65.372} &\colorbox{magenta!10}{70.565 $\rightarrow$ 70.175} &\colorbox{green!10}{34.283 $\rightarrow$ 35.560} &\colorbox{green!10}{36.595 $\rightarrow$ 37.535} &\colorbox{green!10}{37.950 $\rightarrow$ 39.202} &\colorbox{green!10}{54.007 $\rightarrow$ 54.704} &\colorbox{green!10}{55.014 $\rightarrow$ 55.587} &\colorbox{green!10}{81.346 $\rightarrow$ 81.346} \\\hline\multicolumn{6}{l}{\textbf{Axiom 2:} Negatively Consistent $\uparrow$} \\ \hline PE &\colorbox{green!10}{46.237 $\rightarrow$ 52.151} &\colorbox{magenta!10}{47.853 $\rightarrow$ 39.877} &\colorbox{magenta!10}{47.059 $\rightarrow$ 38.971} &\colorbox{green!10}{52.299 $\rightarrow$ 57.471} &\colorbox{green!10}{43.781 $\rightarrow$ 49.254} &\colorbox{green!10}{56.477 $\rightarrow$ 60.622} &\colorbox{magenta!10}{49.275 $\rightarrow$ 46.377} &\colorbox{magenta!10}{42.466 $\rightarrow$ 39.726} &\colorbox{green!10}{57.143 $\rightarrow$ 57.143} \\SE &\colorbox{green!10}{34.409 $\rightarrow$ 40.323} &\colorbox{green!10}{33.742 $\rightarrow$ 34.969} &\colorbox{magenta!10}{31.618 $\rightarrow$ 27.941} &\colorbox{green!10}{42.529 $\rightarrow$ 54.023} &\colorbox{green!10}{35.821 $\rightarrow$ 49.751} &\colorbox{green!10}{45.078 $\rightarrow$ 56.995} &\colorbox{green!10}{34.783 $\rightarrow$ 37.681} &\colorbox{green!10}{31.507 $\rightarrow$ 31.507} &\colorbox{green!10}{42.857 $\rightarrow$ 71.429} \\PE+M &\colorbox{green!10}{39.247 $\rightarrow$ 46.774} &\colorbox{magenta!10}{42.945 $\rightarrow$ 34.969} &\colorbox{magenta!10}{47.794 $\rightarrow$ 40.441} &\colorbox{green!10}{49.425 $\rightarrow$ 58.621} &\colorbox{green!10}{41.791 $\rightarrow$ 48.756} &\colorbox{green!10}{52.332 $\rightarrow$ 59.585} &\colorbox{magenta!10}{44.928 $\rightarrow$ 43.478} &\colorbox{magenta!10}{43.836 $\rightarrow$ 39.726} &\colorbox{green!10}{57.143 $\rightarrow$ 71.429} \\SE+M &\colorbox{green!10}{31.720 $\rightarrow$ 44.086} &\colorbox{green!10}{31.288 $\rightarrow$ 34.969} &\colorbox{magenta!10}{35.294 $\rightarrow$ 30.882} &\colorbox{green!10}{41.379 $\rightarrow$ 51.149} &\colorbox{green!10}{35.323 $\rightarrow$ 48.756} &\colorbox{green!10}{44.301 $\rightarrow$ 57.254} &\colorbox{green!10}{33.333 $\rightarrow$ 34.783} &\colorbox{green!10}{30.137 $\rightarrow$ 31.507} &\colorbox{green!10}{42.857 $\rightarrow$ 71.429} \\EigV &\colorbox{green!10}{19.355 $\rightarrow$ 31.183} &\colorbox{green!10}{12.883 $\rightarrow$ 24.540} &\colorbox{green!10}{5.147 $\rightarrow$ 18.382} &\colorbox{green!10}{29.885 $\rightarrow$ 44.253} &\colorbox{green!10}{24.378 $\rightarrow$ 40.299} &\colorbox{green!10}{37.047 $\rightarrow$ 52.073} &\colorbox{green!10}{15.942 $\rightarrow$ 21.739} &\colorbox{green!10}{6.849 $\rightarrow$ 24.658} &\colorbox{green!10}{42.857 $\rightarrow$ 42.857} \\ECC &\colorbox{green!10}{14.516 $\rightarrow$ 36.022} &\colorbox{green!10}{9.816 $\rightarrow$ 26.994} &\colorbox{green!10}{5.882 $\rightarrow$ 21.324} &\colorbox{green!10}{19.540 $\rightarrow$ 49.425} &\colorbox{green!10}{14.428 $\rightarrow$ 41.294} &\colorbox{green!10}{21.503 $\rightarrow$ 65.026} &\colorbox{green!10}{10.145 $\rightarrow$ 31.884} &\colorbox{green!10}{6.849 $\rightarrow$ 21.918} &\colorbox{green!10}{28.571 $\rightarrow$ 57.143} \\Deg &\colorbox{green!10}{20.968 $\rightarrow$ 26.882} &\colorbox{green!10}{17.178 $\rightarrow$ 18.405} &\colorbox{green!10}{5.147 $\rightarrow$ 9.559} &\colorbox{green!10}{29.885 $\rightarrow$ 37.931} &\colorbox{green!10}{24.378 $\rightarrow$ 30.846} &\colorbox{green!10}{36.788 $\rightarrow$ 50.000} &\colorbox{green!10}{13.043 $\rightarrow$ 15.942} &\colorbox{green!10}{12.329 $\rightarrow$ 13.699} &\colorbox{green!10}{57.143 $\rightarrow$ 57.143} \\\hline\multicolumn{6}{l}{\textbf{Axiom 3:} Positively Changed $\downarrow$} \\ \hline PE &\colorbox{green!10}{82.215 $\rightarrow$ 91.946} &\colorbox{green!10}{77.346 $\rightarrow$ 83.982} &\colorbox{green!10}{82.557 $\rightarrow$ 84.589} &\colorbox{green!10}{73.402 $\rightarrow$ 76.726} &\colorbox{green!10}{70.256 $\rightarrow$ 74.103} &\colorbox{magenta!10}{74.870 $\rightarrow$ 74.697} &\colorbox{green!10}{82.331 $\rightarrow$ 86.842} &\colorbox{green!10}{87.572 $\rightarrow$ 89.484} &\colorbox{green!10}{84.314 $\rightarrow$ 84.691} \\SE &\colorbox{green!10}{93.289 $\rightarrow$ 93.960} &\colorbox{magenta!10}{91.533 $\rightarrow$ 90.847} &\colorbox{magenta!10}{93.057 $\rightarrow$ 89.331} &\colorbox{green!10}{86.445 $\rightarrow$ 88.491} &\colorbox{magenta!10}{84.615 $\rightarrow$ 82.821} &\colorbox{magenta!10}{88.042 $\rightarrow$ 84.575} &\colorbox{green!10}{93.233 $\rightarrow$ 94.361} &\colorbox{green!10}{94.073 $\rightarrow$ 94.073} &\colorbox{magenta!10}{92.534 $\rightarrow$ 89.216} \\PE+M &\colorbox{green!10}{81.544 $\rightarrow$ 91.275} &\colorbox{green!10}{77.574 $\rightarrow$ 84.439} &\colorbox{green!10}{80.271 $\rightarrow$ 83.065} &\colorbox{green!10}{76.982 $\rightarrow$ 79.028} &\colorbox{green!10}{73.590 $\rightarrow$ 75.385} &\colorbox{magenta!10}{80.069 $\rightarrow$ 79.029} &\colorbox{green!10}{88.346 $\rightarrow$ 89.850} &\colorbox{green!10}{90.822 $\rightarrow$ 91.396} &\colorbox{green!10}{84.389 $\rightarrow$ 85.143} \\SE+M &\colorbox{green!10}{90.604 $\rightarrow$ 93.289} &\colorbox{green!10}{88.787 $\rightarrow$ 90.847} &\colorbox{magenta!10}{88.654 $\rightarrow$ 87.214} &\colorbox{green!10}{86.957 $\rightarrow$ 89.258} &\colorbox{magenta!10}{84.359 $\rightarrow$ 83.846} &\colorbox{magenta!10}{88.562 $\rightarrow$ 85.442} &\colorbox{green!10}{93.609 $\rightarrow$ 95.113} &\colorbox{magenta!10}{94.455 $\rightarrow$ 94.264} &\colorbox{magenta!10}{93.439 $\rightarrow$ 90.121} \\EigV &\colorbox{green!10}{90.604 $\rightarrow$ 94.295} &\colorbox{green!10}{88.558 $\rightarrow$ 93.822} &\colorbox{green!10}{89.077 $\rightarrow$ 91.871} &\colorbox{green!10}{86.189 $\rightarrow$ 90.026} &\colorbox{green!10}{86.154 $\rightarrow$ 90.000} &\colorbox{green!10}{83.709 $\rightarrow$ 89.081} &\colorbox{green!10}{91.353 $\rightarrow$ 96.241} &\colorbox{green!10}{92.925 $\rightarrow$ 96.750} &\colorbox{green!10}{86.652 $\rightarrow$ 94.646} \\ECC &\colorbox{green!10}{82.886 $\rightarrow$ 89.933} &\colorbox{green!10}{83.066 $\rightarrow$ 89.931} &\colorbox{green!10}{82.557 $\rightarrow$ 88.400} &\colorbox{green!10}{79.028 $\rightarrow$ 87.724} &\colorbox{green!10}{73.590 $\rightarrow$ 82.308} &\colorbox{green!10}{75.390 $\rightarrow$ 84.749} &\colorbox{green!10}{86.466 $\rightarrow$ 93.985} &\colorbox{green!10}{87.380 $\rightarrow$ 94.837} &\colorbox{green!10}{82.730 $\rightarrow$ 92.911} \\Deg &\colorbox{magenta!10}{90.604 $\rightarrow$ 89.933} &\colorbox{magenta!10}{87.414 $\rightarrow$ 86.499} &\colorbox{green!10}{89.331 $\rightarrow$ 89.331} &\colorbox{green!10}{85.934 $\rightarrow$ 86.701} &\colorbox{magenta!10}{86.410 $\rightarrow$ 85.128} &\colorbox{magenta!10}{85.442 $\rightarrow$ 83.882} &\colorbox{green!10}{91.353 $\rightarrow$ 91.729} &\colorbox{magenta!10}{92.543 $\rightarrow$ 92.352} &\colorbox{magenta!10}{86.576 $\rightarrow$ 86.275} \\\hline\multicolumn{6}{l}{\textbf{Axiom 4:} Negatively Changed $\uparrow$} \\ \hline PE &\colorbox{green!10}{51.136 $\rightarrow$ 55.682} &\colorbox{green!10}{51.163 $\rightarrow$ 51.550} &\colorbox{green!10}{49.231 $\rightarrow$ 58.462} &\colorbox{green!10}{66.944 $\rightarrow$ 69.722} &\colorbox{green!10}{66.879 $\rightarrow$ 69.745} &\colorbox{magenta!10}{66.372 $\rightarrow$ 63.717} &\colorbox{magenta!10}{42.045 $\rightarrow$ 40.909} &\colorbox{green!10}{38.168 $\rightarrow$ 39.695} &\colorbox{green!10}{27.586 $\rightarrow$ 32.184} \\SE &\colorbox{green!10}{36.080 $\rightarrow$ 48.011} &\colorbox{green!10}{36.047 $\rightarrow$ 46.512} &\colorbox{green!10}{44.615 $\rightarrow$ 53.846} &\colorbox{green!10}{55.556 $\rightarrow$ 65.833} &\colorbox{green!10}{54.777 $\rightarrow$ 67.197} &\colorbox{green!10}{52.212 $\rightarrow$ 63.717} &\colorbox{green!10}{31.818 $\rightarrow$ 46.023} &\colorbox{green!10}{29.008 $\rightarrow$ 41.221} &\colorbox{green!10}{25.287 $\rightarrow$ 36.782} \\PE+M &\colorbox{green!10}{47.727 $\rightarrow$ 51.420} &\colorbox{green!10}{50.388 $\rightarrow$ 50.775} &\colorbox{green!10}{50.769 $\rightarrow$ 61.538} &\colorbox{green!10}{63.333 $\rightarrow$ 69.167} &\colorbox{green!10}{66.242 $\rightarrow$ 67.834} &\colorbox{green!10}{64.602 $\rightarrow$ 65.487} &\colorbox{magenta!10}{38.636 $\rightarrow$ 38.068} &\colorbox{green!10}{32.061 $\rightarrow$ 36.641} &\colorbox{green!10}{26.437 $\rightarrow$ 31.034} \\SE+M &\colorbox{green!10}{38.636 $\rightarrow$ 50.568} &\colorbox{green!10}{40.698 $\rightarrow$ 48.450} &\colorbox{green!10}{41.538 $\rightarrow$ 56.923} &\colorbox{green!10}{55.278 $\rightarrow$ 62.778} &\colorbox{green!10}{53.503 $\rightarrow$ 66.879} &\colorbox{green!10}{53.097 $\rightarrow$ 64.602} &\colorbox{green!10}{31.250 $\rightarrow$ 43.182} &\colorbox{green!10}{28.244 $\rightarrow$ 39.695} &\colorbox{green!10}{24.138 $\rightarrow$ 34.483} \\EigV &\colorbox{green!10}{24.432 $\rightarrow$ 35.227} &\colorbox{green!10}{24.419 $\rightarrow$ 34.496} &\colorbox{green!10}{16.923 $\rightarrow$ 32.308} &\colorbox{green!10}{38.333 $\rightarrow$ 55.278} &\colorbox{green!10}{39.172 $\rightarrow$ 55.414} &\colorbox{green!10}{38.938 $\rightarrow$ 53.982} &\colorbox{green!10}{21.591 $\rightarrow$ 34.091} &\colorbox{green!10}{20.611 $\rightarrow$ 32.824} &\colorbox{green!10}{8.046 $\rightarrow$ 17.241} \\ECC &\colorbox{green!10}{19.602 $\rightarrow$ 42.330} &\colorbox{green!10}{18.992 $\rightarrow$ 39.535} &\colorbox{green!10}{16.923 $\rightarrow$ 33.846} &\colorbox{green!10}{30.556 $\rightarrow$ 61.389} &\colorbox{green!10}{30.892 $\rightarrow$ 58.917} &\colorbox{green!10}{26.549 $\rightarrow$ 61.062} &\colorbox{green!10}{18.182 $\rightarrow$ 41.477} &\colorbox{green!10}{18.321 $\rightarrow$ 35.878} &\colorbox{green!10}{8.046 $\rightarrow$ 21.839} \\Deg &\colorbox{green!10}{25.284 $\rightarrow$ 29.830} &\colorbox{green!10}{24.806 $\rightarrow$ 28.295} &\colorbox{green!10}{20.000 $\rightarrow$ 26.154} &\colorbox{green!10}{42.500 $\rightarrow$ 49.167} &\colorbox{green!10}{42.357 $\rightarrow$ 49.363} &\colorbox{green!10}{42.478 $\rightarrow$ 50.442} &\colorbox{green!10}{22.727 $\rightarrow$ 26.136} &\colorbox{green!10}{19.084 $\rightarrow$ 22.137} &\colorbox{green!10}{11.494 $\rightarrow$ 19.540} \\\hline

\end{tabular}
\caption{Changes in the percentage of samples that satisfy the axioms before and after calibration for Llama2-chat. The relation function \(\mathcal{R}\) is implemented using MiniCheck.}
\label{tab:axioms_rel_minicheck}
\end{table*}

\input{sections/ap_related_work}

\input{sections/ap_experimental_setup}


%% file: sections/ap_related_work.tex
\section{Related Work}~\label{sec:related_work}
\shrink

\noindent
\textbf{Uncertainty Estimation (UE)} seeks to quantify the confidence of LLMs in their predictions~\cite{Hou024Decomposing, Zhao24Knowing}. UE methods are commonly divided into two groups: black-box and white-box approaches. Black-box methods rely solely on the LLM’s outputs without accessing internal layers or generation logits. In addition to the semantic similarity-based methods discussed in Section~\ref{sec:background}, other black-box techniques exist. For example, verbalization methods prompt the model to explicitly report its confidence (e.g., “How confident are you that the answer is correct?”). \citet{Xiong24Can} highlight that two key factors influence the quality of verbalized confidence: (i) the prompting strategy, which includes techniques such as vanilla, Chain-of-Thought (CoT), self-probing, multi-step, and Top-K prompting, and (ii) the sampling strategy, employing methods like self-random sampling, prompting-based elicitation, and misleading prompts to generate multiple responses. Additionally, Epi-M~\cite{Zhou24Relying} incorporates epistemic markers into the input prompt to facilitate well-calibrated confidence scores.

White-box approaches, by contrast, leverage access to next-token prediction probabilities for uncertainty calculation. Beyond the methods covered in Section~\ref{sec:background}, several techniques have been proposed. For instance, P(True)~\cite{Kadavath22PE} measures the probability that a model assigns to the correctness of a given response by appending a sentence such as \verb|Is the possible answer: (A) True (B) False.| \verb|The possible answer is:| so that the probability of generating “True” or “False” serves as the measure. Similarly, P(IK)~\cite{Kadavath22PE} estimates the likelihood that the model “knows” the correct answer, that is, the probability of generating the correct response when sampling at unit temperature. Furthermore, LARS~\cite{Yaldiz24LARS} introduces a learning-based approach by training a scoring model on token probabilities to enhance uncertainty prediction.

\noindent
\textbf{Axiomatic Evaluation.}
Axiomatic thinking refers to a problem‐solving approach guided by a set of axioms closely aligned with conventional scientific methodologies~\cite{Amigo20Axiomaticthinking}. More generally, this approach seeks solutions that satisfy all predefined axioms, that is, the desirable properties a solution should possess.

Axiomatic thinking has been successfully applied to the study of Information Retrieval (IR), thereby contributing both to the theoretical understanding and the practical enhancement of existing retrieval models. The objective of Axiomatic IR is to establish formal constraints, or axioms, that delineate the essential properties an effective ranking model must satisfy~\cite{Volske21Towards}. In this context, \citet{Fang04formal} formally defined six fundamental constraints derived from empirical observations of common characteristics in traditional retrieval functions. These constraints correspond to intuitive retrieval heuristics, such as term frequency weighting, term discrimination weighting, and document length normalization.
Building on this foundation, \citet{exploration05Fang} proposed an axiomatic framework for the development of retrieval models. Their framework comprises an inductive scheme for function definitions, which provides a common basis for the analytical comparison of different retrieval functions, as well as a set of formalized retrieval constraints adapted from \cite{Fang04formal}. These axioms have been further examined in subsequent studies. For example, \citet{Chen24Axiomatic} employed causal interventions to identify specific attention heads that encode a robust term frequency signal, thereby aligning with one of the original axioms. 

Beyond IR, axiomatic approaches have been extended to other domains. For instance, \citet{Rosset23AxiomaticPreference} defined axioms representing the qualities that humans value in long-form answers, including usefulness, relevance, groundedness, truthfulness, and thoroughness. They generated training data corresponding to these principles and subsequently used it to train a preference model.

%% file: sections/ap_experimental_setup.tex
\section{Experimental Setup}\label{sec:appendix_es}

\noindent
\textbf{Datasets.}
We conduct our experiments on three open-book Question Answering (QA) datasets: Natural Questions (NQ)~\cite{Kwiatkowski19NQ}, TriviaQA~\cite{Joshi17TriviaQA}, and \textsc{PopQA}~\cite{Mallen23popqa}.
The NQ dataset comprises a large-scale collection of real-world queries derived from Google search data. Each entry includes a user query and the corresponding Wikipedia page that contains the answer. The NQ-open dataset~\cite{Lee19nqopen}, a subset of NQ, differs by removing the restriction of linking answers to specific Wikipedia passages, thereby emulating a more general real-world scenario. We obtain the gold documents for each query from the corpus and dataset annotated by \cite{Cuconasu24Power}~\footnote{\href{https://huggingface.co/datasets/florin-hf/nq_open_gold}{Dataset: florin-hf/nq\_open\_gold}}, in which the gold documents are integrated with the original corpus. For evaluation, we use the test set containing 2,889 queries.
TriviaQA consists of trivia questions sourced from the web~\cite{jeong2024adaptive}. To ensure a dataset size comparable to NQ-open, we randomly sample 3,000 queries from its development set.
\textsc{PopQA} is an open-domain QA dataset designed to evaluate factual knowledge, particularly regarding long-tail entities. Constructed from 16 diverse relationship types in Wikidata, \textsc{PopQA} is originally a closed-book dataset comprising 14,000 QA pairs without gold document annotations. Consequently, following \cite{Soudani24FTvsRAG}, we consider the summary section of the corresponding Wikipedia page as the gold document. Since \textsc{PopQA} is entirely based on Wikipedia, we employ the same corpus for retrieval. To maintain consistency with the other datasets, we randomly select 3,000 samples from the test set.
%

\noindent
\textbf{Language Models.}
In accordance with established baselines, we select two generative LLMs: Llama2-chat-7B and Mistral-7B.
For inputs that are not augmented with retrieved documents, we employ the following template:  
\noindent
\verb|"Answer the question. Question: <question>| \verb|Answer:"|
For inputs augmented with retrieved documents, we utilize this template:  
\verb|"You are given a question, and| \verb|you MUST respond with an answer (max 10| \verb|tokens) using either the provided document| \verb|or your memorized knowledge. Document:| \verb|<context> Question:<question> Answer:"|.
Although more sophisticated prompts were examined in preliminary experiments, the marginal improvement they offered relative to the simple template did not justify their use, particularly given the increased risk of model overfitting. 
Furthermore, following MARS~\cite{Bakman24mars}, we utilize the Huggingface library's "generate" function for model output generation. We designate the token "." as the "eos\_token\_id" to prevent the model from generating overly lengthy paragraphs in response to closed-book questions. We also set "num\_beams" to $1$, corresponding to greedy decoding.

\noindent
\textbf{Retrieval Models.}
We employ a suite of retrieval models to acquire relevant contexts for the RAG approach. The models utilized include BM25~\citep{BM2509Robertson}, Contriever~\citep{Unsupervised22Izacard}, and a two-stage re-ranking system. In the two-stage configuration, BM25 is applied for initial retrieval, followed by re-ranking using a pre-trained cross-encoder model, specifically, \verb|ms-marco-MiniLM-L-6-v2| from the \verb|sentence-transformers| library.
Additionally, we report results for two variations: $\text{Doc}^{+}$, in which the gold context is incorporated into the input prompt, and $\text{Doc}^{-}$, in which an irrelevant context is substituted. Although several methods exist to obtain irrelevant contexts~\citep{Zhao24Enhancing}, in our experiments, these are generated by randomly sampling a context from the corpus.

\noindent
\textbf{NLI Models.}
A NLI classifier takes a sequence pair \((x_1, x_2)\) and outputs a label \(y \in \{\text{\emph{Contradiction}}, \text{\emph{Neutral}}, \text{\emph{Entailment}}\}\) with corresponding probabilities. 
The two sequences are concatenated with a separator token \(\texttt{[SEP]}\) before input. To study ordering effects, we consider both \(x_1 \texttt{[SEP]} x_2\) and \(x_2 \texttt{[SEP]} x_1\). In the reference-free setting (Section~\ref{sec:axioms:inst}), if either order yields a contradiction, the input is labeled as such; otherwise, it is labeled as entailment. In Section~\ref{sec:calibration_function:inst}, we use the maximum entailment probability from the two orders.

\noindent
\textbf{Computational Cost.}
We conducted all experiments using one Nvidia A100 GPUs with 40 GB of memory, accumulating approximately 250 GPU hours. Due to the substantial computational demands, all results presented are based on a single run.
